\documentclass[a4paper,,11pt,twoside]{article}
\usepackage{setspace}
\usepackage[a4paper,left=2.2cm,right=2.2cm,top=2.3cm,bottom=2.6cm]{geometry}
\usepackage[english]{babel}
\usepackage{afterpage}
\usepackage{mathrsfs,hyperref,algorithm, algorithmic, yfonts}
\usepackage{harvard}
\usepackage{float}
\usepackage{comment}
\usepackage{multirow}
\usepackage{tabu}
\usepackage{subfigure}
\usepackage{placeins} 
\usepackage[]{amsmath}
\usepackage{amsthm}
\usepackage{fix-cm}
\usepackage[]{amssymb}
\usepackage[]{latexsym}
\usepackage[latin1]{inputenc}
\usepackage[right]{eurosym}
\usepackage[T1]{fontenc}
\usepackage[]{graphicx}
\usepackage[]{epsfig}
\usepackage{fancyhdr}
\usepackage{bbm}
\usepackage{bm}
\usepackage{pstricks}
\usepackage{multirow}

\makeatletter

\renewcommand{\p@enumi}{theenumi-}
\renewcommand{\@fnsymbol}[1]{\@alph{#1}}
\usepackage{tabularx}
\usepackage{rotating}
\usepackage{float}

\newcommand{\bbr}{\mathbb{R}}  

\newcommand{\ci}{\citeasnoun}

\newcommand{\fil}{\mathcal{F}}

\newcommand{\xcal}{\mathcal{X}}

\newcommand{\bcal}{\mathcal{B}}

\newcommand{\mcal}{\mathcal{M}}

\newcommand{\be}{\begin{equation}}
\newcommand{\ee}{\end{equation}}
\newcommand{\bew}{\begin{equation*}}
\newcommand{\eew}{\end{equation*}}

\newcommand{\var}{{\rm V@R}}
\newcommand{\avar}{{\rm AV@R}}
\newcommand{\revar}{{\rm RecV@R}}
\newcommand{\Lrevar}{{\rm LRecV@R}}
\newcommand{\reavar}{{\rm RecAV@R}}
\newcommand{\Lreavar}{{\rm LRecAV@R}}

\newcommand{\ba}{\begin{array}{ll}}
\newcommand{\bal}{\begin{array}{ll}}
\newcommand{\ea}{\end{array}}

\newcommand{\E}{\mathbb{E}}

\newcommand{\probp}{\mathbb{P}}
\newcommand{\probq}{\mathbb{Q}}
\newcommand{\R}{\mathbb{R}}
\newcommand{\N}{\mathbb{N}}

\newcommand{\cF}{{\mathcal{F}}}

\newcommand{\cB}{{\mathcal{B}}}

\newcommand{\cA}{\mathcal{A}}

\newcommand{\cX}{{\mathcal{X}}}

\newcommand{\VaR}{\mathop {\rm V@R}\nolimits}
\newcommand{\reVaR}{\mathop {\rm RecV@R}\nolimits}
\newcommand{\ES}{\mathop {\rm AV@R}\nolimits}

\newcommand{\rerho}{\mathop {\rm Rec\rho}\nolimits}

\newcommand{\RecAdj}{\mathop {\rm RecAdj}\nolimits}
\newcommand{\AggRecAdj}{\mathop {\rm AggRecAdj}\nolimits}

\newtheorem{thm}{Theorem}
\newtheorem{defi}[thm]{Definition}
\newtheorem{prop}[thm]{Proposition}

\newtheorem{rem}[thm]{Remark}
\newtheorem{remark}[thm]{Remark}
\newtheorem{ex}[thm]{Example}

\renewcommand{\p@enumi}{theenumi-}
\renewcommand{\@fnsymbol}[1]{\@alph{#1}}

\begin{document}


\title{
Capital Requirements and Claims Recovery:\\
A New Perspective on Solvency Regulation}

\author{
Cosimo Munari\footnote{Center for Finance and Insurance and Swiss Finance Institute, University of Zurich, Plattenstrasse 14, 8032 Zurich, Switzerland.
e-mail:   \href{mailto:cosimo.munari@bf.uzh.ch}{\tt cosimo.munari@bf.uzh.ch}.} \\[1.0ex] \textit{University of Zurich}
\and
Stefan Weber\footnote{House of Insurance \& Institute of Actuarial and Financial Mathematics, Leibniz Universit\"at Hannover, Welfengarten 1, 30167 Hannover, Germany.
e-mail:   \href{mailto:stefan.weber@insurance.uni-hannover.de}{\tt stefan.weber@insurance.uni-hannover.de}.} \\[1.0ex] \textit{Leibniz Universit{\"a}t Hannover} \and Lutz Wilhelmy\footnote{Group Risk Management,
Swiss Re Management Ltd, Mythenquai 50/60, 8022 Zurich, Switzerland. e-mail:  \href{mailto:lutz_wilhelmy@swissre.com}{\tt lutz\_wilhelmy@swissre.com}. The opinions expressed in this article are those of the author and do not necessarily coincide with those of his employer.} \\[1.0ex] \textit{Swiss Reinsurance Company} }
\date{\today\footnote{We would like to thank Kerstin Awiszus and Pablo Koch-Medina for helpful comments.}}

\maketitle

\begin{abstract}
Protection of creditors is a key objective of financial regulation. Where the protection needs are high, i.e., in banking and insurance, regulatory solvency requirements are an instrument to prevent that creditors incur losses on their claims. The current regulatory requirements based on Value at Risk and Average Value at Risk limit the probability of default of financial institutions, but they fail to control the size of recovery on creditors' claims in the case of default. We resolve this failure by developing a novel risk measure, Recovery Value at Risk. Our conceptual approach can flexibly be extended and allows the construction of general recovery risk measures for various risk management purposes. By design, these risk measures control recovery on creditors' claims and integrate the protection needs of creditors into the incentive structure of the management.

We provide detailed case studies and applications: We analyze how recovery risk measures react to the joint distributions of assets and liabilities on firms' balance sheets and compare the corresponding capital requirements with the current regulatory benchmarks based on Value at Risk and Average Value at Risk. We discuss how to calibrate recovery risk measures to historic regulatory standards. Finally, we show that recovery risk measures can be applied to performance-based management of business divisions of firms and that they allow for a tractable characterization of optimal tradeoffs between risk and return in the context of investment management.
\end{abstract}
\vspace{0.2cm}
\textbf{Keywords:} Risk Measures, Capital Requirements, Solvency Regulation, Recovery on Liabilities.


\section{Introduction}\label{sec:intro}

Banks and insurance companies are subject to a variety of regulatory constraints. A key objective of financial regulation is the appropriate protection of creditors, e.g., depositors, policyholders, and other counterparties. Corporate governance, reporting, and transparency are cornerstones of regulatory schemes, but equally important is capital regulation. Financial companies are required to respect solvency capital requirements that define a minimum for their current net asset value. Firms that fail to meet these requirements are subject to supervisory interventions.

The computation of solvency capital requirements is often based on some pre-specified notion of acceptable default risk. Banks and insurance companies must hold enough capital to meet their obligations in a sufficient number of future economic scenarios. Regulators typically focus on quantities such as the change of net asset value over a specific time horizon --- for example, one year --- and require that a suitable risk measure applied to such quantities is below the current level of available capital. The risk measure implicitly defines a notion of acceptable default risk. Different risk measures are applied in practice.

The standard example are solvency capital requirements defined in terms of Value at Risk. In this case, a company is adequately capitalized if its default probability is lower than a given threshold. The upcoming regulatory framework for the internationally active insurance groups uses a Value at Risk at the level $0.5\%$. In Europe, insurance companies and groups are subject to the same requirement under the Solvency II regime. Value at Risk has been strongly criticized due to its tail blindness and its lack of convexity -- not encouraging diversification.

An alternative to Value at Risk is the coherent risk measure Average Value at Risk, also called Conditional or Tail Value at Risk or Expected Shortfall. The market risk standards in Basel III, the international regulatory framework for banks, and the Swiss Solvency Test, the Swiss regulatory framework for insurance companies, are both based on Average Value at Risk with levels $2.5\%$ and $1\%$, respectively. In this case, a company or portfolio is deemed adequately capitalized, if it generates profits on average conditional on its tail distribution below the chosen level. Average Value at Risk is sensitive to the tail, and, being convex, it does not penalize diversification. It is also a tractable ingredient to optimization problems in the context of asset-liability-management and provides an instrument for decentralized risk management, e.g., limit systems within firms.

Despite all its merits, Average Value at Risk fails --- just as Value at Risk --- at one central task: It cannot control recovery in the case of default, i.e., the probabilities that creditors recover prespecified fractions of claims! This goal is, of course, important from a regulatory point of view. Recovering, say, 80\% instead of 0\% in the case of default makes a big difference to creditors such as depositors or policyholders.
This failure is apparent when we consider Value at Risk. By design, the corresponding solvency tests only limit the probability of insolvency and are incapable of imposing any stricter bound on the loss given default.

But the same failure is shared by Average Value at Risk. In spite of being sensitive to tail losses, Average Value at Risk still leaves too many degrees of freedom to control recovery. This is because the loss given default is captured by way of an average loss, which is too gross to exert a fine control on the recovery probability. An additional key deficiency is that all monetary risk measures in current solvency regulation focus on a residual quantity, i.e., the difference between assets and liabilities, that is owned by shareholders. This quantity is insufficient to adequately capture what will happen in the case of default.

The goal of this paper is to address the question:
\begin{center}
\emph{How should solvency tests be designed that control \\ the recovery on creditors' claims in the case of default?}
 \end{center}
\noindent Our contributions are the following:
\begin{enumerate}
\item[I.] We demonstrate that classical monetary risk measures such as Value at Risk and Average Value at Risk are unable to control recovery on creditors' claims in the case of default. In fact, we argue that, to capture this important aspect of tail risk, one has to abandon solvency tests based on the net asset value only and consider more articulated solvency tests based on both the net asset value and the firm's liabilities.
\item[II.] A novel risk measure, Recovery Value at Risk, is developed in the paper that can successfully resolve the failure of the standard risk measures employed in solvency regulation. We demonstrate that Recovery Value at Risk can serve as the basis of solvency tests. It admits an operational interpretation as a capital requirement rule. This new risk measure can be applied to both external and internal risk management and helps to quantify how far standard regulatory risk measures are from controlling liability recovery risk.
\item[III.] Our conceptual approach is flexible and leads to the construction of general recovery risk measures that include Recovery Average Value at Risk. This allows to integrate the ability to control the recovery on creditors' claims with other desirable properties such as convexity or subadditivity. Convexity facilitates applications to optimization problems such as portfolio choice under risk constraints. Subadditivity provides incentives for the diversification of positions and enables limit systems within firms for decentralized risk management.
\item[IV.] In order to better understand the behavior of recovery risk measures we illustrate how they react to changes of the joint distribution of the assets and the liabilities on the firm's balance sheet. We focus on two characteristics -- marginal distributions and stochastic dependence -- and compare risk measurements to the classical solvency benchmarks, i.e., Value at Risk and Average Value at Risk.
\item[V.] We discuss a possible strategy to calibrate recovery risk measures consistently with existing regulatory standards, following a common methodology chosen by regulators in the context of classical risk measures.
\item[VI.] We demonstrate how recovery risk measures can be applied to performance-based management of business divisions of firms. We define and investigate the appropriate notion of RoRaC-compatibility.
\item[VII.] Optimal tradeoffs between risk and return, as originally suggested by Harry Markowitz in the special case of the variance, may also be characterized for recovery risk measures. We show how efficient frontiers can be computed in this case.
\end{enumerate}

\noindent The paper is structured as follows. Section~\ref{sec:int_model} reviews solvency regulation based on Value at Risk and Average Value at Risk and reveals its failure to control recovery. Section~\ref{sec:revar} describes how to resolve this problem and develops the novel risk measure Recovery Value at Risk. Section~\ref{sec:convex} introduces a notion of general recovery risk measures that include the convex Recovery Average Value at Risk. Section \ref{sect: numerics} complements the conceptual innovations of this paper with detailed case studies and applications. These provide insights on how risk measures react to the shape of distributions and stochastic dependence. We also discuss calibration issues that arise when solvency regimes are modified. In order to show the wide spectrum of applicability of recovery risk measures in practice, we discuss risk allocation in the context of decentralized risk and performance management of firms and portfolio optimization under risk constraints. Proofs and further technical supplements are collected in the appendix.

\subsection*{Literature}

Solvency capital requirements impose constraints on the operations of businesses such as bank and insurance companies. Their purpose is to protect creditors from excessive downside risk. Capital requirements are an integral part of broader regulatory frameworks that allow companies to freely operate within pre-specified legal boundaries. Historically, regulatory deliberations like Basel I and Solvency I formulated simple rules. However, these could be exploited by regulatory arbitrage, see, e.g., \ci{BaselI}, \ci{SolvencyIb}, \ci{SolvencyIa}, and \ci{Jones}. Regulatory frameworks have been modified multiple times during the past decades, but -- as we will demonstrate in this paper -- serious problems remain.

A key issue is how to define the required solvency capital in an appropriate manner. Basel II, Solvency II, and the upcoming international Insurance Capital Standard compute solvency capital on the basis of Value at Risk, while Basel III and the Swiss Solvency Test use Average Value at Risk. The risk measure Value at Risk has been criticized in the context of solvency regulation since the 1990s, in particular due to its tail blindness and lack of convexity. Alternatives are provided by the axiomatic theory of risk measures --- initiated in a seminal paper by \ci{ADEH99} --- that systematically analyzes properties of risk measures, implications, and examples. The notion of coherent risk measure is introduced in \ci{ADEH99} and is generalized to the class of convex risk measures in \ci{frittelli2002} and \ci{FS02}. Key developments are discussed in the monograph \ci{FS} and the surveys \ci{FSW09} and \ci{FW15}. The coherence of Average Value at Risk is first established in \ci{acerbi2002coherence}. We refer to \ci{wang2020axiomatic} for a recent axiomatic characterization of Average Value at Risk.

Monetary risk measures are based on the notion of acceptability. While preferences rank distributions, random variables, or processes, acceptance sets divide this universe into acceptable objects and those that are not acceptable. Monetary risk measures are numerical representations of acceptance sets and parallel in this respect utility functionals that represent preferences. Within the theory of choice, risk measures provide a model of guard rails for the actions of financial firms. At the same time, they possess an operational interpretation as capital requirement rules, measuring the distance from acceptability in terms of cash or, more generally, eligible assets. We refer to \ci{FS} for a broad discussion on these aspects and to \ci{filipovic2008optimal}, \ci{adk2009}, \ci{farkas2014beyond}, \ci{FRW17}, and \ci{BFFM19} for specific applications to capital adequacy, hedging, risk sharing, and systemic risk. Our paper follows the same approach, i.e., taking the notion of acceptability as the starting point when formalizing recovery-based solvency tests. A different approach is pursued by the literature on acceptability indices that mainly focus on performance measurement, see \ci{aumann2008economic}, \ci{cherny2009}, \ci{foster2009operational}, \ci{brown2012aspirational}, \ci{drapeau2013risk}, \ci{gianin2013acceptability}, \ci{bielecki2014dynamic}.\footnote{Parametric families of Value at Risk were previously studied in this literature. But acceptability indices are applied to fixed univariate positions (modelling net asset values). In our case, a parametric family of Value at Risk or different risk measures are applied to bivariate positions (modelling net asset values jointly with liabilities). As a consequence, the formal construction of acceptability and their financial interpretation differs substantially from our approach.} A related concept is also the notion of Loss Value at Risk introduced by \ci{bignozzi2020risk}.

Monetary risk measures have natural applications to risk allocation and portfolio optimization problems. Risk allocation in the context of decentralized risk and performance management of firms has been widely investigated, e.g., in  \ci{tasche1999risk}, \ci{kalkbrener2005axiomatic}, \ci{tasche2007capital},
\ci{dhaene2012optimal}, \ci{bauer2013capital}, \ci{bauer2016marginal}, \ci{embrechts2018quantile}, \ci{weber_solvency2017}, \ci{HaKnWe20}, and \ci{bauer2020}. Portfolio optimization under risk constraints with a characterization of efficient frontiers goes back to the classical approach in \ci{Markowitz52}. In the context of Average Value at Risk, the solution to this problem was developed in \ci{RU00} and \ci{RU02}. The methodology relies on a representation of Average Value at Risk that is closely related to optimized certainty equivalents, which are discussed in \ci{BT87} and \ci{BT07}.  Applications to robust portfolio management are studied, e.g., in \ci{ZhuFu09}. We complement this previous work by discussing how efficient frontiers can be characterized for risk constraints in terms of recovery-based risk measures.

To the best of our knowledge, this paper is the first to introduce and study solvency capital requirements that are designed to control the recovery on creditors' claims. The literature on recovery rates has historically focused on explaining the determinants of recovery rates in specific settings, e.g., for corporate and government bonds or bank loans. We refer to \ci{duffie1999modeling}, \ci{altman2005link}, and \ci{guo2009modeling} for a presentation of different models for recovery rates and to \ci{khieu2012determinants}, \ci{jankowitsch2014determinants}, and \ci{ivashina2016ownership} for some recent empirical investigations.


\section{Solvency Regulation and Claims Recovery}
\label{sec:int_model}

The protection of creditors is a key goal of capital regulation. To achieve this goal, financial institutions are required to hold a certain amount of capital as a buffer against future losses. The regulatory capital is chosen such that it ensures an acceptable level of safety against the risk of default. The standard rules used in practice to compute solvency capital requirements are based on risk measures such as Value at Risk or Average Value at Risk. We demonstrate that these rules are insufficient to provide a satisfactory control on the recovery on creditors' claims and suggest an alternative approach that achieves this goal.

\subsection{Risk-Sensitive Solvency Regimes}

Most existing regulatory frameworks share a ``balance sheet approach'' to determine capital requirements. The random evolution of assets and liabilities of a financial institution is captured at time horizons specified by regulators, typically one year.\footnote{A balance sheet approach requires an internal model of the stochastic evolution of the balance sheet of the financial firm or insurance company that is subject to capital regulation. Many firms do not have sufficient capacities and expertise to implement and analyze such models. For this reason, in practice, simplifications are admissible which may substantially deviate from the original objectives of the regulator. Examples are the standard approach in the Insurance Capital Standard or the standard formula in Solvency II.} The following table displays a stylized balance sheet of a company at a generic time $t$:

\begin{center}
\begin{tabular}{|c|c|}
\hline
\bf Assets & \bf Liabilities \\
\hline\hline
\multirow{2}{*}{$A_t$}&$L_t$\\
\cline{2-2}
 &$E_t=A_t-L_t$\\
\hline
\end{tabular}
\end{center}

\noindent The quantity $E_t$ represents the net asset value of the firm and can be either positive or negative depending on whether the asset value $A_t$ is larger than the liability value $L_t$ or not. In the typical setting of a one-year horizon  we have two reference dates, which are denoted by $t=0$ (today) and $t=1$ (end of the year). The quantities at time $t=0$ are known whereas the quantities at time $t=1$ are random variables on a given probability space $(\Omega,\cF,\probp)$. In a risk-sensitive solvency framework, a company is deemed adequately capitalized if its {\em available capital} $E_0$ is larger than a suitable {\em solvency capital requirement} that reflects the inherent risk in the evolution of the balance sheet. This is typically captured by applying a suitable risk measure $\rho$ to the net asset value variation $\Delta E_1:=E_1-E_0$.\footnote{In practice, solvency capital requirements may only refer to ``unexpected'' losses. In this case, $E_0$ is replaced by the expected value of (the suitably discounted) $E_1$. In this respect, the European regulatory framework for insurance companies Solvency II is contradictory in itself. We refer to \ci{HaKnWe20} for a detailed discussion.} The corresponding {\em solvency test} is formally defined by:\footnote{For simplicity, we assume in this paper that interest rates over the one-year horizon are approximately zero. For adjustments on the definition of the solvency tests if interest rates are not zero see \ci{christiansen2014}.}
\begin{equation}
\label{solvency test}
\rho(\Delta E_1) \; \leq  \; E_0.
\end{equation}
If $\rho$ is a monetary risk measure such as Value at Risk or Average Value at Risk, condition \eqref{solvency test} can be equivalently expressed in terms of the future net asset value $E_1$ only as
\begin{equation}
\label{solvency test 2}
\rho(E_1) \; \leq  \; 0.
\end{equation}

\noindent The standard risk measures are {\em Value at Risk} ($\VaR$) and {\em Average Value at Risk} ($\avar$) at some pre-specified level $\alpha\in(0,1)$:
\[
\VaR_\alpha(X):=\inf\{x\in\R \,; \ \probp(X+x<0)\leq\alpha\}, \ \ \ \ \avar_\alpha(X):=\frac{1}{\alpha}\int_0^\alpha\VaR_\beta(X)d\beta
\]
where $X$ is some random variable.\footnote{Throughout the paper we apply the following sign convention: Positive values of $X$ represents a profit or a positive balance, negative values of $X$ represent a loss or a negative balance.} In particular, the risk measure $\VaR$ corresponds to a quantile of the underlying probability distribution.

In insurance regulation, $\VaR$ at level $\alpha=0.5\%$ is used in the Insurance Capital Standard and in Solvency II while $\avar$ at level $\alpha=1\%$ is adopted in the Swiss Solvency Test. In banking regulation, $\avar$ with level $\alpha=2.5\%$ has recently become the reference risk measure in Basel III, where it replaces $\VaR$ at level $\alpha=1\%$. In a $\VaR$ setting, the solvency test~\eqref{solvency test 2} can equivalently be reformulated as
\begin{equation}
\label{eq:solvency test VaR}
\VaR_\alpha (E_1) \leq 0 \ \iff \ \probp(E_1<0)\leq\alpha \ \iff \ \probp(E_1\geq0)\geq1-\alpha.
\end{equation}
This shows that a company is adequately capitalized under $\VaR$ if it is able to maintain its default probability below a certain level. Similarly, in an $\avar$ setting, we can equivalently rewrite the solvency test~\eqref{solvency test 2} as\footnote{The second equivalence holds provided the cumulative distribution function of $E_1$ is, e.g., continuous.}
\begin{equation}
\label{eq:solvency test AVaR}
\avar_\alpha (E_1) \leq 0 \ \iff \ \int_0^\alpha\VaR_\beta(E_1)d\beta\leq0 \ \iff \ \E(E_1\vert E_1\leq-\VaR_\alpha(E_1))\geq0.
\end{equation}
Hence, a company is adequately capitalized under $\avar$ if on the lower tail beyond the $\alpha$-quantile it is solvent on average. In this case, we automatically have $\probp(E_1<0)\leq\alpha$. In other words, if we fix the same probability level $\alpha$, capital adequacy under $\avar$ is more conservative than capital adequacy under $\VaR$.

It is often stressed that --- in contrast to $\var$ --- $\avar$ is a tail-sensitive risk measure and, hence, captures tail risk in a more comprehensive way. In fact, $\var$ is completely blind to the tail of the reference loss distribution beyond a certain quantile level. While this is correct, one should bear in mind that $\avar$ captures tail risk in a {\em specific} way, namely via expected losses in the tail, thereby leaving many degrees of freedom to the behavior of the tail distribution.

\subsection{Claims Recovery Under $\VaR$ and $\avar$}

The point of departure of our contribution is to highlight that risk measures such as $\VaR$ and $\avar$ fail to provide a direct control on a fundamental aspect of tail risk, namely the recovery on creditors' claims. The basic problem is that both risk measures are functions of the net asset value $E_1$ only. The net asset value summarizes the financial resources of the equity holders without any reference to leverage, i.e., without imposing any direct constraints on the liabilities $L_1$. However, controlling the recovery on claims requires to deal explicitly with $L_1$.

This failure is documented by the next proposition. To motivate it, observe that, for given $\alpha\in(0,1)$, the solvency test based on $\var$ as described in \eqref{eq:solvency test VaR} guarantees that the probability of solvency is at least $1-\alpha$. The same is true for the solvency test based on $\avar$ at the same level because $\avar$ dominates $\var$. The question we ask is if and how the probability $\probp(A_1\geq\lambda L_1)$ of recovering at least a fraction $\lambda\in(0,1)$ of claims can be made higher than the probability of solvency. For V@R and AV@R the answer is negative: The lower bound $1-\alpha$ is sharp for any target fraction of claims payments. In other words, both $\VaR$ and $\avar$ impose the same weak lower bound on recovery probabilities, and this bound cannot be improved upon, even if the target recovery fraction is arbitrarily small.

\begin{prop}
\label{prop:no recovery control}
We denote by $\cX$ a set of positive random variables on some nonatomic probability space $(\Omega,\cF,\probp)$. We assume that $\cX$ contains all positive discrete random variables. For all $\alpha\in(0,1)$ and $\lambda\in(0,1)$ we have
\begin{eqnarray*}
1-\alpha & = & \inf\{\probp(A\geq\lambda L) \,; \ A,L\in\cX, \ \avar_\alpha(A-L)\leq0\} \\
& = &  \inf\{\probp(A\geq\lambda L) \,; \ A,L\in\cX, \ \VaR_\alpha(A-L)\leq0\}.
\end{eqnarray*}
\end{prop}
\begin{proof}
See Section \ref{proof:no recovery control}.
\end{proof}

The preceding proposition shows that, from the perspective of controlling the probability of claims recovery (beyond the probability of solvency), there is little difference between $\VaR$ and $\avar$. This lack of control is not desirable for a financial regulator, as companies that seek to boost the payoff to shareholders are not prevented from taking on excessive risk thereby significantly reducing recovery payments to creditors in the case of their own default. This is illustrated by the following stylized but insightful example.

\begin{ex}\label{ex:interests}
We consider a scenario space $\Omega$ consisting of two states, $g$ (the good state) and $b$ (the bad state). The probability of the bad state is $\probp(b)=\frac{\alpha}{2}$ with $\alpha$ close to zero, say $\alpha=0.5\%$ or $\alpha=1\%$. A financial company sells a contract to a customer that results in the following liability schedule for the company:
$$
L_1 (\omega)=
\begin{cases}
1 & \mbox{if} \ \omega = g,\\
100 & \mbox{if} \ \omega = b.
\end{cases}
$$
The company can manage its assets by engaging in a stylized financial contract with zero initial cost transferring dollars from the good state to the bad state. More specifically, we assume that the company can choose one of the following asset profiles at time 1:
 \[
A^k_1 (\omega)=
\begin{cases}
101-k & \mbox{if} \ \omega = g,\\
k & \mbox{if} \ \omega = b,
\end{cases} \quad \mbox{with}\quad k \in [0,100].
\]
Hedging its liabilities completely would require the company to choose $k = 100$. However, since the contract transforms dollars in the high probability scenario into dollars in the low probability scenario, this is not attractive from the point of view of the company. Indeed, for any $k \in [0,100]$, the company's net asset value is given by
\[
E^k_1 (\omega)=
\begin{cases}
100-k & \mbox{if} \ \omega = g,\\
k-100 & \mbox{if} \ \omega = b.
\end{cases}
\]
Due to limited liability, the corresponding shareholder value is
$$
\max\{E^k_1 (\omega), 0 \}=
\begin{cases}
100-k & \mbox{if} \ \omega = g,\\
0 & \mbox{if} \ \omega = b.
\end{cases}
$$
Hence, the choice $k = 0$ is optimal from the perspective of shareholders. We show that this choice is possible under capital requirements based on $\VaR$ and $\avar$. In fact, the company is adequately capitalized under $\VaR$ and $\avar$ at level $\alpha$ regardless of the size of $k$. Indeed,
\[
\var_\alpha(E^k_1)= k-100 \leq 0 , \ \ \ \avar_\alpha(E^k_1)=\frac{1}{\alpha}\left(\frac{\alpha}{2}(100-k)+\frac{\alpha}{2}(k-100)\right) = 0.
\]
At the same time, this choice is detrimental for the creditors because it leads to no recovery on the expected claims payment. Indeed, in the default state $b$, the creditor's recovery on their claims is equal to $\frac{k}{100}$ and may take any value between 0 and 1, depending on the level of $k$. For the optimal choice from the perspective of shareholders, namely $k=0$, the recovery fraction in state $b$ is minimal, in fact zero.
\end{ex}

The example shows that pursuing the interests of shareholders might trigger investment decisions with adverse effects on creditors. A solvency framework based on $\var$ and $\avar$ fails to disincentivize firms from taking investment decisions that increase shareholders' value at the price of jeopardizing their ability to cover liabilities. In the next section we show how to mitigate this deficiency of current regulatory frameworks.


\section{Recovery Value at Risk}
\label{sec:revar}

In this section we introduce a solvency test that controls the loss given default by imposing suitable bounds on the recovery on creditors' claims. The test is based on a new risk measure called {\em Recovery Value at Risk}. The main difference with respect to standard risk measures like $\var$ and $\avar$ is that Recovery Value at Risk is not a function of the net asset value $E_1$ only but also of the liabilities $L_1$. As shown in the previous section, this extension is necessary if we want to explicitly control the recovery on claims. Throughout the section we continue to use the balance sheet notation introduced in Section~\ref{sec:int_model}.

\subsection{Introducing $\revar$}\label{sec:introrevar}

Creditors receive at least a recovery fraction $\lambda\in[0,1]$ on their claims payments if\footnote{For simplicity, we neglect bankruptcy costs (administrative expenses, legal fees, etc.) which can substantially impair the size of recovery. Regulators may improve the efficiency of bankruptcy procedures and thereby decrease their costs, e.g., by requiring \emph{last wills of financial institutions}.}
\begin{equation}\label{eq:recovery_ref}
A_1\geq\lambda L_1 \ \iff \ E_1 + (1-\lambda) L_1 \geq 0.
\end{equation}
In this event, assets may not be sufficient to meet all obligations, but they cover at least a fraction $\lambda$ of liabilities. We control recovery by imposing lower bounds on the recovery probabilities
$$
\probp(A_1\geq\lambda L_1)
$$
for all recovery fractions $\lambda\in [0,1]$.\footnote{We can rewrite the event of recovering a fraction of $\lambda$ in different ways:
$$\{A_1\geq \lambda L_1 \}  = \{ A_1-\lambda L_1\geq0 \} = \left\{\frac{A_1}{L_1}\geq\lambda\right\}$$
where the last equality holds only if $L_1>0$. In this sense, our approach can also be interpreted in terms of target probabilities for future leverage ratios. Focusing on the modified net asset value $A_1-\lambda L_1$ instead of $\frac{A_1}{L_1}$ is more aligned with current regulation and has the mathematical advantage to avoid divisions by zero.} For this purpose, we introduce the following risk measure.
\begin{defi}\label{def:revar}
We denote by $L^0$ the vector space of all random variables on some probability space $(\Omega, \fil,\probp)$. Let $\gamma:[0,1]\to(0,1)$ be an increasing function. The \emph{Recovery Value at Risk}
$$
\revar_\gamma:L^0\times L^0\to\bbr\cup\{\infty \}
$$
with \emph{level function} $\gamma$ is defined by
\begin{equation}\label{eq:revar}
\revar_\gamma (X,Y) := \sup_{\lambda \in [0,1]}\var_{\gamma(\lambda)}(X+(1-\lambda)Y).
\end{equation}
\end{defi}

\noindent If the random variables $X$ and $Y$ in Definition \ref{def:revar} are interpreted, respectively, as the net asset value $E_1$ and liabilities $L_1$ in a company's balance sheet,\footnote{To allow for different applications, we mathematically define recovery risk measures over generic pairs $(X,Y)$ without any restriction on the sign of $X$ and $Y$ and without any specific assumptions about their relationship and interpretation. In the relevant applications, we have $X=E_1$ or $X= \Delta E_1$ and $Y=L_1$.} the risk measure $\revar$ can be used to formulate a solvency test of the form \eqref{solvency test}. As shown in Remark~\ref{rem:revar_solvency test}, the condition
\begin{equation}\label{revar_solvency test}
\revar_\gamma(\Delta E_1,L_1)\leq E_0
\end{equation}
is equivalent to requiring that the recovery probabilities satisfy
\begin{equation}\label{revar_solvency test 2}
\probp(A_1<\lambda L_1) \leq \gamma(\lambda) \ \iff \ \probp(A_1\geq\lambda L_1) \geq 1-\gamma(\lambda)
\end{equation}
for all recovery fractions $\lambda \in [0,1]$. This guarantees the desired control on the loss given default. In particular, the solvency test \eqref{revar_solvency test} can be seen as a refinement of the standard solvency test \eqref{eq:solvency test VaR} based on $\VaR$ where the probability bound $\alpha$ is replaced by a bound that depends on the target recovery fraction through the function $\gamma$. The assumption that $\gamma$ is increasing captures the basic requirement that smaller recovery fractions on liabilities should be guaranteed at higher probability levels.

\begin{remark}
The recovery-adjusted solvency test \eqref{revar_solvency test} can easily be combined with a standard solvency test based on $\var$ at level $\alpha$. Indeed, setting $\gamma(1)=\alpha$, it follows that
\begin{equation}
\label{eq: comparison rec var with var}
\revar_\gamma(\Delta E_1,L_1) \geq \VaR_\alpha(\Delta E_1),
\end{equation}
showing that recovery-based capital requirements are more stringent than the standard ones. The standard $\VaR$ test can be reproduced by setting $\gamma(\lambda)=\alpha$ for all recovery fractions $\lambda\in[0,1]$, in which case the inequality in \eqref{eq: comparison rec var with var} becomes an equality. It is worth highlighting that the level $\gamma(1)$ may also be strictly larger than a regulatory level $\alpha$. In this case, the inequality in \eqref{eq: comparison rec var with var} may be reversed. The recovery-based risk measure $\revar$ can be viewed as a flexible generalization of $\VaR$ that reacts to the entire loss tail as specified by the recovery function $\gamma$. As explained in Remark~\ref{rem:cond_recov}, the solvency test \eqref{revar_solvency test} also controls the conditional recovery probabilities given default. Contrary to $\VaR$ and $\avar$, the risk measure $\revar$ depends on the joint distribution of the tuple $(E_1,L_1)$. In particular, the marginal distributions of $E_1$ and $L_1$ are not a sufficient statistic for $\revar$ but knowledge of the dependence structure, as captured, e.g., by the copula of the pair, is additionally required.\footnote{The evaluation of $\revar$ is technically not more complicated than the computation of standard solvency capital requirements, since it only requires the computation of a supremum of distribution-based risk measures, namely $\VaR$'s. In practical situations, knowledge of the precise joint distribution between assets and liabilities is challenging. We refer to Section \ref{sect: numerics} for a detailed numerical illustration. The problem is akin to risk estimation in the presence of aggregate positions where a model for the joint distribution is also needed. The structure of recovery risk measures opens up a variety of interesting technical questions related to dependence modelling that, however, go beyond the scope of the current work. The rich and growing literature on the topic is a good starting point to address such questions, see, e.g., \ci{embrechts2013model}, \ci{bernard2014risk}, \ci{bernard2017risk}, \ci{cai2018asymptotic}.}
\end{remark}


\subsection{Choosing the Recovery Function}
\label{sect: choice of gamma}

The choice of the recovery function $\gamma$ is a critical step in our model and should reflect the risk profile of the (external or internal) regulators. In this section we describe a class of parametric recovery functions\footnote{We describe a methodology to calibrate $\gamma$ to an existing regulatory framework in Section \ref{sect: numerics}. This mirrors a common strategy chosen by regulators when adapting a new solvency setting to replace a pre-existing one. Another possibility to choose $\gamma$ is to elicit it from the risk profile of risk managers or customers, e.g., by way of a questionnaire targeting recovery distributions. This would raise a number of interesting questions for future research that are, however, beyond the scope of the paper.} that  provides an ideal compromise between flexibility and tractability and can be successfully tailored to different applications as demonstrated in Section~\ref{sect: numerics}.

We consider step-wise recovery functions of the form
\begin{equation}
\label{eq: parametric gamma}
\gamma(\lambda)=
\begin{cases}
\alpha_1 & \mbox{if} \ 0=r_0\leq\lambda<r_1,\\
\alpha_2 & \mbox{if} \ r_1\leq\lambda<r_2,\\
 \ \vdots\\
\alpha_n & \mbox{if} \ r_{n-1}\leq\lambda<r_n,\\
\alpha_{n+1} & \mbox{if} \ r_n\leq\lambda\leq r_{n+1}=1,\\
\end{cases}
\end{equation}
with $0<\alpha_1<\cdots<\alpha_{n+1}<1$ and $0<r_1<\cdots<r_n<1$. The parameters $r_i$ correspond to critical target recovery fractions while the parameters $\alpha_i$ define bounds on the corresponding recovery probabilities for every $i=1,\dots,n+1$. As shown in the next proposition, the $\revar$ induced by such recovery functions can be expressed as a maximum of finitely many $\VaR$'s.

\begin{prop}
\label{prop: revar with piecewise gamma}
Let $\gamma$ be defined as in \eqref{eq: parametric gamma}. Then, for all $X,Y\in L^0$ with $Y\geq0$
\[
\revar_\gamma(X,Y) = \max_{i=1,\dots,n+1}\VaR_{\alpha_i}(X+(1-r_i)Y).
\]
\end{prop}
\begin{proof}
See Section \ref{sect: proof parametric gamma}.
\end{proof}

The preceding proposition shows that, under a recovery function of the form \eqref{eq: parametric gamma}, the recovery-based solvency test \eqref{revar_solvency test} takes the particularly simple form:
\begin{equation}
\label{eq: solvency test under stepwise gamma}
\probp(A_1\geq r_i L_1) \geq 1-\alpha_i, \quad i = 1, \dots, n+1.
\end{equation}
In this case, a company is adequately capitalized under $\revar$ if, for every $i=1,\dots,n+1$, assets are sufficient to cover a fraction $r_i$ of liabilities with a probability of at least $1-\alpha_i$. The largest recovery probability $\alpha_{n+1}$ with target recovery fraction $r_{n+1}=1$ caps the default probability and could correspond to the level of $\VaR$ in a classical solvency test. This shows that a solvency test of the form \eqref{eq: solvency test under stepwise gamma} can easily be harmonized with the solvency tests currently used in solvency regulation.


\subsection{Basic Properties of $\revar$}
\label{sect: basic properties}

We ask which basic properties of $\VaR$ are inherited by its recovery counterpart $\revar$. For a comprehensive survey on scalar monetary risk measures we refer to \ci{FS}. A monetary risk measure is a function
$\rho:L^0\to\R\cup\{\infty\}$ that satisfies the following two properties:
\begin{itemize}
    \item \emph{Cash invariance}: $\rho(X+m)=\rho(X)-m$ for all $X\in L^0$ and $m\in\R$;
    \item \emph{Monotonicity}: $\rho(X_1)\leq\rho(X_2)$ for all $X_1,X_2\in L^0$ with $X_1\geq X_2$ $\probp$-almost surely.
\end{itemize}
The cash invariance property formalizes that adding cash to a capital position reduces risk by exactly the same amount and implies that risk is measured on a monetary scale. In particular, cash invariance allows to rewrite the risk measure as a capital requirement rule:
\[
\rho(X) = \inf\{m\in\R \,; \ \rho(X+m)\leq0\},
\]
i.e., the quantity $\rho(X)$ can be interpreted as the minimal amount of cash that needs to be injected into the position $X$ in order to pass the solvency test in \eqref{solvency test 2}. If the position already fulfills this solvency condition, then $-\rho(X)$ corresponds to the maximal amount of capital that can be extracted from the balance sheet without compromising capital adequacy. Monotonicity reflects that larger capital positions correspond to lower risk and to lower capital requirements. In addition to its defining properties, a monetary risk measure may possess the following properties:
\begin{itemize}
    \item \emph{Convexity}: $\rho(aX_1+(1-a)X_2)\leq a\rho(X_1)+(1-a)\rho(X_2)$ for all $X_1,X_2\in L^0$ and $a\in[0,1]$;
    \item \emph{Subadditivity}: $\rho(X_1+X_2)\leq\rho(X_1)+\rho(X_2)$ for all $X_1,X_2\in L^0$;
    \item \emph{Positive homogeneity}: $\rho(aX)=a\rho(X)$ for all $X\in L^0$ and $a\in(0,\infty)$.
    \item \emph{Normalization}: $\rho(0)=0$.
\end{itemize}
The first two properties characterize the behavior of the risk measure with respect to aggregation and require that diversification is not penalized. The third property specifies that risk measurements scale with the size of positions.

The next proposition records elementary properties of $\revar$. In particular, $\revar$ is a standard monetary risk measure if the second argument is fixed.

\begin{prop}\label{prop:elementvar}
The risk measure $\revar_\gamma$ has the following properties:
\begin{enumerate}
  \item \emph{Cash invariance in the first component:} For all $X,Y\in L^0$ and $m\in\R$
\[
\revar_\gamma(X+m,Y)=\revar_\gamma(X,Y)-m.
\]
  \item \emph{Monotonicity:} For all $X_1,X_2,Y_1,Y_2\in L^0$ with $X_1\geq X_2$ and $Y_1\geq Y_2$ $\probp$-almost surely\footnote{Note that monotonicity does not mean that increasing leverage would lead to a decrease in risk. While it is true that an increase in the value of liabilities might be accompanied by a decrease in risk, this is only possible if the value of assets increase in parallel. However, if assets are held constant, then an increase in the value of liabilities will always cause an increase in risk. In this respect, monotonicity does not differ from the standard monotonicity property of classical monetary risk measures.}
\[
\revar_\gamma(X_1,Y_1)\leq\revar_\gamma(X_2,Y_2).
\]
  \item \emph{Positive homogeneity:} For all $X,Y\in L^0$ and $a\in[0,\infty)$
\[
\revar_\gamma(aX,aY) = a\revar_\gamma(X,Y).
\]
  \item \emph{Star-shapedness\footnote{We refer to the recent preprint \ci{castagnoli2021star} for a study of star-shaped risk measures.} in the first component:} For all $X,Y\in L^0$ with $Y\geq0$ and $a\in[1,\infty)$
\[
\revar_\gamma(aX,Y)\geq a\revar_\gamma(X,Y).
\]
 \item \emph{Normalization:} For every $Y\in L^0$ with $Y \geq 0$ we have $\revar_\gamma(0,Y)=0$.
 \item \emph{Finiteness:}
For all $X, Y\in L^0$ with $Y \geq 0$ we have $\revar_\gamma(X,Y)<\infty$ under any of the following conditions:
$\gamma(0)>0$, $\VaR_{\gamma(0)}(X)<\infty$, or $X$ is bounded from below.
\end{enumerate}
\end{prop}
\begin{proof}
See Section \ref{proof:elementvar}.
\end{proof}

The previous proposition shows that $\revar$ is a standard monetary risk measure in its first component and can conveniently be expressed as a capital requirement:
\begin{eqnarray*}
\revar_\gamma(E_1,L_1) &  =  & \inf\{m\in\R \,; \ \revar_\gamma(E_1+m,L_1)\leq0\} \\
&=&
\inf\{m\in\R \,; \ \probp(A_1+m<\lambda L_1)\leq\gamma(\lambda), \ \forall \lambda\in[0,1]\} \\
&=&
\inf\{m\in\R \,; \ \probp(A_1+m\geq\lambda L_1)\geq1-\gamma(\lambda), \ \forall \lambda\in[0,1]\},
\end{eqnarray*}
where the second equality is a consequence of Remark~\ref{rem:revar_solvency test}. This leads to the following useful operational interpretation of $\revar$:
\begin{itemize}
\item  If $\revar_\gamma(E_1,L_1)>0$, the company fails to pass the recovery-based solvency test \eqref{revar_solvency test} and $\revar_\gamma(E_1,L_1)$ is the minimal amount of cash that needs to be added to its assets in order to become adequately capitalized.
\item If $\revar_\gamma (E_1,L_1)<0$, the company is adequately capitalized according to the recovery-based solvency test \eqref{revar_solvency test} and  $-\revar_\gamma (E_1,L_1)$ is the maximal amount of cash that may be extracted from the asset side without compromising capital adequacy.
\end{itemize}

\begin{rem}\label{rem:2nd-oper}
From an operational perspective the interpretation of monetary risk measures as capital requirement rules relies on the cash invariance property. $\revar$ is cash invariant in the first but not in the second argument. If one intends to modify the liabilities, e.g.\ by transferring them to another institution, instead of the assets on the balance sheet, an alternative definition of $\revar$ is appropriate, namely (``L'' stands for ``liabilities'')
\begin{equation}\label{eq:Lrevar}
\Lrevar_\gamma (A_1,L_1) := \sup_{\lambda \in (0,1]} \;  \frac 1 \lambda  \cdot \var_{\gamma(\lambda)}(A_1-\lambda L_1).
\end{equation}
In this case, the correct way to express the solvency test \eqref{revar_solvency test} is
\begin{equation}
\label{eq: Lrevar_solvency test}
\Lrevar_\gamma (A_1,L_1) \leq 0,
\end{equation}
which is still equivalent to condition \eqref{revar_solvency test 2}. Note that $\Lrevar$ is cash invariant (in the appropriate sense) with respect to its second argument, i.e., for all $A_1,L_1\in L^0_+$ and $m\in\R$
\[
\Lrevar_\gamma(A_1,L_1+m)=\Lrevar_\gamma(A_1, L_1)+m.
\]
This leads to the following operational interpretation:
\begin{itemize}
\item  If $\Lrevar_\gamma (A_1,L_1)>0$, the company fails the solvency test \eqref{eq: Lrevar_solvency test} and $\Lrevar_\gamma (A_1,L_1)$ is the minimal nominal amount of liabilities that needs to be removed from the balance sheet in order to pass the test, e.g., by transferring these liabilities to suitable equity holders outside the firm.
\item If $\Lrevar_\gamma (A_1,L_1)<0$, the company is adequately capitalized. The company may at most create an additional amount $-\Lrevar_\gamma (A_1,L_1)$ of liabilities, e.g., via additional debt, and immediately  distribute the same amount of cash to its shareholders.
\end{itemize}
Observe that assets $A_1$ and liabilities $L_1$ are used in the definition of $\Lrevar$ instead of the net asset value $E_1$ and liabilities $L_1$ in order to obtain a simple cash-invariant recovery risk measure with a transparent operational interpretation.\footnote{Combining $\revar$ and $\Lrevar$ leads to the question of how to combine asset and liability management for capital adequacy purposes, which, however, goes beyond the scope of this paper. The literature on set-valued risk measures may help to address this question.}
\end{rem}


\section{General Recovery Risk Measures}\label{sec:convex}

The risk measure $\revar$ allows to control the loss given default by prescribing suitable bounds on the probability that part of the creditors' claims can be recovered. If one replaces $\VaR$ with other monetary risk measures, e.g. convex risk measures, one obtains recovery risk measures of a different type. In particular, by choosing appropriate monetary risk measures as the basic ingredients, it is possible to construct convex recovery risk measures. This may be desirable from the perspective of decentralized risk management or optimal risk sharing and capital allocation. In this section we describe the general structure of recovery risk measures and give special attention to recovery risk measures based on $\avar$. We continue to use the balance sheet notation introduced in Section~\ref{sec:int_model}.

\subsection{Introducing Recovery Risk Measures}
\label{sec:mathstruc}

To motivate the general definition of a recovery risk measure, we observe that $\revar$ may be expressed in terms of a decreasing family of monetary risk measures indexed by recovery fractions $\lambda\in[0,1]$. Indeed, for a given level function $\gamma$, the collection of monetary risk measures $\rho_\lambda:L^0\to\R$ given by
\[
\rho_\lambda(X) = \var_{\gamma(\lambda)}(X), \quad \lambda \in [0,1],
\]
defines the associated $\revar$ by setting
$$
\revar_\gamma(X,Y) =  \sup_{\lambda\in[0,1]}\rho_\lambda(X+(1-\lambda)Y).
$$
By construction, smaller recovery fractions are guaranteed with higher probability, which is captured by $\gamma$ being increasing. As a consequence, the family of maps  $\rho_\lambda$, $\lambda \in [0,1]$, is decreasing in the sense that $\rho_{\lambda_1} \geq \rho_{\lambda_2}$ whenever $\lambda_1\leq\lambda_2$. A smaller recovery fraction corresponds to a more conservative risk measure. This motivates the general definition of a recovery risk measure.

\begin{defi}\label{def:rerisk}
Let $L^0$ be the set of random variables on some probability space $(\Omega, \fil,\probp)$. We denote by $\xcal \subseteq L^0$ a vector space that contains the constants. For every $\lambda\in[0,1]$ consider a map $\rho_\lambda:\cX\to\R\cup\{\infty\}$ and assume that $\rho_{\lambda_1} \geq \rho_{\lambda_2}$ whenever $\lambda_1\leq\lambda_2$. The \emph{recovery risk measure}
$$
\rerho:\xcal\times\xcal \to\bbr\cup\{\infty\}
$$
is defined by
\begin{equation}\label{eq:rerisk}
\rerho(X,Y) := \sup_{\lambda\in[0,1]}\rho_\lambda(X+(1-\lambda)Y).
\end{equation}
\end{defi}

\noindent In line with our discussion on $\revar$, if the random variables $X$ and $Y$ in Definition \ref{def:rerisk} are respectively interpreted as the net asset value $E_1$ and liabilities $L_1$ in a company's balance sheet, the recovery risk measure $\rerho$ can be employed to formulate a solvency test of the form \eqref{solvency test}. Indeed, similarly to what we have shown in Remark~\ref{rem:revar_solvency test}, we have
\begin{equation}
\label{rerho_solvency test}
\rerho(\Delta E_1,L_1)\leq E_0\quad \Longleftrightarrow \quad \forall\,\lambda\in[0,1]\,:\;\rho_\lambda(A_1-\lambda L_1)\leq0.
\end{equation}
The specific interpretation of this recovery-based solvency test will, of course, depend on the choice of the monetary risk measures used to build $\rerho$.

The next result collects some basic properties of recovery risk measures. In particular, we analyze how a recovery risk measure inherits the key properties of its underlying building blocks. If the risk measures $\rho_\lambda$'s are convex, the recovery risk measure $\rerho$ admits a dual representation, which is recorded in Section \ref{sec:dual representation general} in the appendix. This type of duality results plays an important role in applications such as optimization problems involving risk measures.

\begin{prop}\label{prop:elementrerho}
A recovery risk measure $\rerho:\xcal\times\xcal \to\bbr\cup\{\infty\}$ has the following properties:
\begin{enumerate}
  \item  \emph{Cash invariance in the first component:} If $\rho_\lambda$ is cash invariant for every $\lambda\in[0,1]$, then for all $X,Y\in\xcal$ and $m\in\R$
\[
\rerho(X+m,Y)=\rerho(X,Y)-m.
\]
  \item \emph{Monotonicity:} If $\rho_\lambda$ is monotone for every $\lambda\in[0,1]$, then for all $X_1,X_2,Y_1,Y_2\in\xcal$ such that $X_1\geq X_2$ and $Y_1\geq Y_2$ $\probp$-almost surely
\[
\rerho(X_1,Y_1) \leq \rerho(X_2,Y_2).
\]
  \item \emph{Convexity:} If $\rho_\lambda$ is convex for every $\lambda\in[0,1]$, then for all $X_1,X_2,Y_1,Y_2\in\xcal$ and $a\in[0,1]$
\[
\rerho(aX_1+(1-a)X_2,aY_1+(1-a)Y_2) \leq a\rerho(X_1,Y_1)+(1-a)\rerho(X_2,Y_2).
\]
  \item \emph{Subadditivity:} If $\rho_\lambda$ is subadditive for every $\lambda\in[0,1]$, then for all $X_1,X_2,Y_1,Y_2\in\xcal$
\[
\rerho(X_1+X_2,Y_1+Y_2) \leq \rerho(X_1,Y_1)+\rerho(X_2,Y_2).
\]
  \item \emph{Positive homogeneity:} If $\rho_\lambda$ is positively homogeneous for every $\lambda\in[0,1]$, then for all $X,Y\in\xcal$ and $a\in[0,\infty)$
\[
\rerho(aX,aY) = a\rerho(X,Y).
\]
  \item \emph{Star-shapedness in the first component:} If $\rho_\lambda$ is monotone and positively homogeneous for every $\lambda\in[0,1]$, then for all $X, Y\in\xcal$ with $Y\geq 0$ and $a\in[1,\infty)$
\[
\rerho(aX,Y)\geq a\rerho(X,Y).
\]
 \item \emph{Normalization:} If $\rho_\lambda$ is monotone and $\rho_\lambda(0)=0$ for every $\lambda\in[0,1]$, then $\rerho(0,Y)=0$ for every $Y\in\xcal$ with $Y\geq0$.
  \item \emph{Finiteness:} If $\rho_\lambda$ is monotone for every $\lambda\in[0,1]$, then for every $X\in\xcal$ with $\rho_0(X)<\infty$ and for every $Y\in\xcal$ with $Y\geq0$ we have $\rerho(X,Y)<\infty$.
\end{enumerate}
\end{prop}
\begin{proof}
See Section~\ref{proof:elementrerho}.
\end{proof}


\subsection{Recovery Average Value at Risk}\label{sec:avar}

The recovery risk measure $\revar$ shares one major deficiency with the classical $\VaR$, namely the lack of convexity. Unlike convex risk measures, $\VaR$ and its recovery counterpart may thus penalize diversification, i.e., the capital requirement of a diversified position might be higher than the maximum of the capital requirements of the individual non-diversified positions. In addition, the lack of convexity complicates the solution of portfolio optimization problems with constraints on the downside risk and prevents the construction of limit systems within companies that facilitate decentralized risk management. A useful alternative is the following recovery-based version of $\avar$.

\begin{defi}\label{def:reavar}
We denote by $L^1$ the vector space of integrable random variables on some probability space $(\Omega, \fil,\probp)$. Let $\gamma: [0,1] \to (0,1)$ be an increasing function. The \emph{Recovery Average Value at Risk}
$$
\reavar_\gamma:L^1\times L^1\to\bbr\cup\{\infty\}
$$
with level function $\gamma$ is defined by
\begin{equation}\label{eq:reavar}
\reavar_\gamma (X,Y) := \sup_{\lambda \in [0,1]}\avar_{\gamma(\lambda)}(X+(1-\lambda)Y).
\end{equation}
\end{defi}

\noindent In line with our discussion on general recovery risk measures, if the random variables $X$ and $Y$ in Definition \ref{def:reavar} are interpreted, respectively, as the net asset value $E_1$ and liabilities $L_1$ in a company's balance sheet, the recovery risk measure $\reavar$ can be used to formulate the solvency test~\eqref{rerho_solvency test}:
\begin{equation}
\label{reavar_solvency test}
\reavar_\gamma(\Delta E_1,L_1)\leq E_0\quad \Longleftrightarrow \quad \forall\,\lambda\in[0,1]\,:\;\avar_{\gamma(\lambda)}(A_1-\lambda L_1)\leq0.
\end{equation}
This means that a company will be adequately capitalized according to $\reavar$ with level function $\gamma$ if for all recovery fractions $\lambda\in[0,1]$ the modified net asset value $A_1-\lambda L_1$ is positive on average on the lower tail beyond the $\gamma(\lambda)$-quantile. Since $\avar$ dominates $\VaR$ at the same level, domination is inherited by their recovery-based versions, i.e., for all $X,Y\in L^1$
$$
\reavar_\gamma (X,Y) \geq \revar_\gamma (X,Y).
$$
This implies, in particular, that the solvency test \eqref{reavar_solvency test} is stricter than \eqref{revar_solvency test} and the recovery probabilities are still controlled as described in Remark~\ref{rem:revar_solvency test}. In the next example we show that $\reavar$ is the maximum of finitely many AV@R's if the recovery function $\gamma$ is piecewise constant. This parallels the representation of $\revar$ recorded in Proposition \ref{prop: revar with piecewise gamma}.

\begin{prop}
\label{prop: parametric gamma avar}
Let $\gamma$ be defined as in \eqref{eq: parametric gamma}. Then, for all $X,Y\in L^1$ with $Y\geq0$
\[
\reavar_\gamma(X,Y) = \max_{i=1,\dots,n+1}\avar_{\alpha_i}(X+(1-r_i)Y).
\]
\end{prop}
\begin{proof}
See Section \ref{sect: proof parametric gamma avar}.
\end{proof}

As an application of Proposition \ref{prop:elementrerho}, we record some basic properties of $\reavar$ in the next result. We refer to Section \ref{proof:dual representation reavar} in the appendix for a proof of a dual representation of $\reavar$ in line with the general dual representation recorded in Section \ref{sec:dual representation general}. As mentioned above, this type of duality results plays an important role in applications such as optimization problems involving risk measures.

\begin{prop}\label{prop:elementreavar}
The risk measure $\reavar_\gamma$ is cash invariant in its first component, monotone, convex, subadditive, positively homogeneous, star shaped in its first component, and normalized. Moreover, $\reavar_\gamma(X,Y)<\infty$ for all $X,Y\in L^1$ with $Y\geq0$ under any of the following conditions: $\gamma(0)>0$, $\avar_{\gamma(0)}(X)<\infty$, or $X$ is bounded from below.
\end{prop}
\begin{proof}
See Section \ref{proof:elementreavar}.
\end{proof}

\begin{rem}
\label{rem: on subadditivity of reavar}
The subadditivity of $\reavar$ makes it suitable to serve as a basis for limit systems that enable decentralized risk management within firms. We consider a bank or an insurance company that consists of $N$ subentities. For each date $t=0,1$ their assets, liabilities, and net asset value are denoted by $A_t^i$, $L_t^i$, and $E_t^i$, $i=1,\dots,N$. The consolidated figures are denoted by
\[
A_t= \sum_{i=1}^N A_t^i, \ \ \ L_t= \sum_{i=1}^N L_t^i, \ \ \ E_t= \sum_{i=1}^N E_t^i.
\]
The firm may enforce entity-based risk constraints of the form
\[
\reavar_\gamma (E^i_1,L^i_1) \leq c^i, \ \ \ i=1,\dots,N,
\]
where $c^1,\dots,c^N\in\bbr$ are given risk limits. If the limits are chosen to satisfy $\sum_{i=1}^N c^i\leq0$, then
\[
\reavar_\gamma (E_1,L_1) \leq \sum_{i=1}^N\reavar_\gamma (E^i_1,L^i_1) \leq \sum_{i=1}^N c^i \leq 0
\]
by subadditivity. This shows that imposing risk constraints at the level of subentities allows to fulfill the ``global'' solvency test \eqref{reavar_solvency test}. A closely related issue is performance measurement and adaptive management of the balance sheets of firms, as often seen in practice. This is discussed for general recovery risk measures in Section~\ref{sec:performance}.
\end{rem}

\begin{rem}\label{rem:Lreavar}
As in Remark \ref{rem:2nd-oper}, we may construct a version of $\reavar$ that is cash invariant with respect to its second component. This is given by
\[
\Lreavar_\gamma (A_1,L_1) := \sup_{\lambda \in (0,1]} \;  \frac{1}{\lambda}\avar_{\gamma(\lambda)}(A_1-\lambda L_1).
\]
The operational interpretation is analogous to that of $\Lrevar$.
\end{rem}


\section{Applications}
\label{sect: numerics}

We complement the foundations on recovery risk measures with detailed case studies and applications. In Section~\ref{sec:pilh} we demonstrate that recovery-based solvency requirements may help align the decisions of the management of firms with the interest of creditors in protecting their claims in the case of default. In Sections \ref{sec:incs} and \ref{sect: maximal rec adj}  we compare in case studies the standard risk measures adopted in solvency regulation --- $\VaR$ and $\avar$ --- to the recovery risk measure $\revar$. Such a comparison enables to identify and to quantify potential failures of the current regulatory standards by revealing those situations in which creditors are not appropriately protected from low recovery rates on their claims payments. In Section~\ref{ctrf} we address the problem of calibrating the recovery function to pre-specified benchmarks, an issue that is relevant in the context of regulatory regime changes. In Section~\ref{sec:performance} we focus on performance-based management of business divisions of firms for recovery risk measures. Finally, in Section~\ref{sec:optimization} we show that efficient combinations of risk and return can be characterized by a linear program if downside risk is quantified by $\reavar$.


\subsection{Protecting the Interests of Creditors}
\label{sec:pilh}

In Example~\ref{ex:interests}, we demonstrated that capital requirements based on $\var$ and $\ES$ may fail to provide an adequate protection to creditors. We return to this example and show that $\revar$ and $\reavar$ can successfully be employed to enforce guarantees on claims recovery. For detailed calculations we refer to Section~\ref{sec:ex-recovery} in the appendix.

\begin{ex}
\label{ex: recovery VaR}
We consider the situation of Example~\ref{ex:interests}, but with a different risk constraint in terms of $\revar$. While solvency constraints in terms of $\var$ or $\avar$ led to recovery $0$, the recovery risk measure $\revar$ is able to guarantee a pre-specified recovery level.

We fix a recovery function in the class described in Section~\ref{sect: choice of gamma} with $n=1$. For a probability level $\beta\in(0,\alpha)$ and a recovery level $r\in(0,1)$, we set
\[
\gamma(\lambda)=
\begin{cases}
\beta & \mbox{if} \ \lambda\in[0,r),\\
\alpha & \mbox{if} \ \lambda\in[r,1].
\end{cases}
\]
For every choice of $k\in[0,100]$ we obtain from Proposition \ref{prop: revar with piecewise gamma} that
\[
\revar_\gamma(E^k_1,L_1) = \max\{\VaR_\alpha(E^k_1),\VaR_\beta(E^k_1+(1-r)L_1)\}.
\]
A direct computation shows that
\[
\revar_\gamma(E^k_1,L_1) =
\begin{cases}
100r-k & \mbox{if} \ \beta<\frac{\alpha}{2}, \  k\leq50(r+1),\\
k-100 & \mbox{otherwise}.
\end{cases}
\]
According to \eqref{revar_solvency test} the company is adequately capitalized if
\[
\revar_\gamma(E^k_1,L_1)\leq0 \ \iff \
\begin{cases}
k\geq100r & \mbox{if} \ \beta<\frac{\alpha}{2},\\
k\geq0 & \mbox{if} \ \beta\geq\frac{\alpha}{2}.
\end{cases}
\]
A maximal shareholder value under the recovery-based solvency constraint is attained with $k=100r$ when $\beta<\frac{\alpha}{2}$ and with $k=0$ otherwise.  The first case corresponds to successfully controlling recovery. Hence, the regulator may choose a suitable recovery function such that $\revar$ is more stringent than $\VaR$ and the recovery fraction in the default state is equal to $r$. This is in contrast to Example~\ref{ex:interests} with solvency constraints in terms of $\var$ or $\avar$ that led to recovery $0$ when the management maximizes shareholder value.
\end{ex}

\begin{ex}\label{ex: recovery AVaR}
We consider the same situation as in Example~\ref{ex: recovery VaR}, but replace $\revar$ by $\reavar$ with the same recovery function. We will demonstrate that the recovery risk measure $\reavar$ is also able to guarantee a pre-specified recovery level.

To be more specific, it follows from Proposition \ref{prop: parametric gamma avar} for every choice of $k\in[0,100]$  that
\[
\reavar_\gamma(E^k_1,L_1) = \max\{\avar_\alpha(E^k_1),\avar_\beta(E^k_1+(1-r)L_1)\}.
\]
A direct computation shows that
\[
\reavar_\gamma(E^k_1,L_1) =
\begin{cases}
100r-k & \mbox{if} \ \beta<\frac{\alpha}{2}, \ k\leq100r,\\
r-101+\frac{\alpha}{2\beta}(101+99r)+(1-\frac{\alpha}{\beta})k & \mbox{if} \ \beta\geq\frac{\alpha}{2}, \ k\leq\frac{(99\alpha+2\beta)r-101(2\beta-\alpha)}{2(\alpha-\beta)},\\
0 & \mbox{otherwise}.
\end{cases}
\]
Hence, the company is adequately capitalized under \eqref{reavar_solvency test} if
\[
\reavar_\gamma(E^k_1,L_1)\leq0 \ \iff \
\begin{cases}
k\geq100r & \mbox{if} \ \beta<\frac{\alpha}{2},\\
k\geq\max\left\{\frac{(99\alpha+2\beta)r-101(2\beta-\alpha)}{2(\alpha-\beta)},0\right\} & \mbox{if} \ \beta\geq\frac{\alpha}{2}.
\end{cases}
\]
If $\beta<\frac{\alpha}{2}$ and the management selects the individually optimal admissible level of $k$, the recovery fraction in the default state is equal to $r$ as observed in Example~\ref{ex: recovery VaR}. In this case, there is no difference between $\reavar$ and $\revar$.

Interestingly enough, contrary to $\revar$, claims recovery can be controlled under $\reavar$ even in the situation where $\beta\geq\frac{\alpha}{2}$.
In this case, under the assumption that shareholder value is maximized, the fraction of claims recovered in the default state equals
\[
\max\bigg\{\frac{(99\alpha+2\beta)r-101(2\beta-\alpha)}{200(\alpha-\beta)},0\bigg\}.
\]
This expression is strictly positive as soon as $r$ is strictly larger than the bound $\frac{101(2\beta-\alpha)}{99\alpha+2\beta} \in[0,1)$. (For example, taking $\beta=\frac{\alpha}{2}$ always ensures a recovery equal to $r$). As claimed, solvency capital requirements based on $\reavar$ are more effective in controlling claims recovery in comparison to those based on $\revar$ in Example~\ref{ex: recovery VaR}
\end{ex}


\subsection{The Impact of the Distribution of the Balance Sheet}

Standard solvency capital requirements based on $\VaR$ and $\avar$ cannot control the probability of recovering certain pre-specified fractions of claims. Additional capital is required which needs to be computed on the basis of recovery risk measures such as  $\revar$ and $\reavar$. In this section, we study the impact of a variation in the distribution of the underlying balance sheet figures on the size of necessary capital adjustments. Section~\ref{sec:incs} numerically illustrates this for standard parametric distributions. This allows to understand the influence of correlation between assets and liabilities and the tail size of liabilities. Section \ref{sect: maximal rec adj} presents a stylized example demonstrating that  under standard solvency regimes sophisticated asset-liability-management may hide substantial tail risk. These situations correspond to high capital adjustments, if the required capital is instead computed by recovery risk measures.

Throughout the section, we consider a financial institution with assets $A_t$, liabilities $L_t$, and net asset value $E_t=A_t-L_t$ at dates $t=0,1$. The changes of the net asset value over the considered time window or, equivalently, the corresponding cash flows are $\Delta E_1 = E_1 - E_0$.


\subsubsection{Parametric Distributions}
\label{sec:incs}

In this section, we consider parametric distributions that model the evolution of the company's assets and liabilities and show how the gap between standard capital requirements and those based on recovery risk measures is influenced by the dependence between assets and liabilities and by their marginal distributions, in particular the liability tail size.  We refer to Section \ref{sect: complementary numerical study} for further details and to Section \ref{sect: plots appendix} for several complementary plots.

\paragraph{Distribution of assets and liabilities.}  We assume that $A_1$ possesses a lognormal distribution with log-mean $\mu\in\R$ and log-standard deviation $\sigma>0$. This specification for the asset distribution is standard in the finance literature and compatible, e.g., with the Black-Scholes setting. We fix $\mu=2$ and $\sigma=0.2$. Liabilities $L_1$ follow a mixture gamma distribution. More precisely, up to the $95\%$ quantile $L_1$ possesses a gamma distribution with shape parameter  $\tau_0>0$ and rate parameter $\delta_0>0$; beyond the $95\%$ quantile $L_1$ is determined by a gamma distribution with shape parameter $\tau>0$ and rate parameter $\delta>0$. This specification is encountered in many applications, including insurance, and allows a flexible control on the tail distribution (heavier tails correspond to higher levels of $\tau$). Setting $\delta_0=\delta=1$ and $\tau_0=1$, we focus on the range $\tau\in[1,5]$.

Assets and liabilities are linked by a Gaussian copula. This choice allows to capture dependence by a single parameter, the correlation coefficient $\rho\in[-1,1]$. Under positive dependence ($\rho>0$), shocks increasing the value of liabilities are more frequently accompanied by increased asset values. In this case, the asset position may be considered a reasonable hedge of the liability position. We focus on the range $\rho\in[0,1]$.

\begin{figure}[t]
\vspace{-0.8cm}
\centering
\subfigure{
\includegraphics[width=0.42\textwidth]{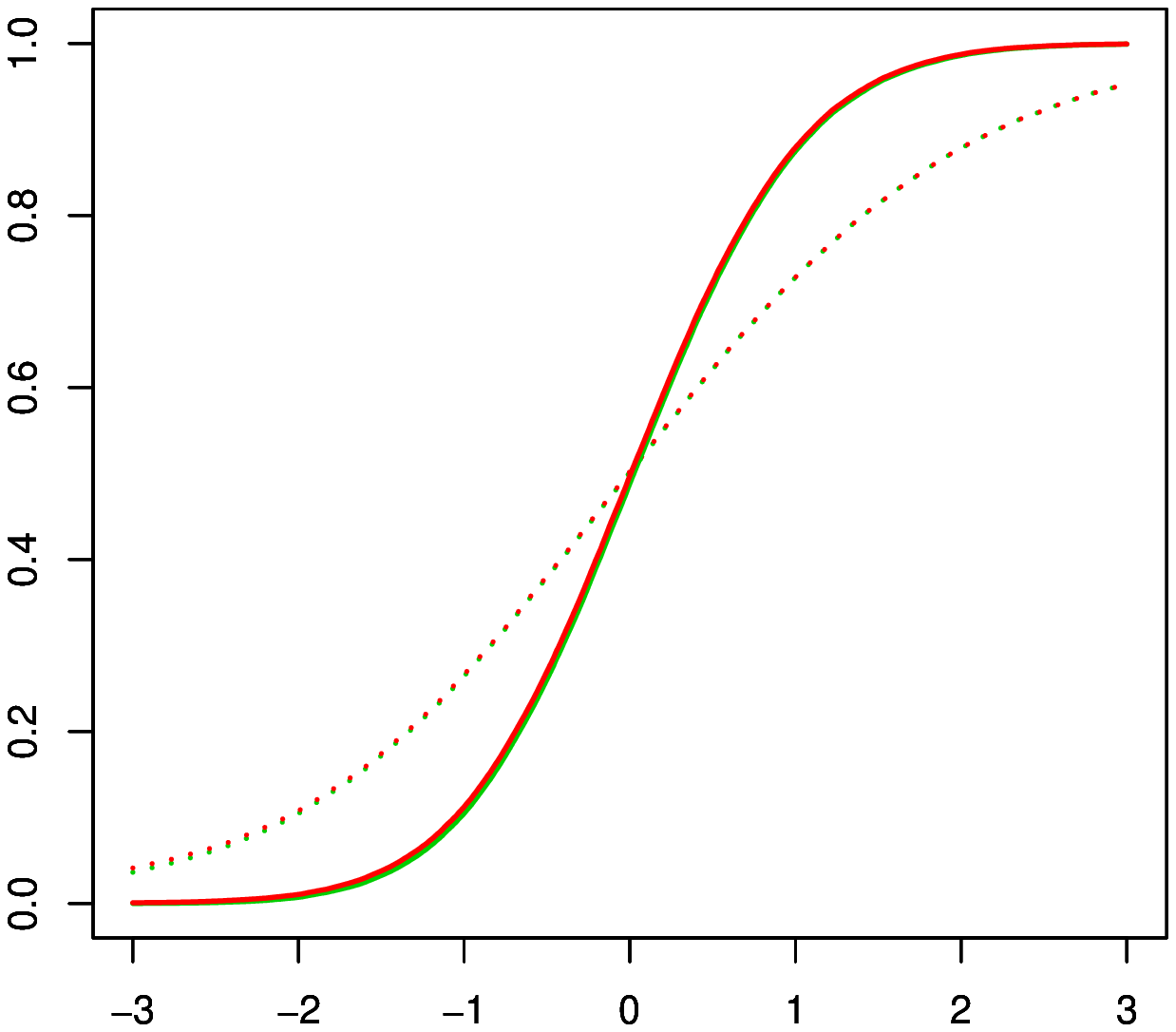}
}
\hspace{1cm}
\subfigure{
\includegraphics[width=0.42\textwidth]{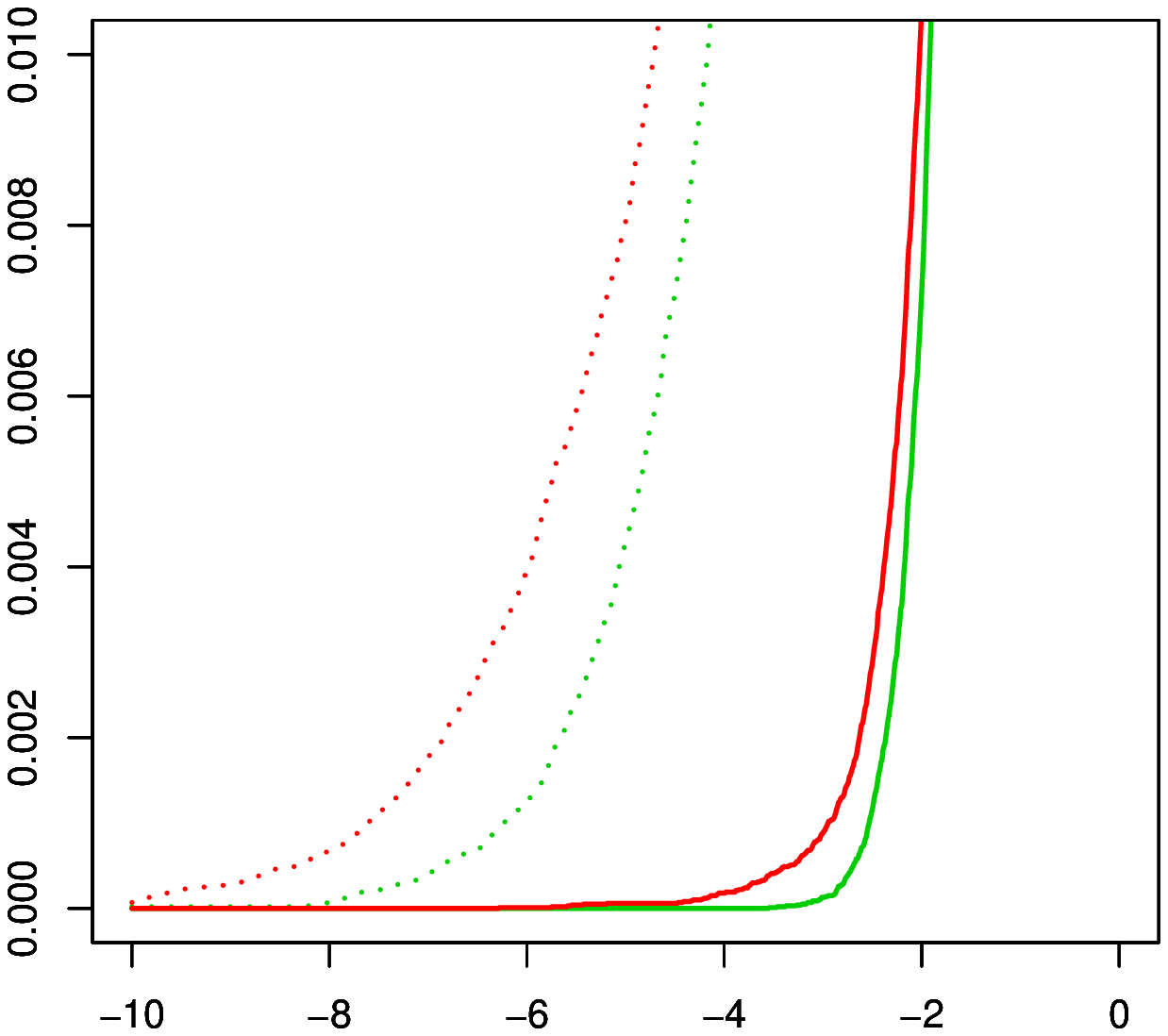}
}
\vspace{-0.5cm}
\caption{Probability distribution function of $\Delta E_1$ (left) and a detail of its tail (right) for $\rho=0.1$ (dotted) and $\rho=0.9$ (plain) and for $\tau=1$ (green) and $\tau=5$ (red).}
\label{fig: distribution equity}
\end{figure}

\paragraph{Simulated distribution of assets and liabilities.}  Our computations are implemented using the software R. We resort to standard Monte Carlo simulation based on quantile inversion (for the marginal distributions) and Cholesky decomposition (for the joint distribution); see, e.g., \ci{glasserman2013}. The simulated cash flow distribution is displayed in Figure \ref{fig: distribution equity}. The choice of $E_0$ is made to ensure a realistic probability of observing negative cash flows over the considered period of time, i.e., $\probp(E_1<E_0)$; we target a value of about $50\%$. This constraint is met in our case if, e.g., $E_0=6.5$. As expected, increasing the correlation level between assets and liabilities leads to a more concentrated cash flow distribution. Increasing the size of the liability tail leads to a heavier cash flow tail. Probabilities of negative cash flows are decreasing functions of the correlation level and increasing functions of the liability tail size.\footnote{This is illustrated in Figure~\ref{fig: default probabilities} in the appendix.} The first observation is due to the fact that the more positive the dependence, the more effective the assets as a hedge against liabilities and the lower the probability of negative cash flows.

\paragraph{Regulatory capital requirements.}  We focus on the two most prominent solvency regimes in insurance,  Solvency II and the Swiss Solvency Test, with regulatory capital requirements
\[
\rho_{reg}(\Delta E_0)=
\begin{cases}
\VaR_{0.5\%}(\Delta E_0) & \mbox{under Solvency II},\\
\ES_{1\%}(\Delta E_0) & \mbox{under the Swiss Solvency Test}.
\end{cases}
\]
\noindent Figure~\ref{fig: regulatory standards} displays solvency capital requirements as functions of the correlation level between assets and liabilities and of the liability tail size. In line with our previous discussion, the level of regulatory capital is a decreasing function of correlation and an increasing function of tail size. The risk measure $\rho_{reg}(E_1)$ is always negative under our specifications, indicating that we are focusing on companies that are technically solvent with respect to the regulatory solvency tests under consideration. In addition, the solvency ratio $\frac{E_0}{\rho_{reg}(\Delta E_0)}$ lies in the interval $[1,3]$, which is the relevant range in practice.\footnote{This is illustrated in Figure \ref{fig: regulatory risk measures} and Figure~\ref{fig: solvency ratio} in the appendix.}

\begin{figure}[t]
\vspace{-0.8cm}
\centering
\subfigure{
\includegraphics[width=0.42\textwidth]{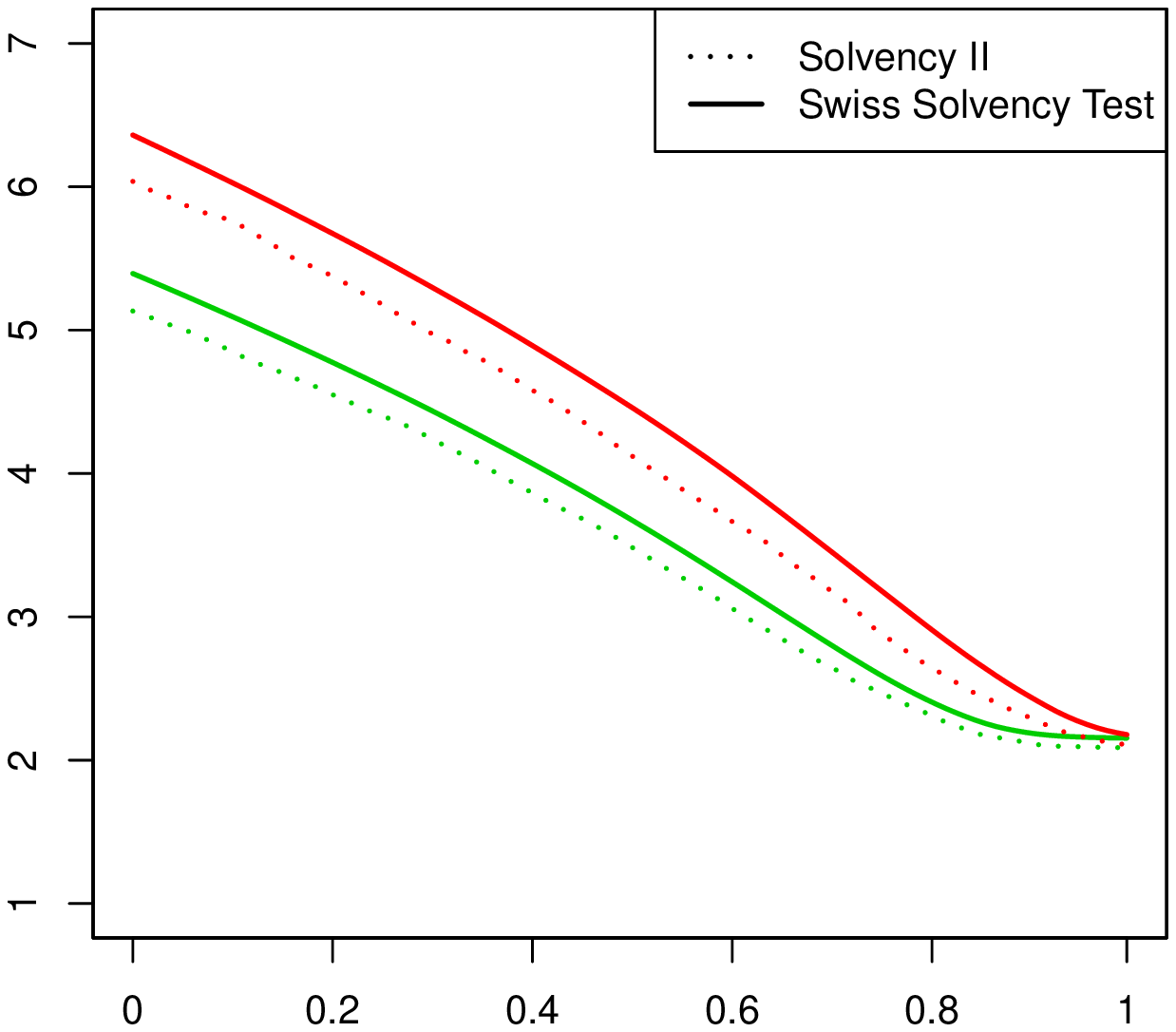}
}
\hspace{1cm}
\subfigure{
\includegraphics[width=0.42\textwidth]{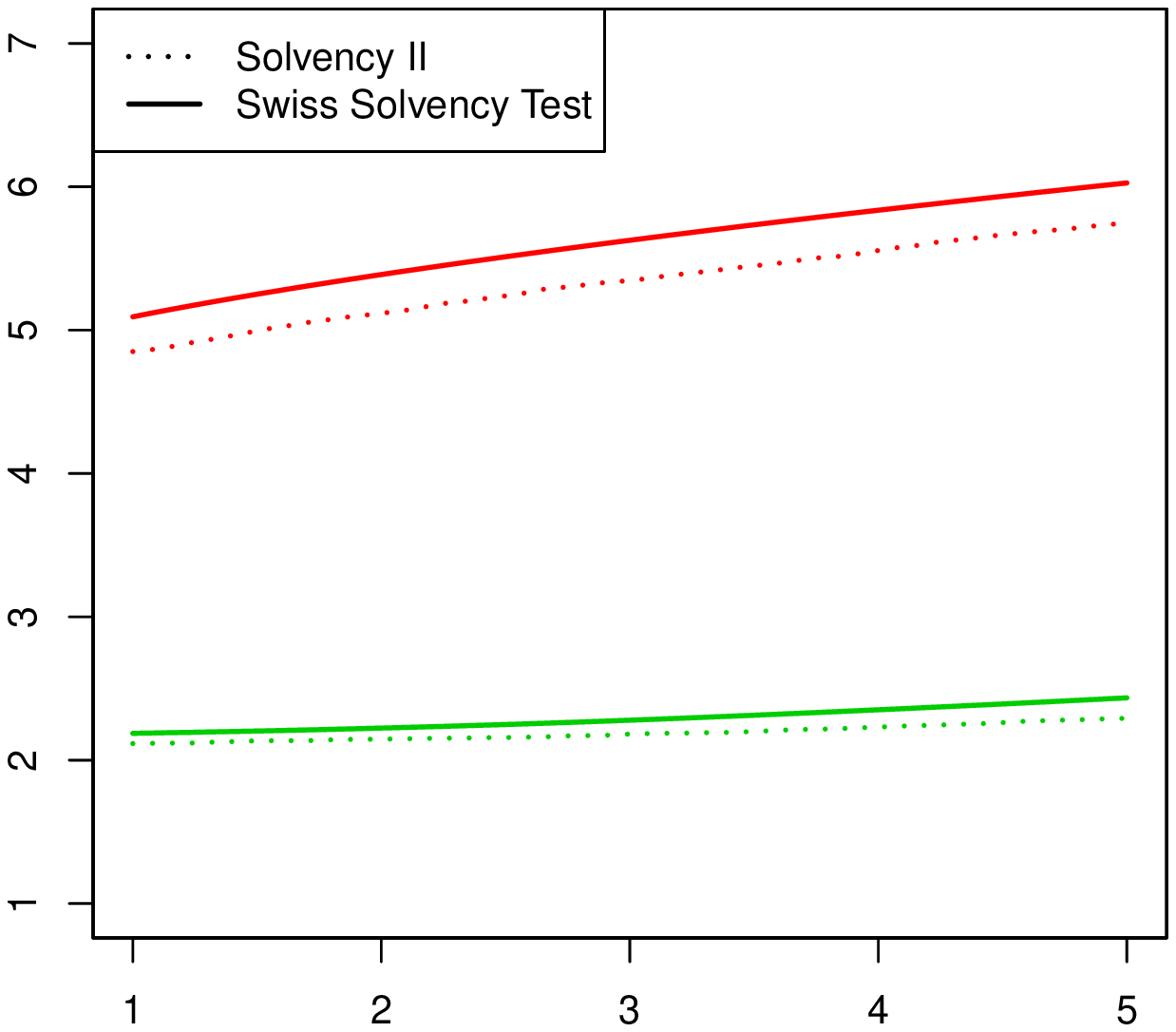}
}
\vspace{-0.5cm}
\caption{The solvency capital requirement $\rho_{reg}(\Delta E_1)$ as a function of $\rho$ (left) for $\tau=1$ (green) and $\tau=5$ (red) and as a function of $\tau$ (right) for $\rho=0.1$ (red) and $\rho=0.9$ (green).}
\label{fig: regulatory standards}
\end{figure}

\paragraph{Recovery-based capital requirements.} For comparison, we consider recovery risk measures with a simple parametric recovery function belonging to the class described in Section~\ref{sect: choice of gamma} with $n=1$. We fix a regulatory level $\alpha\in(0,1)$ and consider a piecewise constant recovery function
\begin{equation}
\gamma(\lambda)=
\begin{cases}
\beta & \mbox{if} \ \lambda\in[0,r)\\
\alpha & \mbox{if} \ \lambda\in[r,1]
\end{cases}
\end{equation}
for suitable $\beta\in(0,\alpha)$ and $r\in(0,1)$. In line with the Solvency II standards we take $\alpha=0.5\%$. The solvency capital requirement induced by the corresponding $\revar$ is given by
\[
\reVaR_\gamma(\Delta E_1,L_1) = \max\{\VaR_\alpha(\Delta E_1),\VaR_\beta(\Delta E_1+(1-r)L_1)\},
\]
the solvency test \eqref{revar_solvency test} is equivalent to
\[
\probp(A_1<L_1)\leq\alpha \ \ \ \mbox{and} \ \ \ \probp(A_1<rL_1)\leq\beta.
\]
Besides controlling the default probability $\probp(A_1<L_1)$ at the pre-specified regulatory level $\alpha$, the recovery risk measure additionally bounds the probability $\probp(A_1<rL_1)$ of covering less than a fraction $r$ of liabilities by a more stringent level $\beta$.

\paragraph{Recovery adjustments.}  In order to capture the extent to which the regulatory solvency capital requirements fail to control the recovery on liabilities we define the {\em recovery adjustment}
\begin{equation}\label{eq:rec_adj}
\RecAdj_\gamma(\Delta E_1,L_1) := \max\left\{\frac{\reVaR_\gamma(\Delta E_1,L_1)}{\rho_{reg}(\Delta E_1)},1\right\}.
\end{equation}
This quantity is the maximum of 1 and the multiplicative factor by which regulatory requirements would have to be adjusted to guarantee the considered recovery levels. Recovery adjustments may also be conveniently expressed as a function of the regulatory level $\beta$ and the recovery rate $r$ as
\[
\RecAdj(\beta,r) := \max\left\{\frac{\max\{\VaR_\alpha(\Delta E_1),\VaR_\beta(\Delta E_1+(1-r)L_1)\}}{\rho_{reg}(\Delta E_1)},1\right\}.
\]
\noindent
 We consider the {\em aggregate recovery adjustment}\footnote{We refer to Section \ref{sect: plots appendix} in the appendix for several plots of recovery adjustments for specific choices of $\beta$ and $r$. These confirm the findings that are observed on an aggregate level.}
\[
\AggRecAdj := \int_{\beta_{min}}^{\beta_{max}}\int_{r_{min}}^{r_{max}}\RecAdj(\beta,r)d\beta dr
\]
with $(\beta_{min},\beta_{max})=(0.1\%,0.25\%)$ and $(r_{min},r_{max})=(80\%,90\%)$. Apart from a normalization constant, this quantity corresponds to the average recovery adjustment over the chosen range of the recovery parameters $\beta$ and $r$. Figure~\ref{fig: aggregate recovery} displays the aggregate recovery adjustment as a function of the correlation between assets and liabilities and of the size of the liability tail.

\begin{figure}[t]
\vspace{-0.8cm}
\centering
\subfigure{
\includegraphics[width=0.42\textwidth]{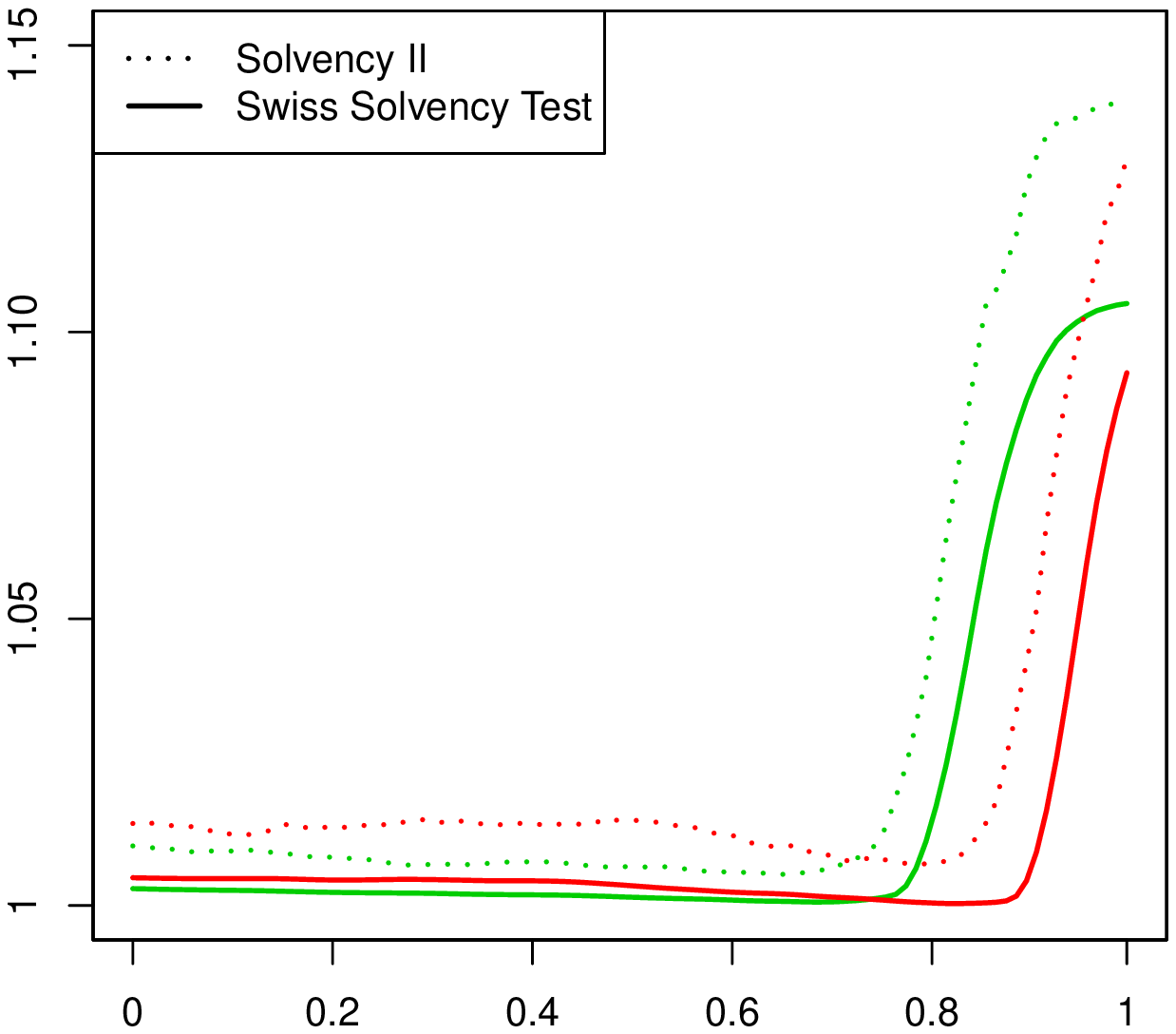}
}
\hspace{1cm}
\subfigure{
\includegraphics[width=0.42\textwidth]{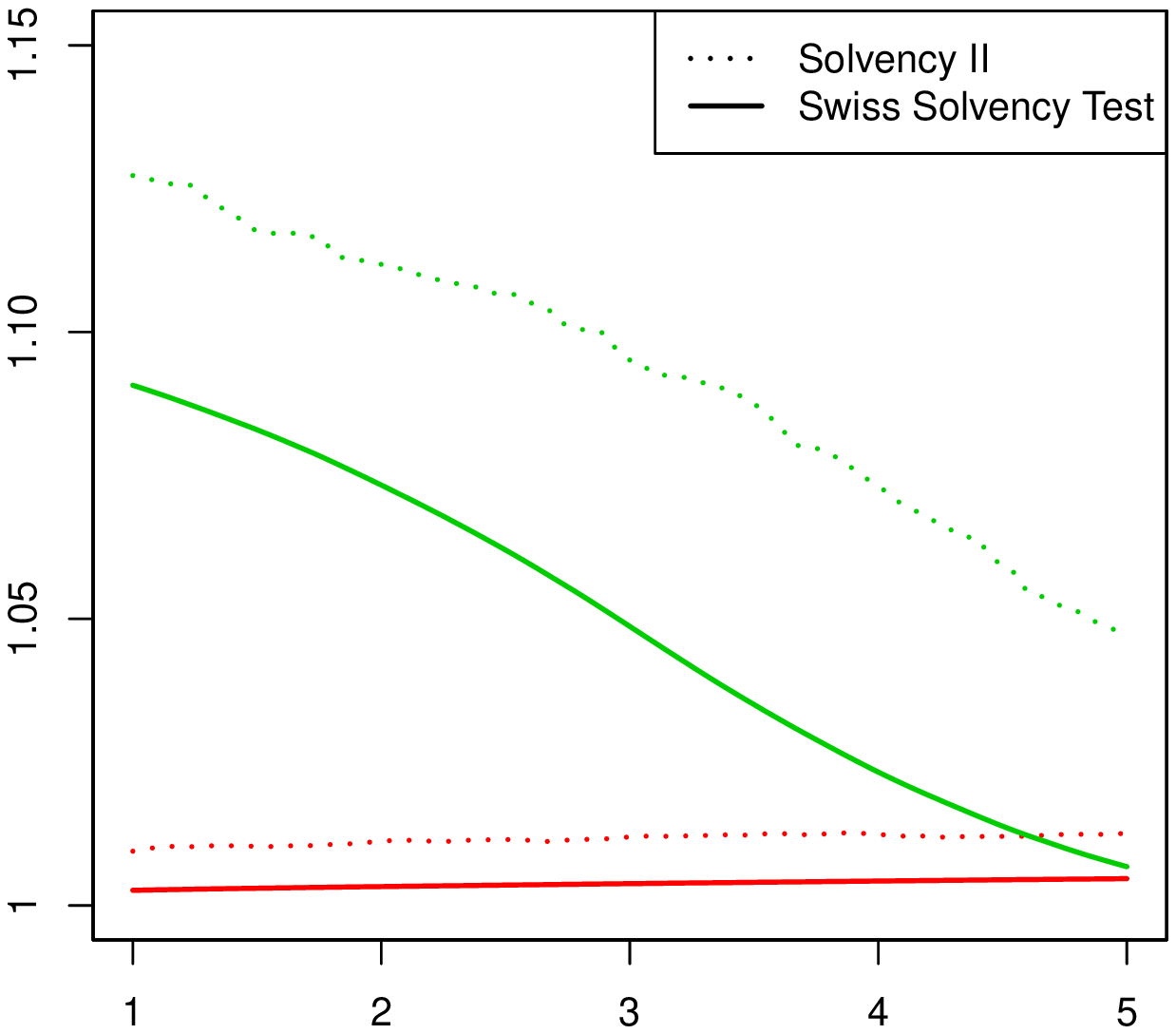}
}
\vspace{-0.5cm}
\caption{The aggregate recovery adjustment $\AggRecAdj$ as a function of $\rho$ (left) for $\tau=1$ (green) and $\tau=5$ (red) and as a function of $\tau$ (right) for $\rho=0.1$ (red) and $\rho=0.9$ (green).}
\label{fig: aggregate recovery}
\end{figure}

\paragraph{Observations.} The qualitative behavior of recovery adjustments can be described as follows:
\begin{enumerate}
    \item Recovery adjustments are typically larger than $1$, indicating that regulatory solvency requirements are too low to fulfill the target recovery-based solvency condition.
    \item The size of the recovery adjustment depends on the recovery-based regulatory level $\beta$ and the recovery rate $r$ as expected: It is larger if $\beta$ is lower (a tighter constraint on the recovery probability) and if $r$ is higher (a larger portion of liabilities to recover). This relation holds across liability tail sizes and correlation levels but is more pronounced in the presence of lighter liability tails and higher correlations.
    \item For sufficiently large correlation levels, recovery adjustments are increasing functions of the correlation level between assets and liabilities, suggesting that the failure of regulatory capital requirements to control the recovery on liabilities is more pronounced in the presence of large correlation levels. This relation holds across all liability tail sizes.
    \item For sufficiently light tails, recovery adjustments are decreasing functions of the liability tail size, suggesting that the failure of regulatory capital requirements to control the recovery on liabilities is stronger in the presence of lighter liability tails. This relation holds across different correlation levels.
    \item These observations hold for both regulatory frameworks under investigation. In comparison, recovery adjustments in the Swiss Solvency Test are lower than those induced in Solvency II under our distributional specifications.
\end{enumerate}
Our observations demonstrate the importance of recovery risk measures from a risk management perspective. First, we observe that, under standard distributional assumptions, there may exist a considerable gap between the standard risk measures used in practice and our reference recovery risk measure. Second, this gap tends to be wider in the presence of lighter liability tails and higher levels of correlation between assets and liabilities. In situations when assets appear to better hedge liability claims, standard risk measures show lower ability to control recovery rates.

\subsubsection{Sophisticated Asset-Liability-Management}
\label{sect: maximal rec adj}

In this section we consider a firm with a stylized balance sheet. Assets are deterministic, but the firm is capable of controlling the shape of the liability distribution in a sophisticated way. Our case study provides another perspective on the failure of standard solvency regulation to control recovery and highlights that this deficiency might be associated with large recovery adjustments. For detailed calculations we refer to Section \ref{sect: preliminary supplementary section maximal rec adj}.

\paragraph{Distribution of assets and liabilities.}  We assume that assets evolve in a deterministic way with $A_1$ being equal to a constant $k>0$. The future value of liabilities $L_1$ follows a probability density function with two peaks as displayed in Figure \ref{fig: peaked density}. The probability of falling in the light tail peak is equal to $99.5\%$ and that of falling in the heavy tail peak is equal to $0.5\%$.\footnote{The explicit expression of the density function is provided in Section \ref{sect: preliminary supplementary section maximal rec adj} in the appendix.}

\paragraph{Regulatory capital requirements.}  In line with Solvency II and the Swiss Solvency Test, we focus on VaR at level $0.5\%$ and AVaR at level $1\%$. The chosen regulatory risk measure is denoted by $\rho_{reg}$. The corresponding solvency capital requirements admit analytic solutions:
\[
\rho_{reg}(\Delta E_1)=
\begin{cases}
\VaR_{0.5\%}(\Delta E_1) = a-k+E_0,\\
\ES_{1\%}(\Delta E_1) = \Big(\frac{1}{2}-\frac{1}{3}\sqrt{\frac{\alpha}{2(1-\alpha)}}\Big)a+\frac{b+c}{4}-k+E_0.
\end{cases}
\]
The capital requirements based on $\VaR$ are  blind to all liability payments beyond the first peak while capital requirements based on $\avar$ react to the entire distribution of liabilities.

\paragraph{Recovery-based capital requirements.} Fixing a regulatory level $\alpha\in(0,1)$,  consider a piecewise constant recovery function belonging to the class described in Section~\ref{sect: choice of gamma} with $n=1$:
\begin{equation}\label{eq:gammafctn}
\gamma(\lambda)=
\begin{cases}
\beta & \mbox{if} \ \lambda\in[0,r)\\
\alpha & \mbox{if} \ \lambda\in[r,1]
\end{cases}
\end{equation}
with $\beta\in(0,\alpha)$, $r\in(0,1)$, and $\alpha=0.5\%$. The choice of $\alpha$ is motivated by the standards implemented in Solvency II.
The solvency capital requirement corresponding to $\revar$ is
\[
\reVaR_\gamma(\Delta E_1,L_1) = \max\{\VaR_\alpha(\Delta E_1),\VaR_\beta(\Delta E_1+(1-r)L_1)\} = \max\left\{a,r\frac{b+c}{2}\right\}-k+E_0
\]
and depends on the entire distribution of liabilities. This paralles $\avar$, but tail risk is captured in a more sophisticated way: The recovery level $r$ determines the relative importance of the peaks of the liability distribution.

\begin{figure}[t]
\vspace{-0.8cm}
\centering
\includegraphics[width=0.4\textwidth]{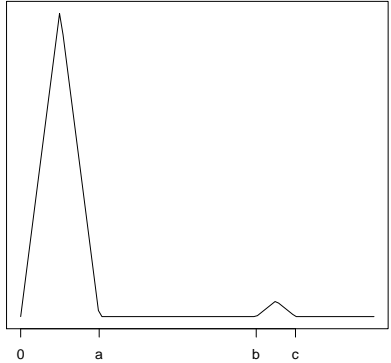}
\vspace{-0.5cm}
\caption{Qualitative plot of the probability density function of $L_1$. The area below the left peak equals $99.5\%$ while the area below the right peak equals $0.5\%$.}
\label{fig: peaked density}
\end{figure}

\paragraph{Recovery adjustments.}
As in Section~\ref{sec:incs} we consider recovery adjustments as introduced in \eqref{eq:rec_adj}. We will answer the question how large the recovery adjustments may become, if a firm's asset-liability-management is constrained by the following conditions:
\begin{center}
\begin{tabular}{|c|l|c|}
\hline
 (1) & Solvent profile under $\rho_{reg}$ & $\rho_{reg}(E_1)\leq0$ \\ \hline
 (2) & Capital requirement under $\rho_{reg}$  & $\rho_{reg}(\Delta E_1)>0$  \\ \hline
 (3) & Solvent profile under $\revar_\gamma$ & $\revar_\gamma(E_1,L_1)\leq0$  \\ \hline
 (4) & Capital requirement under $\revar_\gamma$ & $\revar_\gamma(\Delta E_1,L_1)>0$  \\ \hline
(5) & $\VaR_\alpha$ insufficient to control claims recovery &  $\revar_\gamma(E_1,L_1)>\VaR_\alpha(E_1)$  \\ \hline
(6) & Range of admissible regulatory solvency ratios & $s_{min}\leq\frac{E_0}{\rho_{reg}(\Delta E_1)}\leq s_{max}$\\
\hline
\end{tabular}
\end{center}
More precisely, we focus on the optimization problem
\begin{equation}\label{eq: optimal rec adj}
\max \; \RecAdj_\gamma(\Delta E_1,L_1)
\mbox{ \ \ over~}
 A_1 \ \mbox{and} \ L_1 \ \mbox{as specified above}
 \end{equation}
under the constraints (1) to (6). The solvency ratios in (6) are in practice typically in the range between $1.2$ and $3$.\footnote{In general, we assume that $1<s_{min}<s_{max}$.}

It turns out that in many cases the highest admissible recovery adjustment  coincides with $s_{\max}$, i.e., capital requirements under $\revar$ can be as large as $s_{\max}$ times the capital requirements under $\VaR$ or $\avar$. This implies that the difference between the current capital requirements and their recovery-based versions in our asset-liability setting may be substantial. The next proposition states this in detail.

\begin{prop}
\label{prop: optimal rec adj}
The optimal value of \eqref{eq: optimal rec adj} is bounded from above by
$s_{max}.$
If $\rho_{reg}=\VaR_{0.5\%}$, this upper bound is attained for every choice of $\gamma$. If $\rho_{reg}=\avar_{1\%}$, this upper bound is attained for special choices of $\gamma$, e.g., when $\beta\geq\frac{\alpha}{2}$ and
$
\frac{1}{4}\sqrt{\frac{2\alpha}{\alpha-\beta}} < r \leq \frac{1}{4}\frac{\sqrt{2\alpha}}{\sqrt{2\alpha}-\sqrt{\alpha-\beta}}\frac{1}{\frac{1}{2}+\frac{1}{3}\sqrt{\frac{\alpha}{2(1-\alpha)}}}.
$\footnote{For instance, if $\beta=\frac{\alpha}{2}$, then we can take $50\%<r\leq95\%$.}
\end{prop}
\begin{proof}
See Section \ref{sect: supplementary section maximal rec adj}.
\end{proof}

The preceding result suggests that companies subject to capital requirements based on $\VaR$ and $\avar$ may be far from guaranteeing acceptable recovery rates on their creditors' claims. In the $\VaR$ case, this is a consequence of tail blindness, which, in the absence of external controls, allows companies to accumulate tail risk without any regulatory cost. Also in the $\avar$ case this problem does not disappear because increased tail risk may often be compensated by a suitable shift in the asset distribution or in the body of the liability distribution. In our example, a more dispersed distribution beyond the $99.5\%$ quantile may leave the $\avar$ unchanged provided the distribution within the same quantile level shrinks.\footnote{This phenomenon refers to the lack of {\em surplus invariance} and was studied in \ci{koch2016unexpected}.}


\subsection{Calibrating the Recovery Function}\label{ctrf}

When regulatory solvency standards in practice are modified and improved, the old regulatory framework is often used as a benchmark for the new one. New and old requirements will, of course, differ for many distributions of assets and liabilities at the considered time horizon, and many companies might experience corresponding changes in solvency requirements. For this reason, a common approach in practice is to calibrate the new standards in such a way that they produce the same solvency requirement for a prototypical benchmark company. The rational behind this strategy is to ensure some form of continuity in the sense that benchmark firms are not too much affected over short time horizons. At the same time, regulatory standards are ideally modified in such a way that their new design is more efficient in achieving key regulatory goals in the long run. Choosing a benchmark balance sheet for calibration naturally remains a political decision.

In this section, we explain in the context of an example how recovery risk measures could be calibrated to existing regulatory standards. We begin by recalling the transition from Basel II to Basel III. Basel II was based on $\VaR$ at level $\alpha=1\%$ while the new Basel III has adopted $\avar$ at level $\beta=2.5\%$. The choice of $\beta=2.5\%$ was justified a) by assuming in a benchmark model that changes in net asset value $\Delta E_1$ are normally distributed and b) by requiring\footnote{We refer to \ci{PELVEwang} for a general study on calibration of $\var$ and $\avar$.}
$$
\ES_\beta(\Delta E_1) \approx \VaR_\alpha(\Delta E_1).
$$
The new regulatory level equates capital requirements of normally distributed positions for old and new standards. In the same spirit, we describe how to calibrate the recovery level function of $\revar$. A challenge is that we deal with a function $\gamma$ instead of a single parameter $\beta$ as well as with a pair of random variables, $E_1-E_0$ and $L_1$, instead of just one random variable, $E_1-E_0$. The aim is to equate, for a given level $\alpha\in(0,1)$ close to zero in the case that $\Delta E_1$ is normally distributed,
\begin{equation}\label{eq:goal}
\reVaR_\gamma(\Delta E_1,L_1) = \VaR_\alpha(\Delta E_1).
\end{equation}

\noindent {\bf Definition of a benchmark.} \ The choice of a benchmark is a political decision of the regulator. In our recovery-based setting, we consider a particularly simple choice. As discussed before, we assume that $\Delta E:=E_1-E_0$ is normally distributed with mean $\mu_{\Delta E}\in\R$ and standard deviation $\sigma_{\Delta E}>0$. In addition, we suppose that $L:=L_1$ is normally distributed with mean $\mu_L\in\R$ and standard deviation $\sigma_L>0$ and is independent of $\Delta E$. The latter assumption is not meant to capture realistic balance sheets, but is simply chosen for illustration as it leads to explicit calculations. By independence, for every $\lambda\in[0,1]$ the random variable $\Delta E+(1-\lambda)L$ is also normal with mean $\mu_{\Delta E}+(1-\lambda)\mu_L$ and standard deviation $\sqrt{\sigma^2_{\Delta E}+(1-\lambda)^2\sigma^2_L}$. We denote by $\Phi$ the distribution function of a standard normal random variable. (Note that a positive random variable like $L$ cannot have a normal distribution. In practice, this can be taken into account by imposing the condition $\Phi(-\frac{\mu_L}{\sigma_L}) = \probp(L<0)\leq\varepsilon$ for a sufficiently small $\varepsilon>0$).

\noindent {\bf Deriving the level function.} \ We seek a function $\gamma$ such that \eqref{eq:goal} holds. A sufficient requirement is that $\VaR_{\gamma(\lambda)}(\Delta E+(1-\lambda)L) = \VaR_\alpha(\Delta E)$ for every $\lambda\in[0,1]$, or equivalently
\[
-(1-\lambda)\mu_L-\sqrt{\sigma^2_{\Delta E}+(1-\lambda)^2\sigma^2_L}\Phi^{-1}(\gamma(\lambda)) =
-\sigma_{\Delta E}\Phi^{-1}(\alpha), \quad  \lambda\in[0,1].
\]
Solving for $\gamma$ gives
\[
\gamma(\lambda) =  \Phi\left(\frac{\sigma_{\Delta E}\Phi^{-1}(\alpha)-(1-\lambda)\mu_L}{\sqrt{\sigma^2_{\Delta E}+(1-\lambda)^2\sigma^2_L}}\right), \quad  \lambda\in[0,1],
\]
with $\gamma(1)=\alpha$. This function might be inconsistent with the requirements in Definition~\ref{def:revar}, namely with $\gamma$ being increasing. A potential remedy to this problem could be to modify the choice of $\gamma$ as follows. Since our solution for $\gamma$ is differentiable with respect to $\lambda$, by taking derivatives, its increasing part can easily be characterized by \[
\gamma'(\lambda)\geq0 \ \iff \ (1-\lambda)\Phi^{-1}(\alpha)\sigma_L^2+\mu_L\sigma_{\Delta E}\geq0, \quad \lambda\in(0,1).
\]
Here, we have used that $\Phi^{-1}(\alpha)<0$ because $\alpha$ is assumed to be close to zero. Hence, $\gamma$ is increasing on the interval $[\lambda^\ast,1]$ where
\[
\lambda^\ast := \max\left\{1+\frac{\mu_L\sigma_{\Delta E}}{\sigma_L^2\Phi^{-1}(\alpha)},0\right\} < 1.
\]
If $\lambda^\ast=0$, the function $\gamma$ is increasing on the whole interval $[0,1]$. Otherwise, we compute
\[
\gamma(\lambda^\ast) = \Phi\left(-\sqrt{\Phi^{-1}(\alpha)^2+\frac{\mu_L^2}{\sigma_L^2}}\right).
\]
and redefine $\gamma$ as follows:
\[
\gamma(\lambda) :=
\begin{cases}
\Phi\left(-\sqrt{\Phi^{-1}(\alpha)^2+\frac{\mu_L^2}{\sigma_L^2}}\right) & \mbox{if} \ \lambda\in[0,\lambda^\ast),\\
\Phi\left(\frac{\sigma_{\Delta E}\Phi^{-1}(\alpha)-(1-\lambda)\mu_L}{\sqrt{\sigma^2_{\Delta E}+(1-\lambda)^2\sigma^2_L}}\right) & \mbox{if} \ \lambda\in[\lambda^\ast,1].
\end{cases}
\]
\noindent Finally, if a piecewise-constant recovery function is sought, see Section \ref{sect: choice of gamma}, a suitable approximation of $\gamma$ may be chosen. This provides a strategy to calibrate recovery risk measures to existing regulatory standards.


\subsection{Performance Measurement}\label{sec:performance}

An important issue that is closely related to solvency capital requirements is performance measurement. We show that this task can be implemented on the basis of recovery risk measures. A popular metric in practice is the \emph{return on risk-adjusted capital} (RoRaC). Consider a recovery risk measure $\rerho$ that is subadditive and positively homogeneous as defined in Section~\ref{sec:mathstruc}. The associated RoRaC is defined by
$$
{\rm RoRaC} (\Delta E_1, L_1) \; := \; \frac{\E(\Delta E_1)}{\rerho(\Delta E_1, L_1)}.
$$
This quantity measures the \emph{expected return} \emph{per unit of economic capital} expressed in terms of the risk measure $\rerho$. The goal of the firm is to improve its RoRaC. We assume that the company is composed of different subentities labelled $i=1,\dots, n$. The central management may impose risk limits and adjust the size of different business units. A key question is which allocation of economic capital to business units and corresponding performance measurements provide appropriate information to improve the overall performance of the firm. We denote the net asset values and liabilities of the subentities at time $t=0,1$ by $E^i_t$ and $L^i_t$ for $i=1,\dots,N$, respectively. Note that for $t=0,1$
\[
L_t= \sum_{i=1}^N L_t^i, \ \ \ E_t= \sum_{i=1}^N E_t^i.
\]
The subadditivity of $\rerho$ implies
$$
\rerho (\Delta E_1, L_1 )  \;   \leq   \;  \sum _{i=1}^N \rerho (\Delta E^i_1, L^i_1 ).
$$
We seek an allocation of economic capital $\kappa^i:= \rerho^{\Delta E_1, L_1 } (\Delta E ^i_1, L^i_1 )$, $i=1,\dots,N$, satisfying:
\begin{itemize}
\item {\em Full allocation}: $\sum_{i=1}^N \kappa^i = \rerho (\Delta E_1, L_1 )$;
\item {\em Diversification}: $\kappa^i \leq \rerho (\Delta E^i_1, L^i_1 )$ for all $i= 1,\dots, N$;
\item {\em RoRaC-compatibility}: If for some $i=1,\dots,N$ we have
\[
{\rm RoRaC}^i := \frac{\E(\Delta E^i_1)}{\kappa^i} \; > \; {\rm RoRaC} (\Delta E_1,L_1) \quad  (\mbox{resp. } <),
\]
then there exists $\varepsilon >0$ such that for every $h\in(0,\varepsilon)$
\[
{\rm RoRaC}(\Delta E_1+h\Delta E^i_1,L_1+h L^i_1) \; > \; {\rm RoRaC}(\Delta E_1,L_1)   \quad  (\mbox{resp. } <).
\]
\end{itemize}

\noindent The full allocation property requires that the entire solvency capital is allocated to the individual subentities. The diversification property specifies that no more capital is allocated to the individual subentities than their stand-alone solvency capital, taking beneficial diversification effects into account, which is feasible due to the subadditivity of $\rerho$. Finally, RoRaC-compatibility guarantees that performance measurement based on the chosen capital allocation provides the correct information to the management of the firm to improve the overall performance of the firm. To be more precise, if the performance of subentity $i$ --- as captured by ${\rm RoRaC}^i$ --- is better than the overall RoRaC, the performance of the entire firm can be improved by growing subentity $i$. An allocation fulfilling the above three properties is called a \emph{suitable allocation}.

The existence of suitable allocations has been extensively studied in the literature, see, e.g., \ci{tasche1999risk}, \ci{tasche2004allocating}, \ci{kalkbrener2005axiomatic}, \ci{tasche2007capital}, \ci{dhaene2012optimal}, \ci{bauer2013capital}, and the general review by \ci{bauer2020}. It follows from the general results in these papers that the only suitable allocation in the above sense is the {\em Euler allocation}
$$
\kappa^i  \; = \;  \frac{d}{dh} \rerho(\Delta E_1+h\Delta E^i_1,L_1+h L^i_1)_{|_{h=0}}
$$
for all $i=1,\dots, N$. In the specific case of $\rerho=\reavar$ with a simple piecewise constant level function, the Euler allocation can explicitly be computed.

\begin{prop}\label{prop:euler for reavar}
Let $\gamma$ be defined as in \eqref{eq: parametric gamma} and let $j\in\{1,\dots,N\}$ satisfy
\[
\avar_{\alpha_j}(\Delta E_1+(1-r_j)L_1) > \max_{i=1,\dots,N,\,\,i\neq j}\avar_{\alpha_i}(\Delta E_1+(1-r_i)L_1).
\]
Under suitable assumptions on the joint distribution of $\Delta E_1^1+(1-r_j)L_1^1,\dots,\Delta E_1^N+(1-r_j)L_1^N$ (see Section \ref{proof:euler for reavar}), the Euler allocation based on $\reavar_\gamma$ is given for every $i=1,\dots,N$ by
\[
\kappa^i \; = \; -E\big(\Delta E^i_1+(1-r_j)L^i_1 \, | \, E_1+(1-r_j)L_1\leq -\var_{\alpha_j}(E_1+(1-r_j)L_1)\big).
\]
\end{prop}
\begin{proof}
See Section \ref{proof:euler for reavar}.
\end{proof}
In summary, performance measurement inside firms can be based on recovery risk measures. Notions such as return on risk-adjusted capital (RoRaC) and RoRaC-compatible allocations may be extended to solvency regimes that control the size of recovery on creditors' claims in the case of default.


\subsection{Portfolio Optimization}\label{sec:optimization}

Risk measures are an important instrument to limit downside risk in portfolio optimization problems. This idea is related to the classical Markowitz problem, originally studied by \ci{Markowitz52}, in which standard deviation quantifies the risk. Efficient frontiers characterize the best tradeoffs between return and risk. In this section, we show how recovery risk measures may successfully be applied to portfolio optimization in practice. Our results extend related contributions on the risk measure $\ES$ --- discussed in \ci{RU00}, \ci{RU02}, and \ci{ZhuFu09} --- to the recovery risk measure $\reavar$.

\noindent {\bf Assets and liabilities.} \ We consider $k=1,\dots,K$ assets whose (ask) prices at dates $t=0,1$ are described by $S^k_t$ and whose random one-period returns are denoted by $R^k$ so that
$$
S^k_1 = S^k_0\cdot(1 + R^k).
$$
We assume that one-period returns have finite expectation. For every $k=1, \dots, K$ an investor invests a fraction $x^k>0$ of his or her total budget $b>0$ into asset $k$ so that $\sum_{k=1}^K x^k =1$. The total asset value at time $t=1$ is thus equal to
$$
b \cdot  \sum_{k=1}^K x^k   (1 + R^k)    \; = \; b \cdot \left( 1 + \sum_{k=1}^K x^k   R^k \right).
$$
We set $\bm{R} = (R^1, \dots , R^K)^\top$ and $\bm{x} = (x^1, \dots, x^K)^\top$. In addition, we suppose that the investor's liabilities at time $t=1$ amount to a random fraction $Z$ of the initial budget, i.e., the liabilities are equal to $bZ$. We assume that liabilities have finite expectation.

\noindent {\bf Efficient frontier.} \ We are interested in optimal combinations of return and downside risk -- the \emph{efficient frontier} -- but with risk measured by $\reavar$ instead of classical risk measures such as standard deviation or $\ES$. This problem can equivalently be stated either as the maximization of return under a risk constraint or as the minimization of risk for a given expected return. We will consider the latter formulation.

\noindent {\bf Expected return.} \ The expected future net asset value of the investor equals
$$
b \, \cdot \,  \left( 1\; + \;\sum_{k=1}^K  \, x^k  \, \E(R^k) \; - \; \E(Z) \right),
$$
implying that a target expected return can be achieved by requiring for some $c\in \bbr$ that
$$
\sum_{k=1}^K \, x^k \,   \E(R^k) \;   = \;  c.
$$

\noindent {\bf Downside risk and a minimax theorem.} \ We are interested in computing and optimizing
$$
\reavar_\gamma \left(  b \cdot \left[ 1 + \sum_{k=1}^K x^k   R^k  - Z  \right]  ,  bZ \right)   =   -b + b \cdot \reavar_\gamma \left(   \sum_{k=1}^K x^k   R^k  - Z ,  Z \right).
$$
We focus on the special case of piecewise-constant recovery functions $\gamma$ introduced in Section \ref{sect: choice of gamma}. In this case, as shown in Proposition \ref{prop: parametric gamma avar}, $\reavar$ is a maximum of finitely many terms involving $\ES$, namely
$$
\reavar_\gamma \left(   \sum_{k=1}^K x^k   R^k  - Z ,  Z \right)
\quad = \quad
\max_{i=1,\dots,n+1}\;\ES_{\alpha_i}\left(\sum_{k=1}^K x^kR^k-r_iZ\right).
$$
For convenience, for every $i=1, \dots ,n+1$ we define the auxiliary function $\Psi^i(\bm{x},\cdot):\R\to\R$ by
$$
\Psi^i(\bm{x},v) = \frac 1 {\alpha_i}  \cdot  E \left( \left[v- \sum_{k=1}^K x^k   R^k - r_i Z \right]^+\right) - v ,
$$
which allows us to write
$$
\reavar_\gamma \left(   \sum_{k=1}^K x^k   R^k  - Z ,  Z \right)
\quad = \quad
\max_{i=1,\dots,n+1}
\;
\min_{v\in\bbr}
\;
\Psi^i (\bm{x}, v) .
$$
This follows from a representation of $\ES$ originally due to \ci{RU00} and \ci{RU02}. Their results show that $\ES$ belongs to the family of divergence risk measures, which also correspond to optimized certainty equivalents, see \ci{BT87} and \ci{BT07}. The following theorem facilitates the computation of an efficient frontier in the case of $\reavar$.

\begin{thm}\label{thm:minimax}
The following minimax equality holds:
$$
\max_{i=1,\dots,n+1}
\;
\min_{v\in\bbr}
\;
\Psi^i (\bm{x}, v)
\; = \;
\min_{v\in\bbr}
\;
\max_{i=1,\dots,n+1}
\;
\Psi^i (\bm{x}, v).
$$
\end{thm}
\begin{proof}
See Section~\ref{proof:minimax}.
\end{proof}

\noindent {\bf Computing the efficient frontier.} \ We now focus on the problem of determining the efficient frontier. This extends the methodology of \ci{RU00} and \ci{ZhuFu09} to the case of $\reavar$. Characterizing the efficient frontier --- consisting of pairs of returns and downside risk --- is equivalent to minimizing the function in Theorem~\ref{thm:minimax} additionally over $\bm{x}\in \xcal$ where $\xcal = \{ \bm{x}\in[0,\infty)^K \,; \ \sum_{k=1}^K  x^k=1, \ \sum_{k=1}^K  x^k\,\E(R^k)    =  c \}$ is a convex polyhedron. This problem can be conveniently reformulated as
$$
\min_{(\bm{x}, v, \Upsilon )\in \xcal \times \bbr \times \bbr}  \;
\left\{
\Upsilon \,; \
\Psi^i (\bm{x}, v) \leq \Upsilon , \ i=1,\dots, n+1 \right\}.
$$
The evaluation of $\Psi^i$ involves the computation of an expectation. In typical applications in practice, these expectations are approximated via Monte Carlo simulations. This approach allows to reformulate the problem as a linear program, i.e., the minimization of a linear function on a convex polyhedron of the form
\begin{eqnarray*}
\min \; \Upsilon&&\\
\mbox{s.t.}
&& \frac 1 {M \cdot \alpha_i}  \cdot  \sum_{m=1}^M  u^{i,m} - v \leq \Upsilon , \quad i=1,\dots, n+1,\\
&& u^{i,m} \geq v- \sum_{k=1}^K x^k   R^{k,m} - r_i Z^m, \quad i=1,\dots, n+1, \ m=1,\dots, M,\\
&& u^{i,m}\geq 0, \quad i=1,\dots, n+1, \ m=1,\dots, M,\\
\mbox{over}
&& (\bm{x}, v, \Upsilon)\in \xcal \times \bbr \times \bbr,
\end{eqnarray*}
where $(\bm{R}^1, Z^1),\dots,(\bm{R}^M, Z^M)$ are $M$ independent simulations of the pair $(\bm{R},Z)$. This problem is tractable on the basis of standard algorithmic implementations.

In this section, we showed how efficient frontiers can be computed that characterize the best tradeoffs between risk and return when risk is measured by recovery risk measures. This demonstrates that these risk measures can   successfully be applied to portfolio optimization in practice.


\section{Conclusion}

Risk measures used in solvency regulation specify guard rails for financial firms such as banks or insurance companies. Within their legal boundaries firms can otherwise freely choose their actions, e.g., in order to maximize shareholder value. As a consequence, an axiomatic theory of risk measures for solvency regulation should carefully formulate and capture the goals of regulation, and determine and investigate suitable instruments to meet them. The issue of recovery on creditors' claims has not yet been considered in sufficient detail, and the existing literature lacks solvency requirements that provide adequate protection of customers and counterparties in the case of default. In this paper, we propose the novel concept of recovery risk measures to resolve this issue. We analyze the properties of these risk measures and describe how to apply them in the context of solvency regulation, performance measurement, and portfolio optimization. Our findings suggest that recovery risk measures add value to the current risk management toolkit. They are tractable tools for both internal risk management and solvency regulation and can be employed to provide a more comprehensive picture on tail risk with a focus on safeguarding the interests of creditors and policyholders.

Various extensions of the suggested framework are possible. Our framework focuses on a static setting as common in solvency regulation where relevant time horizons are fixed. Dynamic or conditional solvency risk measures would be an interesting extension; see \ci{bc17} for a survey on previous research on this topic. Recovery risk measures are applied to single firms in this paper. From this perspective, another natural extension is the regulation of financial systems. As outlined in the literature,  systemic risk measures should capture the local and global interaction of economic agents and operationalize the emerging risk at the level of the entire system; see, e.g., \ci{chen2013axiomatic}, \ci{kromer2013systemic}, \ci{FRW17}, \ci{BFFM19}. A related issue is a more comprehensive analysis of optimal investment under constraints in terms of recovery risk measures and equilibrium models of markets. Our approach to model firms' balance sheets directly offers a natural and flexible starting point to address such problems.


\vfill
\pagebreak

\appendix


\section{Proofs}

\subsection{Proof of Proposition \ref{prop:no recovery control}}
\label{proof:no recovery control}

\begin{proof}
Fix $\lambda\in(0,1)$. By nonatomicity, for every $p\in(0,\alpha)$ we find an event $F_p\in\cF$ such that $\probp(F_p)=p$. Set $A_p=\frac{p}{\alpha-p}1_{F^c_p}$ and $L_p=1_{F_p}$ and note that $A_p-L_p=-1_{F_p}+\frac{p}{\alpha-p}1_{F_p^c}$. Note also that both $A_p$ and $L_p$ belong to $\cX$. A simple computation shows that
\[
\VaR_\alpha(A_p-L_p) \leq \avar_\alpha(A_p-L_p) = \frac{1}{\alpha}\left(p-(\alpha-p)\frac{p}{\alpha-p}\right) = 0.
\]
Moreover, we have $\probp(A_p\geq\lambda L_p)=\probp(F_p^c)=1-p$. As a result,
\begin{align*}
1-\alpha
&\leq
\inf\{\probp(A\geq L) \,; \ A,L\in\cX, \ \VaR_\alpha(A-L)\leq0\} \\
&\leq
\inf\{\probp(A\geq\lambda L) \,; \ A,L\in\cX, \ \VaR_\alpha(A-L)\leq0\} \\
&\leq
\inf\{\probp(A\geq\lambda L) \,; \ A,L\in\cX, \ \avar_\alpha(A-L)\leq0\} \\
&\leq
\inf\{\probp(A_p\geq\lambda L_p) \,; \ 0<p<\alpha\} \\
&=
\inf\{1-p \,; \ 0<p<\alpha\} = 1-\alpha.
\end{align*}
This yields the desired statements.
\end{proof}


\subsection{Remarks to Section~\ref{sec:introrevar}}

\begin{rem}\label{rem:revar_solvency test}
In order to justify the equivalence between \eqref{revar_solvency test} and \eqref{revar_solvency test 2}, observe that
\begin{eqnarray*}
\revar_\gamma(\Delta E_1,L_1)\leq E_0 \ &\iff & \ \revar_\gamma(E_1,L_1)\leq 0 \\
&\iff& \ \forall\, \lambda\in[0,1]\,:\quad \var_{\gamma(\lambda)}(A_1-\lambda L_1)\leq 0 \\
&\iff& \ \forall\, \lambda\in[0,1]\,:\quad \probp(A_1<\lambda L_1)\leq\gamma(\lambda) \\
&\iff& \ \forall\, \lambda\in[0,1]\,:\quad  \probp(A_1\geq\lambda L_1)\geq1-\gamma(\lambda).
\end{eqnarray*}
\end{rem}

\begin{rem}\label{rem:cond_recov}
The solvency test \eqref{revar_solvency test} also controls the conditional recovery probabilities given default. Indeed, assuming that $\probp(E_1<0)>0$, for all fractions $\lambda\in[0,1]$ we have
\[
\probp(A_1\geq\lambda L_1\,|\,E_1<0) = \frac{\probp(\lambda L_1\leq A_1<L_1)}{\probp(E_1<0)} = 1-\frac{\probp(A_1<\lambda L_1)}{\probp(E_1<0)}.
\]
This implies the following equivalent formulation of the recovery-adjusted solvency test:
\[
\revar_\gamma(E_1,L_1)\leq 0 \ \iff \ \forall \, \lambda\in[0,1]\,:\quad \probp(A_1\geq\lambda L_1\,|\,A_1<L_1)\geq 1-\frac{\gamma(\lambda)}{\probp(E_1<0)}.
\]
In particular, if the company's unconditional default probability $\probp(E_1<0)$ attains $\gamma(1)$, then the lower bound on conditional recovery probabilities depends only on $\gamma$.
\end{rem}

\subsection{Proof of Proposition \ref{prop: revar with piecewise gamma}}\label{sect: proof parametric gamma}

\begin{proof}
Fix $i=1,\dots,n$ and observe that $\gamma$ is constant and equal to $\alpha_i$ on the interval $[r_{i-1},r_i)$. Hence, we get
\[
\VaR_{\gamma(\lambda)}(X+(1-\lambda)Y) = \VaR_{\alpha_i}(X+(1-\lambda)Y) \leq \VaR_{\alpha_i}(X+(1-r_i)Y)
\]
for every $\lambda\in[r_{i-1},r_i)$ by positivity of $Y$ and monotonicity of $\VaR$. As a result,
\[
\sup_{\lambda\in[r_{i-1},r_i)}\VaR_{\gamma(\lambda)}(X+(1-\lambda)Y) = \VaR_{\alpha_i}(X+(1-r_i)Y).
\]
Similarly, observe that $\gamma$ is constant and equal to $\alpha_{n+1}$ on the interval $[r_n,1]$. Hence, we get
\[
\VaR_{\gamma(\lambda)}(X+(1-\lambda)Y) = \VaR_{\alpha_{n+1}}(X+(1-\lambda)Y) \leq \VaR_{\alpha_{n+1}}(X)
\]
for every $\lambda\in[r_n,1]$ by positivity of $Y$ and monotonicity of $\VaR$. As a result,
\[
\sup_{\lambda\in[r_n,1]}\VaR_{\gamma(\lambda)}(X+(1-\lambda)Y) = \VaR_{\alpha_{n+1}}(X).
\]
The desired assertion is a direct consequence of the above identities.
\end{proof}


\subsection{Proof of Proposition \ref{prop:elementvar}}\label{proof:elementvar}

\begin{proof}
The statements from {\em (a)} to {\em (d)} follow directly from Proposition \ref{prop:elementrerho}. To establish {\em (e)}, take an arbitrary $X\in L^0$ and note that both conditions $\gamma(0)>0$ and $X$ is bounded from below imply that $\VaR_{\gamma(0)}(X)<\infty$. In particular, $X\geq m$ for some $m\in\R$ yields $\VaR_{\gamma(0)}(X) \leq \VaR_{\gamma(0)}(m) = -m$ by monotonicity and cash invariance of $\VaR$. As a result, it suffices to show that $\VaR_{\gamma(0)}(X)<\infty$ implies $\revar(X)<\infty$. This is an immediate consequence of Proposition \ref{prop:elementrerho}.
\end{proof}

\subsection{Proof of Proposition~\ref{prop:elementrerho}}\label{proof:elementrerho}

\begin{proof}
{\em (a)} For every $X\in\xcal$ and $m\in\R$ the cash invariance of $\rho_\lambda$ readily implies
\[
\rerho(X+m,Y)
=
\sup_{\lambda\in[0,1]}\rho_\lambda(X+m+(1-\lambda)Y)
=
\sup_{\lambda\in[0,1]}\rho_\lambda(X+(1-\lambda)Y)-m
=
\rerho(X,Y)-m.
\]

\noindent{\em (b)} Since $X_1+(1-\lambda)Y_1\leq X_2+(1-\lambda)Y_2$ for every $\lambda\in[0,1]$, the monotonicity of $\rho_\lambda$ yields
\[
\rerho(X_1,Y_1) = \sup_{\lambda\in[0,1]}\rho_\lambda(X_1+(1-\lambda)Y_1) \geq \sup_{\lambda\in[0,1]}\rho_\lambda(X_2+(1-\lambda)Y_2) = \rerho(X_2,Y_2).
\]

\noindent{\em (c)} It follows from the convexity of $\rho_\lambda$ that
{\footnotesize
\begin{eqnarray*}
\rerho(aX_1+(1-a)X_2,aY_1+(1-a)Y_2)
&=&
\sup_{\lambda\in[0,1]}\rho_\lambda(a(X_1+(1-\lambda)Y_1)+(1-a)(X_2+(1-\lambda)Y_2)) \\
&\leq&
\sup_{\lambda\in[0,1]}\big(a\rho_\lambda(X_1+(1-\lambda)Y_1)+(1-a)\rho_\lambda(X_2+(1-\lambda)Y_2)\big) \\
&\leq&
a\sup_{\lambda\in[0,1]}\rho_\lambda(X_1+(1-\lambda)Y_1)+(1-a)\sup_{\lambda\in[0,1]}\rho_\lambda(X_2+(1-\lambda)Y_2) \\
&=&
a\rerho(X_1,Y_1)+(1-a)\rerho(X_2,Y_2).
\end{eqnarray*}
}

\noindent{\em (d)} Similarly, we can use the subadditivity of $\rho_\lambda$ to get
\begin{eqnarray*}
\rerho(X_1+X_2,Y_1+Y_2)
&=&
\sup_{\lambda\in[0,1]}\rho_\lambda(X_1+(1-\lambda)Y_1+X_2+(1-\lambda)Y_2) \\
&\leq&
\sup_{\lambda\in[0,1]}\{\rho_\lambda(X_1+(1-\lambda)Y_1)+\rho_\lambda(X_2+(1-\lambda)Y_2))\} \\
&\leq&
\sup_{\lambda\in[0,1]}\rho_\lambda(X_1+(1-\lambda)Y_1)+\sup_{\lambda\in[0,1]}\rho_\lambda(X_2+(1-\lambda)Y_2)) \\
&=&
\rerho(X_1,Y_1)+\rerho(X_2,Y_2).
\end{eqnarray*}

\noindent{\em (e)} Using the positive homogeneity of $\rho_\lambda$, we easily see that
\[
\rerho(aX,aY) = \sup_{\lambda\in[0,1]}\rho_\lambda(aX+(1-\lambda)aY) = \sup_{\lambda\in[0,1]}a\rho_\lambda(X+(1-\lambda)Y) = a\rerho(X,Y).
\]

\noindent{\em (f)} Observe that $a(1-\lambda)Y\geq(1-\lambda)Y$ for every $\lambda\in[0,1]$. This is because $Y$ is positive. Then, it follows from the monotonicity and positive homogeneity of $\rho_\lambda$ that
\[
\rerho(aX,Y) = \sup_{\lambda\in[0,1]}\rho_\lambda(aX+(1-\lambda)Y) \geq \sup_{\lambda\in[0,1]}\rho_\lambda(a(X+(1-\lambda)Y)) = a\rerho(X,Y).
\]

\noindent{\em (g)} It follows from monotonicity that $\rho_\lambda((1-\lambda)Y)\leq\rho_\lambda(0)$ for every $\lambda\in[0,1]$. This is because $Y$ is positive. Then, normalization yields
\[
0 = \rho_1(0) \leq \rerho(0,Y) = \sup_{\lambda\in[0,1]}\rho_\lambda((1-\lambda)Y) \leq \sup_{\lambda\in[0,1]}\rho_\lambda(0) = 0.
\]

\noindent{\em (h)} Take $X\in\xcal$ such that $\rho_0(X)<\infty$ and observe that $X+(1-\lambda)Y\geq X$ for every $\lambda\in[0,1]$. This is because $Y$ is positive. Then, monotonicity implies
\[
\rerho(X,Y) = \sup_{\lambda\in[0,1]}\rho_\lambda(X+(1-\lambda)Y) \leq \sup_{\lambda\in[0,1]}\rho_\lambda(X) = \rho_0(X) < \infty,
\]
where we used that $\rho_0(X)\geq\rho_\lambda(X)$ for every $\lambda\in[0,1]$ by assumption.
\end{proof}


\subsection{General Dual Representation}\label{sec:dual representation general}

The next proposition focuses on dual representations of recovery risk measures defined on standard $L^p$ spaces. We refer to \ci{FS} for a broad discussion on dual representations. For every $q\in[1,\infty]$ we denote by $\mcal^q_1(\probp)$ the set of probability measures $\probq$ over $(\Omega,\fil)$ that are absolutely continuous with respect to $\probp$ and satisfy $\frac{d\probq}{d\probp}\in L^q$. Recall that a map $\rho:\xcal\to\R\cup\{\infty\}$ satisfies the {\em Fatou property} if for every sequence $(X_n)\subset\xcal$ and every $X\in\xcal$ we have
\[
\mbox{$X_n\to X$ $\probp$-almost surely}, \  \sup_{n\in\N}\vert X_n\vert\in\xcal \ \implies \ \rho(X)\leq\liminf_{n\to\infty}\rho(X_n).
\]
The Fatou property corresponds to a weak form of continuity, namely lower semicontinuity, with respect to dominated almost-sure convergence. A well-known result by \ci{Jouini2006} shows that every distribution-based monetary risk measure defined on $L^\infty$ has the Fatou property. We refer to \ci{gao2018fatou} for a general result beyond bounded random variables.

\begin{prop}
\label{prop:dual representation general}
Let $\cX=L^p$ for some $p\in[1,\infty]$ and assume that $\rho_\lambda$ is a convex monetary risk measure with the Fatou property for every $\lambda\in[0,1]$. Then, for all $X,Y\in\xcal$
\[
\rerho(X,Y) = \sup_{\probq\in\mcal^q_1(\probp)}\big(\E_\probq(-X)-\alpha_Y(\probq)\big)
\]
where $q=\frac{p}{1-p}$ and for every $\probq\in\mcal^q_1(\probp)$ we set
\[
\alpha_Y(\probq) = \E_\probq(Y)+\inf_{\lambda\in[0,1]}\bigg(\sup_{X\in\cA_\lambda}\E_\probq(-X)-\lambda \E_\probq(Y)\bigg)
\]
where $\cA_\lambda=\{X\in L^p \,; \ \rho_\lambda(X)\leq0\}$.
\end{prop}

\begin{proof}
Set $q=\frac{p}{1-p}$ and fix $\lambda\in[0,1]$. If $p=\infty$, then it follows from Theorem 4.33 in \ci{FS} that
\[
\rho_\lambda(X) = \sup_{\probq\in\mcal^q_1(\probp)}\big(\E_\probq(-X)-\alpha_\lambda(\probq)\big)
\]
for every $X\in\cX$, where
\[
\alpha_\lambda(\probq) = \sup_{X\in\{\rho_\lambda\leq0\}}\E_\probq(-X).
\]
If $p<\infty$, the same result follows from Corollary 7 in \ci{frittelli2002} once we note that $\rho_\lambda$ is lower semicontinuous with respect to the $L^p$ norm. To see this, take a sequence $(X_n)\subset L^p$ and $X\in L^p$ such that $X_n\to X$ in the $L^p$ norm. By Theorem 13.6 in \ci{aliprantis06}, there exists a subsequence of $(X_n)$, which we still denote by $(X_n)$ without loss of generality, such that $X_n\to X$ $\probp$-almost surely and $\sup_{n\in\mathbb{N}}|X_n|\in L^p$. As a result of the Fatou property, we have
\[
\rho_\lambda(X) \leq \liminf_{n\to\infty}\rho_\lambda(X_n),
\]
which implies the desired lower semicontinuity. Now, for every $X\in\cX$ we obtain
\begin{eqnarray*}
\rerho(X,Y)
&=&
\sup_{\lambda\in[0,1]}\sup_{\probq\in\mcal^q_1(\probp)}\big(\E_\probq(-X-(1-\lambda)Y)-
\alpha_\lambda(\probq)\big) \\
&=&
\sup_{\probq\in\mcal^q_1(\probp)}\sup_{\lambda\in[0,1]}\big(\E_\probq(-X)-\E_\probq(Y)+\lambda \E_\probq(Y)-\alpha_\lambda(\probq)\big) \\
&=&
\sup_{\probq\in\mcal^q_1(\probp)}\bigg(\E_\probq(-X)-\E_\probq(Y)-\inf_{\lambda\in[0,1]}
\big(\alpha_\lambda(\probq)-\lambda \E_\probq(Y)\big)\bigg) \\
&=&
\sup_{\probq\in\mcal^q_1(\probp)}\big(\E_\probq(-X)-\alpha_Y(\probq)\big).
\end{eqnarray*}
This establishes the desired representation and concludes the proof.
\end{proof}


\subsection{Proof of Proposition \ref{prop: parametric gamma avar}}\label{sect: proof parametric gamma avar}

\begin{proof}
Fix $i=1,\dots,n$ and observe that $\gamma$ is constant and equal to $\alpha_i$ on the interval $[r_{i-1},r_i)$. Hence, we get
\[
\avar_{\gamma(\lambda)}(X+(1-\lambda)Y) = \avar_{\alpha_i}(X+(1-\lambda)Y) \leq \avar_{\alpha_i}(X+(1-r_i)Y)
\]
for every $\lambda\in[r_{i-1},r_i)$ by positivity of $Y$ and monotonicity of $\avar$. As a result,
\[
\sup_{\lambda\in[r_{i-1},r_i)}\avar_{\gamma(\lambda)}(X+(1-\lambda)Y) = \avar_{\alpha_i}(X+(1-r_i)Y).
\]
Similarly, observe that $\gamma$ is constant and equal to $\alpha_{n+1}$ on the interval $[r_n,1]$. Hence, we get
\[
\avar_{\gamma(\lambda)}(X+(1-\lambda)Y) = \avar_{\alpha_{n+1}}(X+(1-\lambda)Y) \leq \avar_{\alpha_{n+1}}(X)
\]
for every $\lambda\in[r_n,1]$ by positivity of $Y$ and monotonicity of $\avar$. As a result,
\[
\sup_{\lambda\in[r_n,1]}\avar_{\gamma(\lambda)}(X+(1-\lambda)Y) = \avar_{\alpha_{n+1}}(X).
\]
The desired assertion is a direct consequence of the above identities.
\end{proof}


\subsection{Proof of Proposition \ref{prop:elementreavar}}\label{proof:elementreavar}

\begin{proof}
The statements from {\em (a)} to {\em (g)} follow directly from Proposition \ref{prop:elementrerho}. To establish {\em (h)}, we can proceed as in the proof of Proposition \ref{prop:elementvar}. In particular, take an arbitrary $X\in L^1$ and note that both conditions $\gamma(0)>0$ and $X$ is bounded from below imply that $\avar_{\gamma(0)}(X)<\infty$. In particular, $X\geq m$ for some $m\in\R$ yields $\avar_{\gamma(0)}(X) \leq \avar_{\gamma(0)}(m) = -m$ by monotonicity and cash invariance of $\avar$. As a result, it suffices to show that $\avar_{\gamma(0)}(X)<\infty$ implies $\reavar(X)<\infty$. This follows immediately from Proposition \ref{prop:elementrerho}.
\end{proof}

\subsection{Dual Representation for Recovery Average Value at Risk}\label{proof:dual representation reavar}

Since $\avar$ is convex and satisfies the Fatou property, we can derive from Proposition \ref{prop:dual representation general} in the appendix a dual representation of $\reavar$ in terms of probability measures. Here, we denote by $\mcal^\infty_1(\probp)$ the set of all probability measures $\probq$ over $(\Omega,\fil)$ that are absolutely continuous with respect to $\probp$ and such that the Radon-Nikodym derivative $\frac{d\probq}{d\probp}$ belongs to $L^\infty$.

\begin{prop}
\label{prop:dual representation reavar}
For all $X,Y\in L^1$ the following representation holds:
\[
\reavar_\gamma(X,Y) = \sup_{\probq\in\mcal^\infty_1(\probp)}\big(\E_\probq(-X)-(1-\lambda(\probq))\E_\probq(Y)\big)
\]
where for each $\probq\in\mcal^\infty_1(\probp)$ we set
\[
\lambda(\probq) = \sup\left\{\lambda\in[0,1] \,; \ \gamma(\lambda)\leq\left\|\frac{d\probq}{d\probp}\right\|_\infty^{-1}\right\}.
\]
\end{prop}
\begin{proof}
It is well-known, see e.g.\ Theorem 4.52 in \ci{FS} for the classical statement in $L^\infty$, that for all $\lambda\in[0,1]$ and $X\in L^1$
\[
\avar_{\gamma(\lambda)}(X) = \sup_{\probq\in\mcal^\infty_1(\probp)}\big(\E_\probq(-X)-\alpha_\lambda(\probq)\big)
\]
where
\[
\alpha_\lambda(\probq) =
\begin{cases}
0 & \mbox{if $\frac{d\probq}{d\probp}\leq\frac{1}{\gamma(\lambda)}$},\\
\infty & \mbox{otherwise}.
\end{cases}
\]
As a result, it follows from Proposition \ref{prop:dual representation general} that
\[
\reavar_\gamma(X,Y) = \sup_{\probq\in\mcal^\infty_1(\probp)}\big(\E_\probq(-X)-\alpha_Y(\probq)\big)
\]
where
\[
\alpha_Y(\probq) = \E_\probq(Y)+\inf_{\lambda\in[0,1]}\big(\alpha_\lambda(\probq)-\lambda \E_\probq(Y)\big).
\]
Now, for every $\probq\in\mcal^\infty_1(\probp)$ define
\[
\Lambda(\probq) = \left\{\lambda\in[0,1] \,; \  \left\|\frac{d\probq}{d\probp}\right\|_\infty\leq\frac{1}{\gamma(\lambda)}\right\}
\]
and set $\lambda(\probq)=\sup\Lambda(\probq)$. Since $Y$ is positive, it is easy to see that
\[
\inf_{\lambda\in[0,1]}\big(\alpha_\lambda(\probq)-\lambda \E_\probq(Y)\big) = \inf_{\lambda\in\Lambda(\probq)}\big(-\lambda \E_\probq(Y)\big) = -\lambda(\probq)\E_\probq(Y)
\]
for every $\probq\in\mcal^\infty_1(\probp)$. This implies that
\[
\reavar_\gamma(X,Y) = \sup_{\probq\in\mcal^\infty_1(\probp)}\big(\E_\probq(-X)-(1-\lambda(\probq))\E_\probq(Y)\big).\qedhere
\]
\end{proof}

As discussed in \eqref{eq:recovery_ref}, the term
$E_1+ (1-\lambda) L_1$ represents the available resources of the firm at time $1$ beyond a recovery level $\lambda$. For a fixed probability measure $\probq\in\mcal^\infty_1(\probp)$, we can thus interpret $\E_\probq(-E_1)-(1-\lambda)\E_\probq(L_1)$
as the expected shortfall below the recovery level $\lambda$ with respect to the measure $\probq$. Proposition \ref{prop:dual representation reavar} thus represents $\reavar_\gamma$ as a supremum over expected shortfalls below different recovery levels over the collection of absolutely continuous probability measures with bounded Radon-Nikodym density: The recovery levels depend on $\probq$ and are given by the generalized inverse of $\gamma$ evaluated at the inverse of the supremum norm of the Radon-Nikodym derivative of $\probq$ with respect to the reference measure $\probp$. The robust representation in Proposition \ref{prop:dual representation reavar} can thus be interpreted as a worst-case approach in the face of Knightian uncertainty: The recovery risk measure $\reavar_\gamma$ is the worst-case expected shortfall beyond different recovery levels over a class of probability measures. The size of the recovery level encodes to what extent different probability measures contribute to the worst-case.


\subsection{Complementary Material for Section \ref{sec:pilh}}
\label{sec:ex-recovery}

Fix $k\in[1,100]$. We know from Example \ref{ex:interests} that $\VaR_\alpha(E_1^k)=k-100$. Moreover,
\[
\VaR_\beta(E_1^k+(1-r)L_1)=
\begin{cases}
100r-k & \mbox{if} \ \beta<\frac{\alpha}{2}, \  k\leq\frac{101+99r}{2},\\
k+r-101 & \mbox{otherwise}.
\end{cases}
\]
It follows from Proposition \ref{prop: revar with piecewise gamma} that
\[
\revar_\gamma(E^k_1,L_1) =
\max\{\VaR_\alpha(E^k_1),\VaR_\beta(E^k_1+(1-r)L_1).
\]
As a result, a direct calculation yields
\[
\revar_\gamma(E^k_1,L_1) =
\begin{cases}
100r-k & \mbox{if} \ \beta<\frac{\alpha}{2}, \  k\leq50(r+1),\\
k-100 & \mbox{otherwise}.
\end{cases}
\]
This shows that
\[
\revar_\gamma(E^k_1,L_1)\leq0 \ \iff \
\begin{cases}
k\geq100r & \mbox{if} \ \beta<\frac{\alpha}{2},\\
k\geq0 & \mbox{if} \ \beta\geq\frac{\alpha}{2}.
\end{cases}
\]
We turn to $\reavar$. We know from Example \ref{ex:interests} that $\avar_\alpha(E^k_1)=0$. Moreover,
\[
\avar_\beta(E_1^k+(1-r)L_1)=
\begin{cases}
100r-k & \mbox{if} \ \beta<\frac{\alpha}{2}, \  k\leq\frac{101+99r}{2},\\
r-101+\frac{\alpha}{2\beta}(101+99r)+(1-\frac{\alpha}{\beta})k & \mbox{if} \ \beta\geq\frac{\alpha}{2}, \ k\leq\frac{101+99r}{2},\\
k+r-101 & \mbox{otherwise}.
\end{cases}
\]
It follows from Proposition \ref{prop: parametric gamma avar} that
\[
\reavar_\gamma(E^k_1,L_1) = \max\{\avar_\alpha(E^k_1),\avar_\beta(E^k_1+(1-r)L_1)\}.
\]
As a result, we obtain
\[
\reavar_\gamma(E^k_1,L_1) =
\begin{cases}
100r-k & \mbox{if} \ \beta<\frac{\alpha}{2}, \ k\leq100r,\\
r-101+\frac{\alpha}{2\beta}(101+99r)+(1-\frac{\alpha}{\beta})k & \mbox{if} \ \beta\geq\frac{\alpha}{2}, \ k\leq\frac{(99\alpha+2\beta)r-101(2\beta-\alpha)}{2(\alpha-\beta)},\\
0 & \mbox{otherwise}.
\end{cases}
\]
We infer that
\[
\reavar_\gamma(E^k_1,L_1)\leq0 \ \iff \
\begin{cases}
k\geq100r & \mbox{if} \ \beta<\frac{\alpha}{2},\\
k\geq\max\left\{\frac{(99\alpha+2\beta)r-101(2\beta-\alpha)}{2(\alpha-\beta)},0\right\} & \mbox{if} \ \beta\geq\frac{\alpha}{2}.
\end{cases}
\]


\subsection{Complementary Material for Section \ref{sec:incs}}
\label{sect: complementary numerical study}

The probability distribution function of $A_1$ is specified by
\[
\probp(A_1\leq x)=
\begin{cases}
\Phi\left(\frac{\log(x)-\mu}{\sigma}\right) & \mbox{if} \ x>0,\\
0 & \mbox{if} \ x\leq0,
\end{cases}
\]
where $\Phi$ is the distribution function of a standard normal random variable.
To define the probability distribution function of $L_1$, recall that the gamma distribution with rate $a>0$ and shape $b>0$ is given by
\[
G_{a,b}(x) :=
\begin{cases}
\frac{a^b}{\Gamma(b)}\int_0^x y^{b-1}e^{-ay}dy & \mbox{if} \ x>0,\\
0 & \mbox{if} \ x\leq0,
\end{cases}
\]
where $\Gamma$ is the gamma function. For every $p\in(0,1)$ set
\[
q_p(G_{a,b}) := \inf\{x\in\R \,; \ G_{a,b}(x)\geq p\}.
\]
The probability distribution function of $L_1$ is then specified by
\[
\probp(L_1\leq x)=
\begin{cases}
G_{\delta_0,\tau_0}(x) & \mbox{if} \ x<q_{97.5\%}(G_{\delta_0,\tau_0}),\\
G_{\delta,\tau}(x+q_{97.5\%}(G_{\delta,\tau})-q_{97.5\%}(G_{\delta_0,\tau_0})) & \mbox{if} \ x\geq q_{97.5\%}(G_{\delta_0,\tau_0}).
\end{cases}
\]

Recall that, by Sklar's Theorem, see, e.g., Theorem 2.3.3 in \ci{nelsen2007}, the joint distribution of $A_1$ and $L_1$ can be expressed in terms of a suitable copula function $C:[0,1]\times[0,1]\to[0,1]$ as
\[
\probp(A_1\leq x,L_1\leq y) = C(\probp(A_1\leq x),\probp(L_1\leq y)).
\]
The assets and liabilities are linked through the Gaussian copula
\[
C(p,q) = \frac{1}{2\pi\sqrt{1-\rho^2}}\int_{-\infty}^{\Phi^{-1}(p)}\int_{-\infty}^{\Phi^{-1}(q)}e^{-\frac{u^2-2\rho uv+v^2}{2(1-\rho^2)}}dudv
\]
where $\rho\in(-1,1)$ is the correlation coefficient.

\subsection{Complementary Material for Section \ref{sect: maximal rec adj}}
\label{sect: preliminary supplementary section maximal rec adj}

For better comparability, we use the same paragraph titles from Section \ref{sect: maximal rec adj}.

\noindent {\bf Distribution of assets and liabilities}. The probability density function of $L_1$ is explicitly given by
\[
f(x)=
\begin{cases}
\frac{4(1-\alpha)}{a^2}x & \mbox{if} \ 0\leq x\leq\frac{a}{2},\\
-\frac{4(1-\alpha)}{a^2}x+\frac{4(1-\alpha)}{a} & \mbox{if} \ \frac{a}{2}<x\leq a,\\
\frac{4\alpha}{(c-b)^2}x-\frac{4\alpha b}{(c-b)^2} & \mbox{if} \ b\leq x\leq\frac{b+c}{2},\\
-\frac{4\alpha}{(c-b)^2}x+\frac{4\alpha c}{(c-b)^2} & \mbox{if} \ \frac{b+c}{2}<x\leq c,\\
0 & \mbox{otherwise},
\end{cases}
\]
for $\alpha=0.5\%$ and for fixed parameters $0<a<b<c$.

\noindent {\bf Regulatory benchmarks}. \ Note that $\probp(L_1\leq a)=1-\alpha$ and $\probp(L_1\leq x)<1-\alpha$ for every $x<a$. This yields
\begin{align*}
\VaR_{0.5\%}(E_1)
&=
\VaR_\alpha(-L_1)-k = \inf\{x\in\R \,; \ \probp(-L_1+x<0)\leq\alpha\}-k \\
&=
\inf\{x\in\R \,; \ \probp(L_1\leq x)\geq1-\alpha\}-k = a-k.
\end{align*}
The computation of $\avar_{1\%}(E_1)$ requires a bit more effort. First of all, define
\[
q = \VaR_{2\alpha}(-L_1) = \inf\{x\in\R \,; \ \probp(L_1\leq x)\geq1-2\alpha.
\]
Note that $q$ must satisfy
\[
\frac{1}{2}(a-q)\left(-\frac{4(1-\alpha)}{a^2}q+\frac{4(1-\alpha)}{a}\right) = \alpha.
\]
This gives $q=a\big(1\pm\sqrt{\frac{\alpha}{2(1-\alpha)}}\big)$. Since $q<a$ must hold, we obtain
\[
q = a\left(1-\sqrt{\frac{\alpha}{2(1-\alpha)}}\right).
\]
As a next step, observe that
\[
\avar_{2\alpha}(-L_1) = \E(L_1\vert L_1\geq\VaR_{2\alpha}(-L_1)) = \E(L_1\vert L_1\geq q) = \frac{\E(L_11_{\{L_1\geq q\}})}{\probp(L\geq q)}.
\]
We clearly have $\probp(L_1\geq q)=2\alpha$. Moreover, noting that $q>\frac{a}{2}$, we get
\begin{align*}
\E(L_11_{\{L_1\geq q\}}) \,
&=
\int_q^\infty xf(x)dx = \int_q^a xf(x)dx+\int_b^{\frac{b+c}{2}}xf(x)dx+\int_{\frac{b+c}{2}}^cxf(x)dx \\
&=
\frac{4(1-\alpha)}{a^2}\int_q^a(ax-x^2)dx+\frac{4\alpha}{(b-c)^2}\left(\int_b^{\frac{b+c}{2}}(x^2-bx)dx+\int_{\frac{b+c}{2}^c}(cx-x^2)dx\right) \\
&=
\frac{4(1-\alpha)}{a^2}\left(\frac{1}{6}a^3+\frac{1}{3}q^3-\frac{1}{2}aq^2\right)+\frac{4\alpha}{(b-c)^2}\left(\frac{1}{6}(b^3+c^3)-\frac{1}{24}(b+c)^3\right) \\
&=
\frac{4(1-\alpha)}{a^2}\frac{\alpha}{2(1-\alpha)}\left(\frac{1}{2}-\frac{1}{3}\sqrt{\frac{\alpha}{2(1-\alpha)}}\right)a^3+\frac{4\alpha}{(b-c)^2}\frac{1}{8}(b+c)(b-c)^2 \\
&=
2\alpha\left(\frac{1}{2}-\frac{1}{3}\sqrt{\frac{\alpha}{2(1-\alpha)}}\right)a+\alpha\frac{b+c}{2}.
\end{align*}
As a result, it follows that
\[
\avar_{1\%}(E_1)
=
\avar_{2\alpha}(-L_1)-k = \left(\frac{1}{2}-\frac{1}{3}\sqrt{\frac{\alpha}{2(1-\alpha)}}\right)a+\frac{b+c}{4}-k.
\]

\noindent {\bf Recovery-based capital requirements}. We turn to the computation of $\revar_\gamma(E_1,L_1)$. To this effect, define
\[
q_\beta = \VaR_{\beta}(-L_1) = \inf\{x\in\R \,; \ \probp(L_1\leq x)\geq1-\beta.
\]
If $\beta<0.25\%$, then $q_\beta$ must satisfy
\[
\frac{1}{2}(c-q_\beta)\left(-\frac{4\alpha}{(b-c)^2}q_\beta+\frac{4\alpha c}{(b-c)^2}\right) = \beta.
\]
This gives $q_\beta=c\pm\sqrt{\frac{\beta}{2\alpha}}(c-b)$. Since $q_\beta<c$ must hold, we obtain
\[
q_\beta = \sqrt{\frac{\beta}{2\alpha}}b+\left(1-\sqrt{\frac{\beta}{2\alpha}}\right)c.
\]
Similarly, if $\beta\geq0.25\%$, then $q_\beta$ must satisfy
\[
\frac{1}{2}(q_\beta-b)\left(\frac{4\alpha}{(b-c)^2}q_\beta-\frac{4\alpha b}{(b-c)^2}\right) = \alpha-\beta.
\]
This gives $q_\beta=b\pm\sqrt{\frac{\alpha-\beta}{2\alpha}}(c-b)$. Since $q_\beta>b$ must hold, we obtain
\[
q_\beta = \left(1-\sqrt{\frac{\alpha-\beta}{2\alpha}}\right)b+\sqrt{\frac{\alpha-\beta}{2\alpha}}c.
\]
For convenience, we set
\[
q_\beta(b,c)=
\begin{cases}
\sqrt{\frac{\beta}{2\alpha}}b+\left(1-\sqrt{\frac{\beta}{2\alpha}}\right)c & \mbox{if} \ \beta<0.25\%,\\
\left(1-\sqrt{\frac{\alpha-\beta}{2\alpha}}\right)b+\sqrt{\frac{\alpha-\beta}{2\alpha}}c & \mbox{if} \ \beta\geq0.25\%.
\end{cases}
\]
It follows that
\[
\VaR_{\beta}(A_1-rL_1)
=
r\VaR_\beta(-L_1)-k = rq_\beta(b,c)-k.
\]
As a result, we conclude that
\[
\revar_\gamma(E_1,L_1)
=
\max\{\VaR_\alpha(E_1),\VaR_\beta(A_1-rL_1)\} =
\max\{a,rq_\beta(b,c)\}-k.
\]

\subsection{Proof of Proposition \ref{prop: optimal rec adj}}
\label{sect: supplementary section maximal rec adj}

\begin{proof}
Throughout the proof we use the notation from Section \ref{sect: preliminary supplementary section maximal rec adj}. In the $\VaR$ case, the optimization problem can be reformulated in explicit terms as
\begin{eqnarray*}
\max \; \frac{rq_\beta(b,c)-k+E_0}{a-k+E_0}&&\\
\mbox{s.t.}
&& (1) \ \ k\geq a,\\
&& (2) \ \ k<a+E_0,\\
&& (3) \ \ k\geq rq_\beta(b,c),\\
&& (4) \ \ k<rq_\beta(b,c)+E_0,\\
&& (5) \ \ rq_\beta(b,c)>a,\\
&& (6) \ \ s_{min}\leq\frac{E_0}{a-k+E_0}\leq s_{max},\\
\mbox{over}
&& k>0 \ \,\mbox{and}\, \ 0<a<b<c.
\end{eqnarray*}
It is clear that, due to (5), both conditions (1) and (4) can be dropped as they are implied by conditions (3) and (2), respectively. Moreover, (6) clearly implies (2). Since the objective function is increasing in the term $q_\beta(b,c)$, the maximum is achieved by taking $rq_\beta(b,c)=k$ in (3). In this case, condition (5) is implied by (6). By conveniently rewriting condition (6), we thus obtain the equivalent problem
\begin{eqnarray*}
\max \; \frac{E_0}{a-k+E_0}&&\\
\mbox{s.t.}
&& a+\frac{s_{min}-1}{s_{min}}E_0 \leq k\leq a+\frac{s_{max}-1}{s_{max}}E_0,\\
\mbox{over}
&& k>0 \ \,\mbox{and}\, \ a>0.
\end{eqnarray*}
Note that the new objective function is increasing in $k$. As a result, the maximum is achieved by taking $k=a+\frac{s_{max}-1}{s_{max}}E_0$, which yields a recovery adjustment equal to
\[
\frac{E_0}{a-\Big(a+\frac{s_{max}-1}{s_{max}}E_0\Big)+E_0} = s_{max}.
\]
(The parameter $a$ can be selected to ensure a realistic loss probability. Indeed, we have
\[
\probp(E_1<E_0) = P\left(L_1> a-\frac{1}{s_{max}}E_0\right).
\]
It is then clear that we can always choose $a$ so as to ensure a loss probability around $50\%$. To this effect, it suffices to have $a-\frac{1}{s_{max}}E_0\approx\frac{a}{2}$). To deal with the $\avar$ case, it is first convenient to define the quantity
\[
\xi = \frac{1}{2}-\frac{1}{3}\sqrt{\frac{\alpha}{2(1-\alpha)}} = 0.47\dots.
\]
The corresponding optimization problem can be rewritten in explicit terms as
\begin{eqnarray*}
\max \; \frac{rq_\beta(b,c)-k+E_0}{\xi a+\frac{b+c}{4}-k+E_0}&&\\
\mbox{s.t.}
&& (1) \ \ k\geq\xi a+\frac{b+c}{4},\\
&& (2) \ \ k<\xi a+\frac{b+c}{4}+E_0,\\
&& (3) \ \ k\geq rq_\beta(b,c),\\
&& (4) \ \ k<rq_\beta(b,c)+E_0,\\
&& (5) \ \ rq_\beta(b,c)>a,\\
&& (6) \ \ s_{min}\leq\frac{E_0}{\xi a+\frac{b+c}{4}-k+E_0}\leq s_{max},\\
\mbox{over}
&& k>0 \ \,\mbox{and}\, \ 0<a<b<c.
\end{eqnarray*}
Observe that we cannot proceed as in the $\VaR$ case because of a more complex dependence on the parameters $b$ and $c$. As a first step, note that condition (6) implies both conditions (1) and (2). The objective function is decreasing in $a$. As a result, the maximum is achieved by taking $a=\frac{1}{\xi}\left(k-\frac{b+c}{4}-\frac{s_{max}-1}{s_{max}}E_0\right)$ in (6). We thus obtain the equivalent problem
\begin{eqnarray*}
\max \; \frac{rq_\beta(b,c)-k+E_0}{\frac{E_0}{s_{max}}}&&\\
\mbox{s.t.}
&& (1') \ \ k\geq rq_\beta(b,c),\\
&& (2') \ \ k<rq_\beta(b,c)+E_0,\\
&& (3') \ \ rq_\beta(b,c)>\frac{1}{\xi}\left(k-\frac{b+c}{4}-\frac{s_{max}-1}{s_{max}}E_0\right),\\
&& (4') \ \ 0<\frac{1}{\xi}\left(k-\frac{b+c}{4}-\frac{s_{max}-1}{s_{max}}E_0\right)<b,\\
\mbox{over}
&& k>0 \ \,\mbox{and}\, \ 0<b<c.
\end{eqnarray*}
The new objective function is decreasing in $k$. This implies that the maximum is achieved by taking $k=rq_\beta(b,c)$ in (1'), which entails a recovery adjustment equal to
\[
\frac{rq_\beta(b,c)-rq_\beta(b,c)+E_0}{\frac{E_0}{s_{max}}} = s_{max}.
\]
In this case, condition (2') is automatically satisfied. We need to show when there exist $0<b<c$ satisfying all the remaining conditions, namely (3') and (4'), under $k=rq_\beta(b,c)$. We focus on the case $\beta\geq\frac{\alpha}{2}$. In this case, setting $\lambda=\sqrt{\frac{\alpha-\beta}{2\alpha}}\in(0,\frac{1}{2}]$, we can write $q_\beta(b,c)=(1-\lambda)b+\lambda c$. Condition (3') can equivalently be written as
\[
\left((1-\xi)r(1-\lambda)-\frac{1}{4}\right)b+\left((1-\xi)r\lambda-\frac{1}{4}\right)c < \frac{s_{max}-1}{s_{max}}E_0.
\]
This holds for all $0<b<c$ provided that the two expressions multiplying $b$ and $c$ are both negative, i.e.
\begin{equation}
\label{eq: aux 1}
r \leq \frac{1}{4}\frac{1}{(1-\lambda)(1-\xi)}.
\end{equation}
Similarly, condition (4') is equivalent to
\[
\frac{s_{max}-1}{s_{max}}E_0-\left(r(1-\lambda)-\frac{1}{4}\right)b < \left(r\lambda-\frac{1}{4}\right)c < \frac{s_{max}-1}{s_{max}}E_0.
\]
For every $b>0$, this holds for a suitable $c>0$ provided that the expression in parenthesis multiplying $b$ and the expression multiplying $c$ are both strictly positive, i.e.
\begin{equation}
\label{eq: aux 2}
r > \frac{1}{4\lambda}.
\end{equation}
To ensure that $c$ can be taken to satisfy $c>b$, it suffices to impose the bound $b<\frac{1}{r\lambda-\frac{1}{4}}\frac{s_{max}-1}{s_{max}}E_0$. This shows that we can indeed find $0<b<c$ satisfying (3') and (4') under $k=rq_\beta(b,c)$ provided that \eqref{eq: aux 1} and \eqref{eq: aux 2} hold.
\end{proof}


\subsection{Proof of Proposition \ref{prop:euler for reavar}}\label{proof:euler for reavar}

\begin{proof}
We rely on Lemma 5.6 in \ci{tasche1999risk}. To this effect, the random variables
\[
X^1=-(\Delta E^1_1+(1-r_j)L^1_1),\dots,X^N=-(\Delta E^N_1+(1-r_j)L^N_1)
\]
have to satisfy the so-called (S)-Assumption in that paper. This stipulates some requirements on the joint distribution of the above random variables, namely:
\begin{itemize}
    \item $X^1,\dots,X^N$ are integrable and continuously distributed;
    \item the conditional distribution of $X^1$ given $X^2,\dots,X^N$ has a density $\phi$;
    \item $x^1\mapsto\phi(x^1,x^2,\dots,x^N)$ is continuous for all $x^2,\dots,x^N\in\R$;
    \item the maps $\Phi_1,\dots,\Phi_N:\R\times\R\setminus\{0\}\times\R^{N-1}\to\R$ defined by
\[
\Phi_1(x^1,u_1,\dots,u_N) =  E_P\left(\phi\bigg(u_1^{-1}\bigg(x^1-\sum_{i=2}^Nu_iX^i\bigg),X^2,\dots,X^N\bigg)\right),
\]
\[
\Psi_l(x^1,u_1,\dots,u_N) = E_P\left(X^l\phi\bigg(u_1^{-1}\bigg(x^1-\sum_{i=2}^Nu_iX^i\bigg),X^2,\dots,X^N\bigg)\right), \ \ \ l=2,\dots,N,
\]
are finite valued and continuous;
   \item for every $u=(u_1,\dots,u_N)\in\R\setminus\{0\}\times\R^{N-1}$ we have
\[
E_P\left(\phi\bigg(u_1^{-1}\bigg(q_{\alpha_j}(u)-\sum_{i=2}^Nu_iX^i\bigg),X^2,\dots,X^N\bigg)\right)>0,
\]
where $q_{\alpha_j}(u)=\inf\{x\in\R \,; \ \probp(\sum_{i=1}^Nu_iX^i\leq x)\geq1-\alpha_j\}$.
\end{itemize}
Now, fix $i=1,\dots,N$ and for every $k=1,\dots,n+1$ consider the convex (hence, continuous) function $f^{i,k}:\R\to\R$ defined by setting
\[
f^{i,k}(h)=\avar_{\alpha_k}(\Delta E_1+(1-r_k)L_1+h(\Delta E^i_1+(1-r_k)L_1^i)).
\]
By assumption we have that
\[
f^{i,j}(0) > \max_{k=1,\dots,n+1,\,\,k\neq j}f^{i,k}(0).
\]
It follows from continuity that there exists $\varepsilon>0$ such that
\[
f^{i,j}(h) > \max_{k=1,\dots,n+1,\,\,k\neq j}f^{i,k}(h).
\]
for every $h\in(-\varepsilon,\varepsilon)$. As a result,
\[
\reavar_\gamma(\Delta E_1+h \Delta E^i_1,L_1+hL^i_1) = \max_{k=1,\dots,n+1}f^{i,k}(h) = f^{i,j}(h)
\]
for every $h\in(-\varepsilon,\varepsilon)$. This immediately yields
\[
\frac{d}{dh} \reavar_\gamma(\Delta E_1+h \Delta E^i_1,L_1+h L^i_1)_{|_{h=0}} = \frac{df^{i,j}}{dh}(0).
\]
Since the (S)-Assumption holds, we infer from Theorem 4.4 in \ci{tasche1999risk} that
\[
\frac{df^{i,j}}{dh}(0)
=
-E\big( \Delta E^i_1 +(1-r_j)L^i_1 \, | \, \Delta E_1 +(1-r_j)L_1\leq -\var_{\alpha_j}(\Delta E_1 +(1-r_j)L_1)\big).
\]
This delivers the desired statement.
\end{proof}


\subsection{Proof of Theorem~\ref{thm:minimax}}\label{proof:minimax}

\begin{proof}
Set $\Delta^n = \{(\theta^1, \dots, \theta^{n+1})^\top\in [0,\infty)^{n+1} \,; \ \sum_{i=1}^{n+1} \theta^i=1\}$. We first observe that
$$
\max_{i=1,\dots,n+1}
\;
\inf_{v\in\bbr}
\;
\Psi^i (\bm{x}, v)
\; = \;
\max_{\bm{\theta}\in \Delta^n }
\;
\inf_{v\in\bbr}
\;
\sum_{i=1}^{n+1} \theta ^i
\Psi^i (\bm{x}, v) .
$$
Now arguing as in the proof of Theorem 1 of \ci{ZhuFu09} one shows that there exists a nonempty, closed, bounded interval $\bcal$ (which depends on $\bm{x}$) such that
\begin{equation}
\label{eq: minimax auxiliary}
\max_{\bm{\theta}\in \Delta^n }
\;
\inf_{v\in\bbr}
\;
\sum_{i=1}^{n+1} \theta ^i
\Psi^i (\bm{x}, v)
\; = \;
\max_{\bm{\theta}\in \Delta^n }
\;
\min_{v\in\bcal}
\;
\sum_{i=1}^{n+1} \theta ^i
\Psi^i (\bm{x}, v)  .
\end{equation}
More precisely, note that for every $i=1,\dots,n+1$ the function $\Psi^i$ is convex in $v$. We know from the classical results by Rockafellar and Uryasev \ci{RU00} and Rockafellar and Uryasev \ci{RU02} that $\Psi^i$ attains its minimum in a (nonempty) closed bounded interval $[q^-_i,q^+_i]$ (where the extremes are given by the lower, respectively, upper $\alpha_i$ quantiles of $\sum_{k=1}^K x^k R^k-r_iZ$). At this point, the critical observation is that any convex combination of the functions $\Psi^1,\dots,\Psi^{n+1}$ also attains its minimum in a (nonempty) closed bounded interval which is included in the closed bounded interval
$$
\bcal = \left[\min_{i=1,\dots,n+1}\;q^-_i,\max_{i=1,\dots,n+1}\;q^+_i\right].
$$
This can be easily established by induction on $n$. It is key to note that $\cB$ does not depend on the particular convex combination chosen. This yields \eqref{eq: minimax auxiliary}. Now, observe that the function $(\bm{\theta}, v) \mapsto \sum_{i=1}^{n+1} \theta ^i
\Psi^i (\bm{x}, v)$ is linear in $\bm{\theta}$, convex in $v$, and continuous in both variables. Thus, by a classical minimax theorem, see, e.g., Theorem 1 in \ci{fan1953minimax}, it follows that
$$
\max_{\bm{\theta}\in \Delta^n }
\;
\min_{v\in\bcal}
\;
\sum_{i=1}^{n+1} \theta ^i
\Psi^i (\bm{x}, v)
\; = \;
\min_{v\in\bcal}
\;
\max_{\bm{\theta}\in \Delta^n }
\;
\sum_{i=1}^{n+1} \theta ^i
\Psi^i (\bm{x}, v)
 \; = \;
\min_{v\in\bbr}
\;
\max_{\bm{\theta}\in \Delta^n }
\;
\sum_{i=1}^{n+1} \theta ^i
\Psi^i (\bm{x}, v) .
$$
In the last step, the inequality $\geq$  follows from $\bcal \subseteq \bbr$ while the inequality $\leq$ is a consequence of the minimax inequality
$$
\max_{\bm{\theta}\in \Delta^n }
\;
\inf_{v\in\bbr}
\;
\sum_{i=1}^{n+1} \theta ^i
\Psi^i (\bm{x}, v)
 \; \leq \;
\inf_{v\in\bbr}
\;
\max_{\bm{\theta}\in \Delta^n }
\;
\sum_{i=1}^{n+1} \theta ^i
\Psi^i (\bm{x}, v) .$$
Finally, observe that
$$
\min_{v\in\bbr}
\;
\max_{\bm{\theta}\in \Delta^n }
\;
\sum_{i=1}^{n+1} \theta ^i
\Psi^i (\bm{x}, v)
 \; = \;
 \min_{v\in\bbr}
\;
\max_{i=1,\dots,n+1}
\;
\Psi^i (\bm{x}, v) .
$$
This delivers the desired minimax equality.
\end{proof}


\newpage

\section{Complementary Figures}
\label{sect: plots appendix}

\begin{figure}[h]
\vspace{-0.8cm}
\centering
\subfigure{
\includegraphics[width=0.4\textwidth]{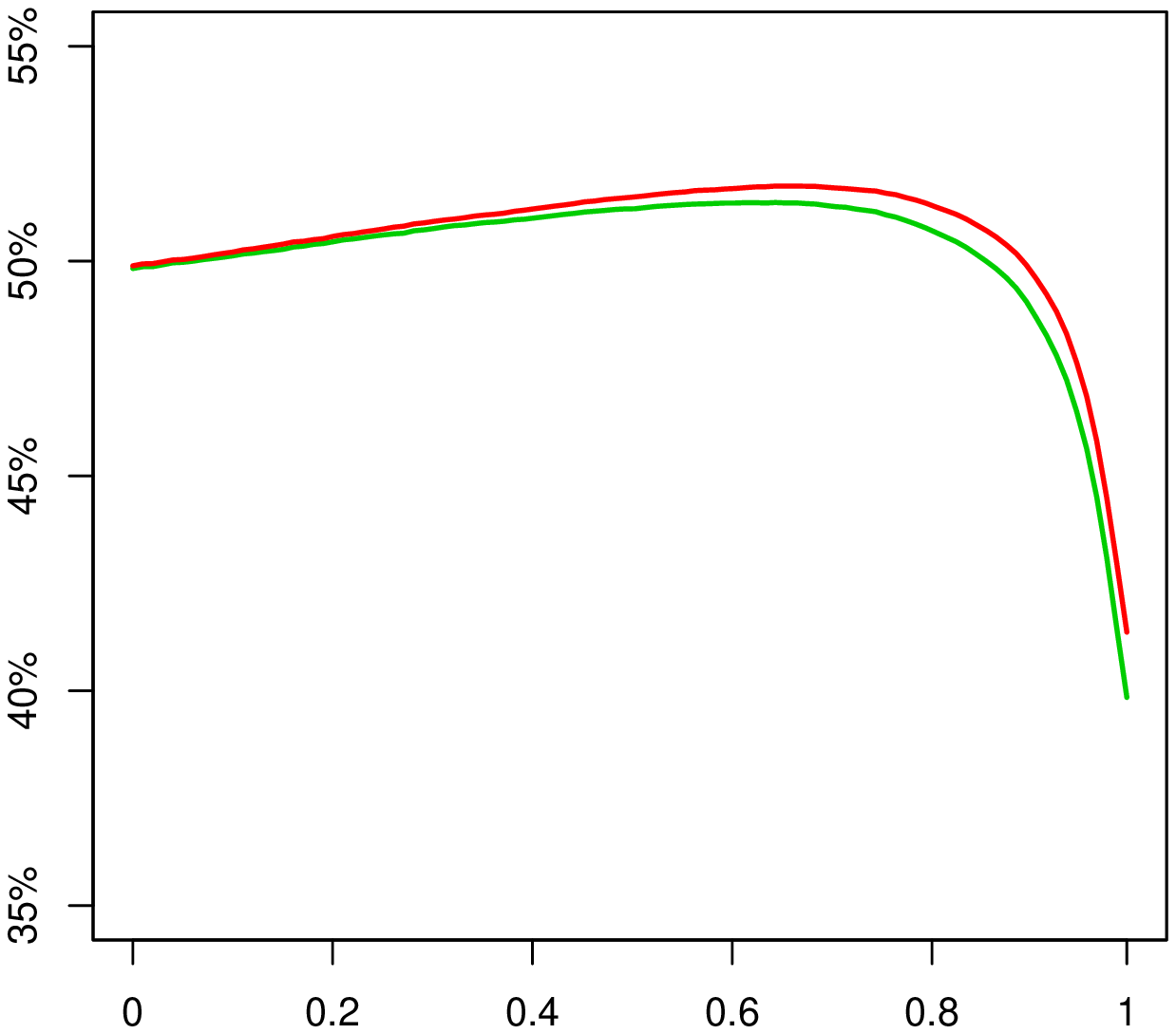}
}
\hspace{1cm}
\subfigure{
\includegraphics[width=0.4\textwidth]{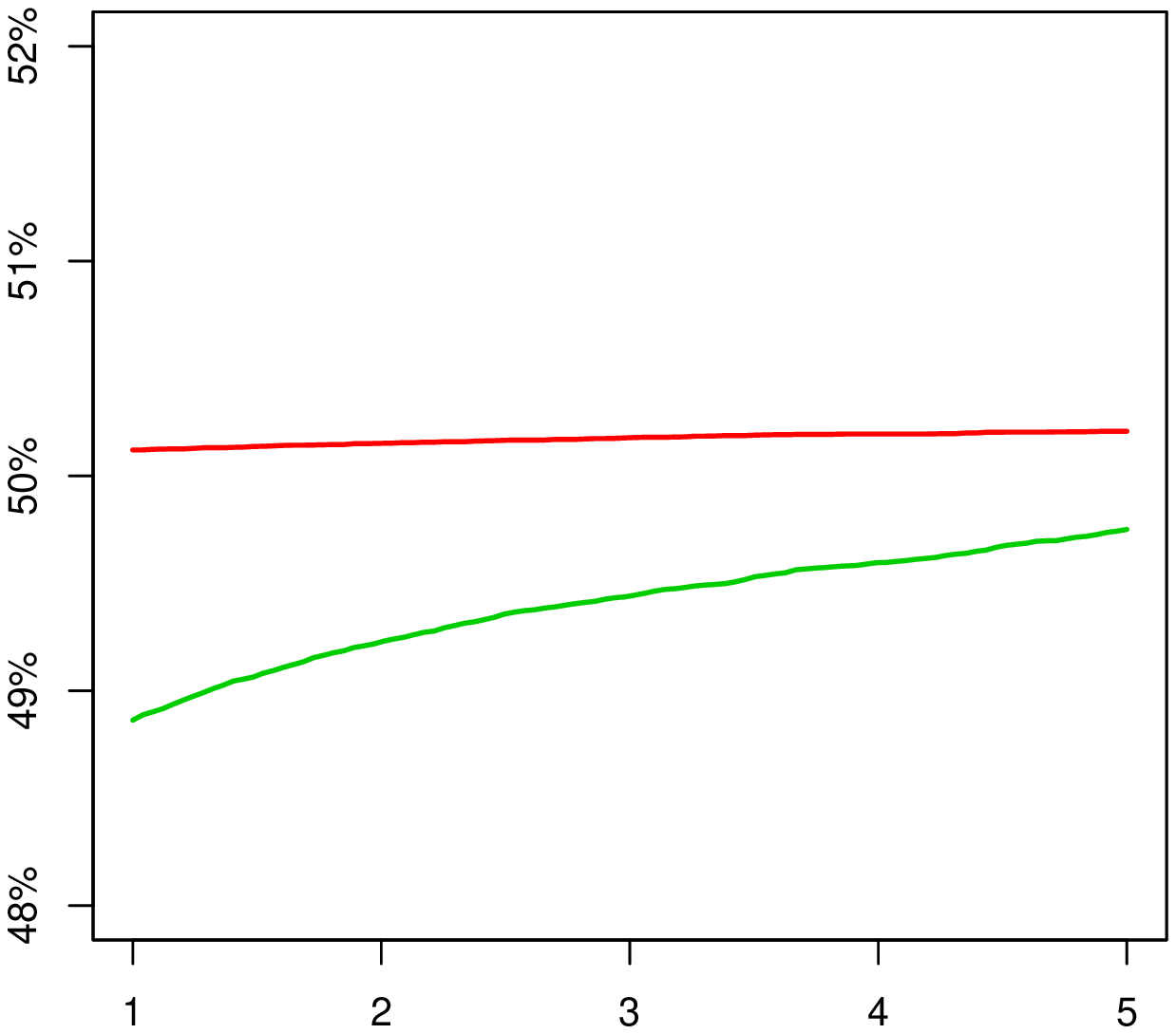}
}
\vspace{-0.5cm}
\caption{The loss probability $\probp(\Delta E_1<0)$ as a function of $\rho$ (left) for $\tau=1$ (green) and $\tau=5$ (red) and as a function of $\tau$ (right) for $\rho=0.1$ (red) and $\rho=0.9$ (green).}
\label{fig: default probabilities}
\end{figure}

\begin{figure}[h]
\vspace{-0.8cm}
\centering
\subfigure{
\includegraphics[width=0.4\textwidth]{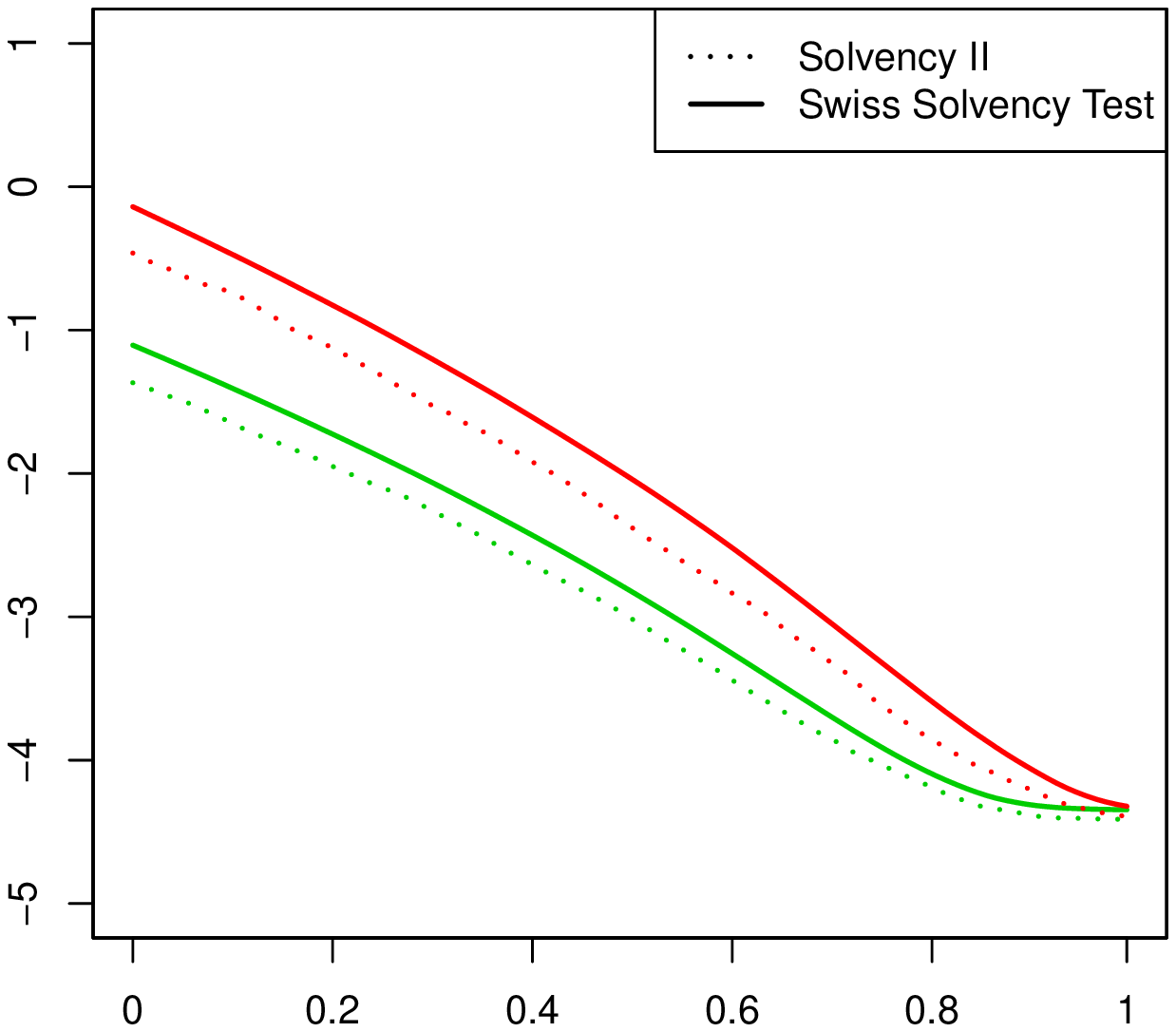}
}
\hspace{1cm}
\subfigure{
\includegraphics[width=0.4\textwidth]{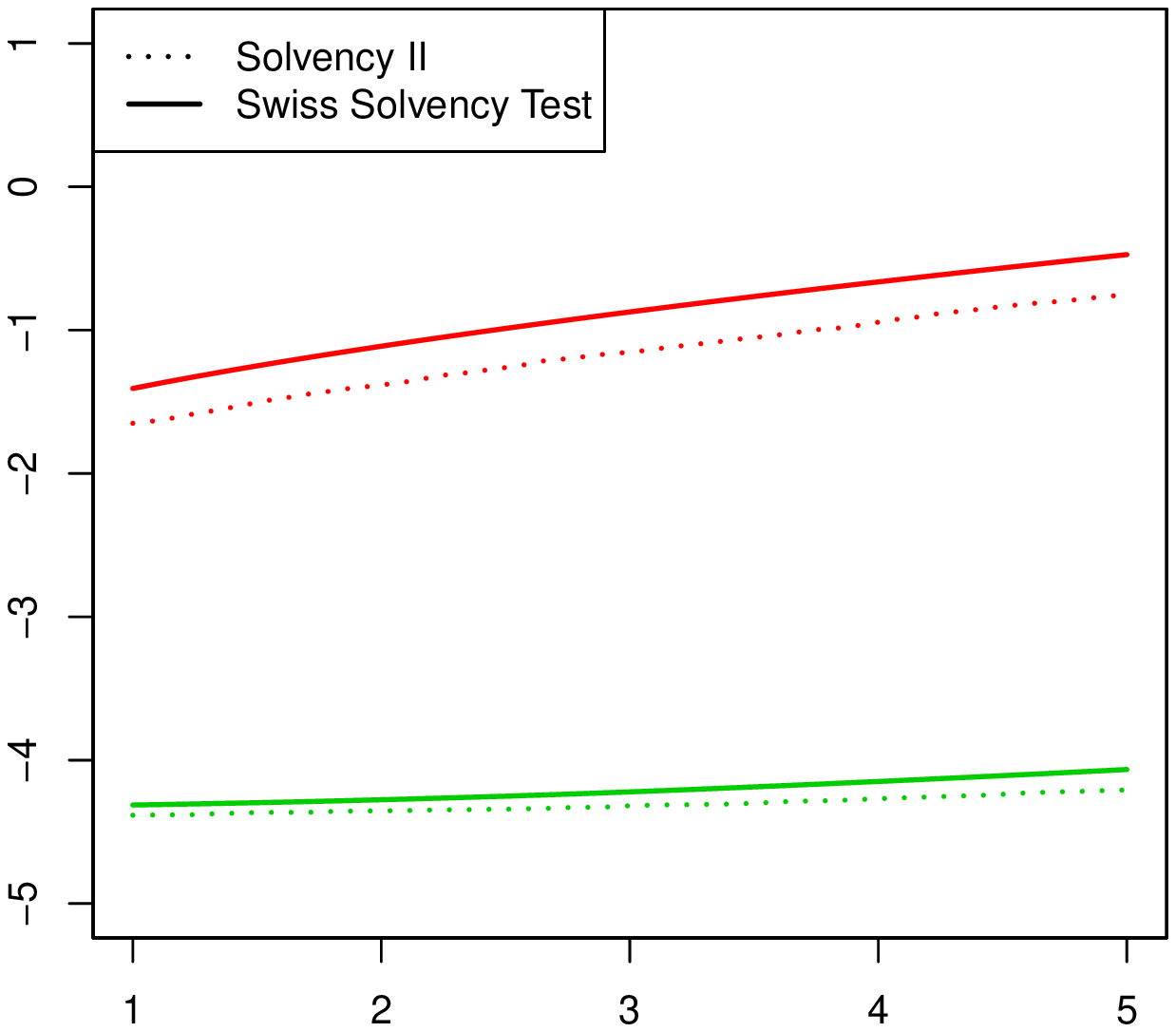}
}
\vspace{-0.5cm}
\caption{The regulatory risk measure $\rho_{reg}(E_1)$ as a function of $\rho$ (left) for $\tau=1$ (green) and $\tau=5$ (red) and as a function of $\tau$ (right) for $\rho=0.1$ (red) and $\rho=0.9$ (green).}
\label{fig: regulatory risk measures}
\end{figure}

\begin{figure}[h!]
\vspace{-0.8cm}
\centering
\subfigure{
\includegraphics[width=0.4\textwidth]{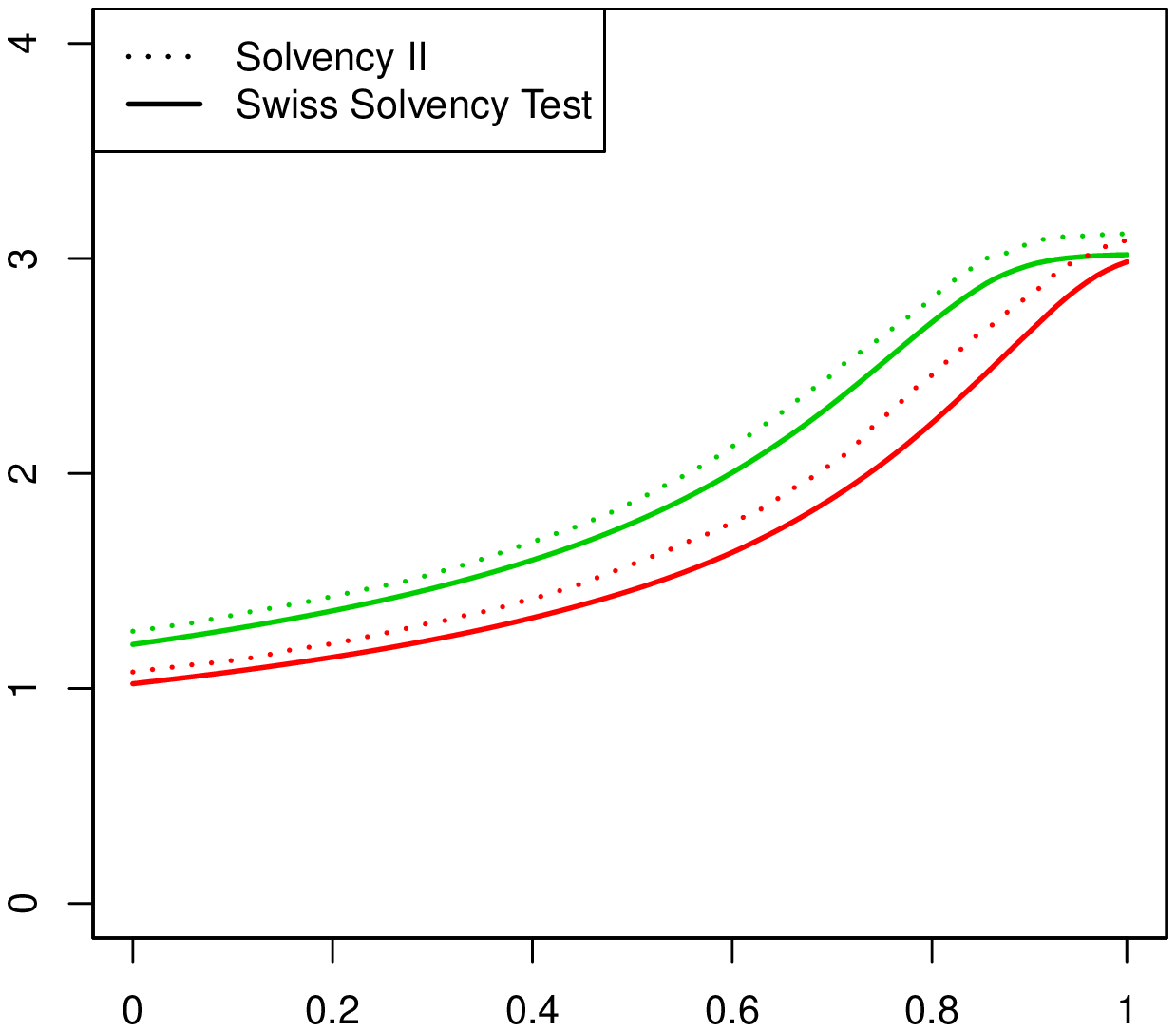}
}
\hspace{1cm}
\subfigure{
\includegraphics[width=0.4\textwidth]{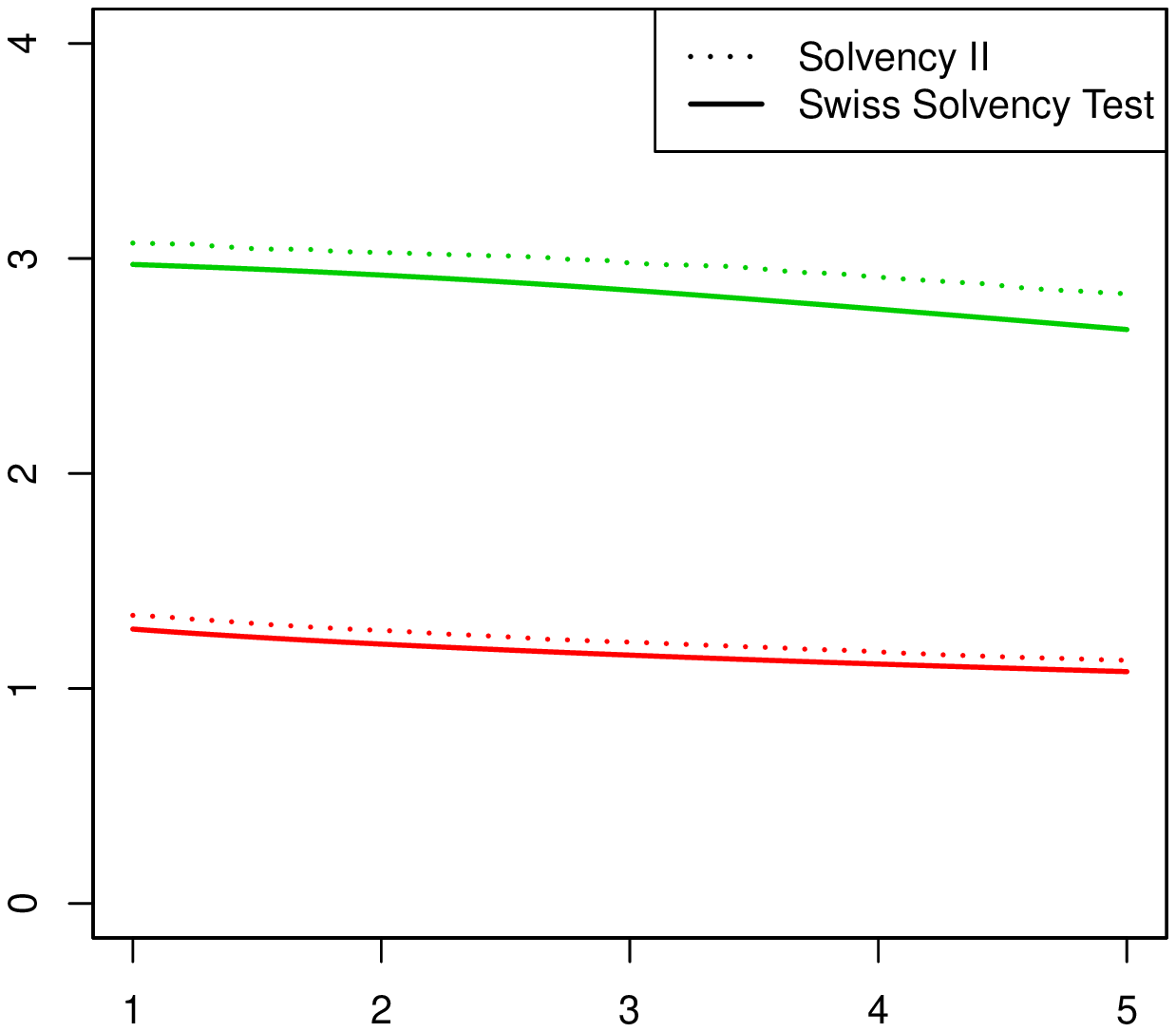}
}
\vspace{-0.5cm}
\caption{The solvency ratio $\frac{E_0}{\rho_{reg}(\Delta E_1)}$ as a function of $\rho$ (left) for $\tau=1$ (green) and $\tau=5$ (red) and as a function of $\tau$ (right) for $\rho=0.1$ (red) and $\rho=0.9$ (green).}
\label{fig: solvency ratio}
\end{figure}

\newpage

\begin{figure}[t]
\vspace{-0.8cm}
\centering
\subfigure{
\includegraphics[width=0.4\textwidth]{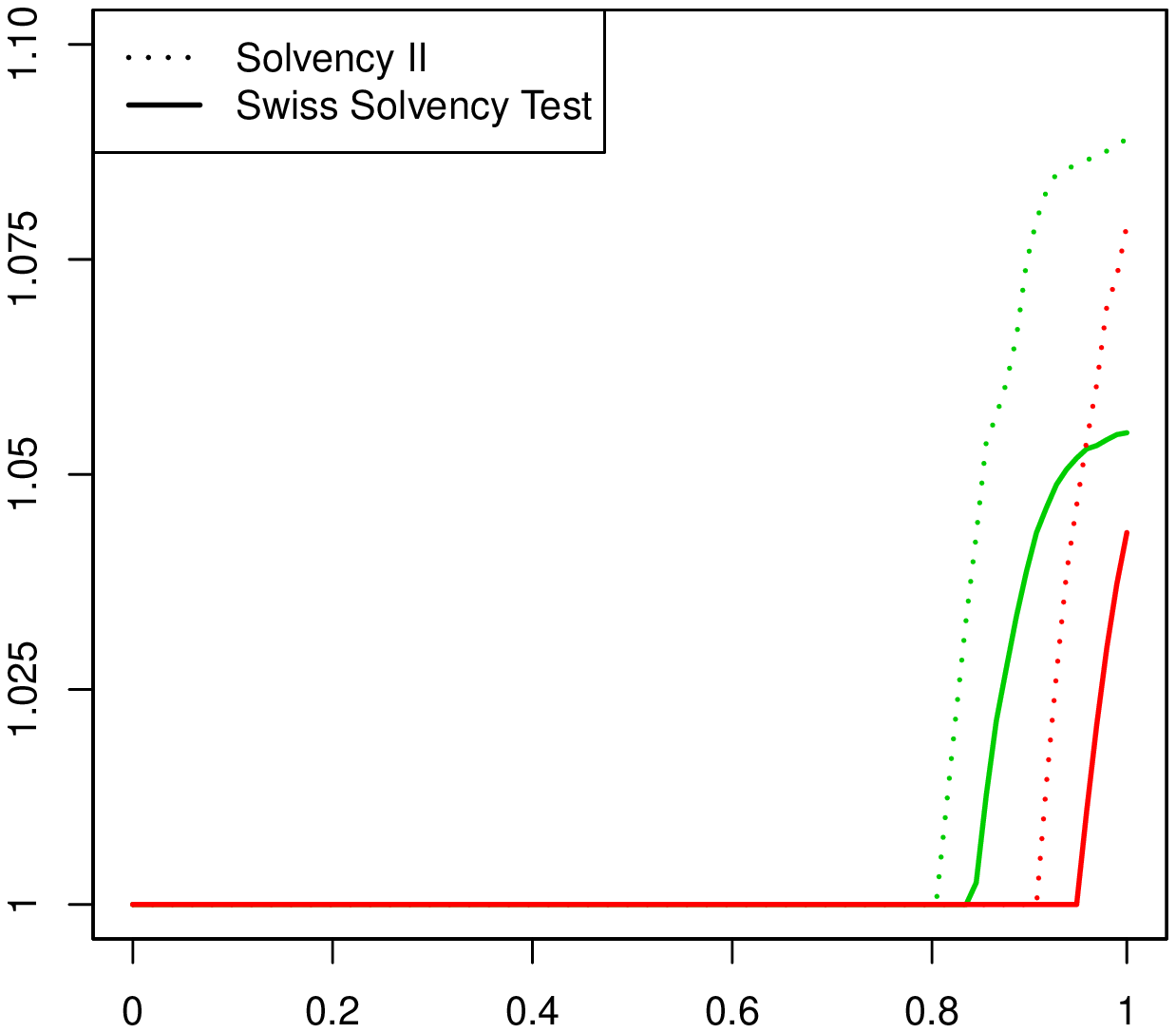}
}
\hspace{1cm}
\subfigure{
\includegraphics[width=0.4\textwidth]{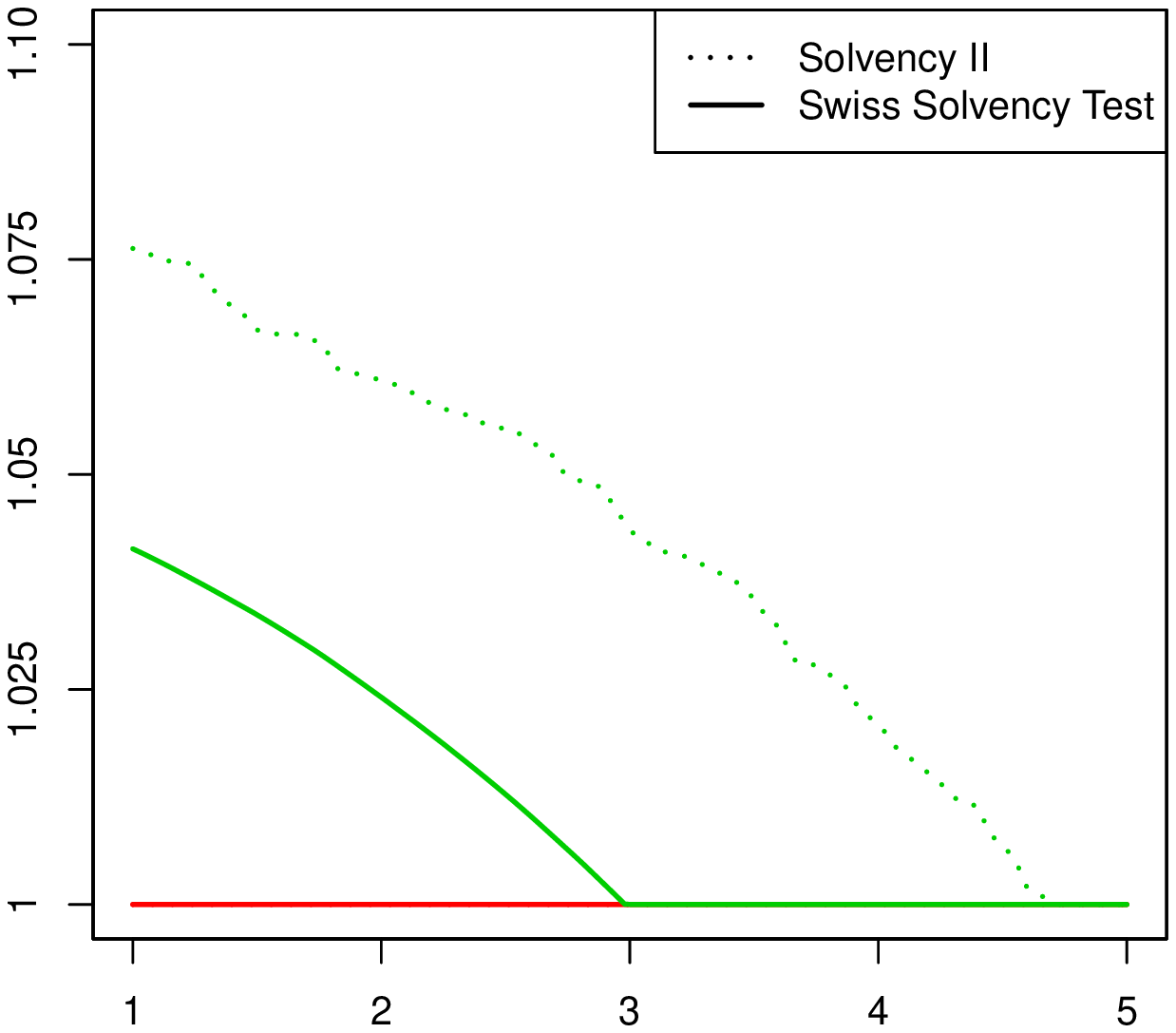}
}
\vspace{-0.5cm}
\caption{The recovery adjustment $\RecAdj(\beta,r)$ for $\beta=\beta_{max}=0.25\%$ and $r=r_{min}=50\%$ as a function of $\rho$ (left) for $\tau=1$ (green) and $\tau=5$ (red) and as a function of $\tau$ (right) for $\rho=0.1$ (red) and $\rho=0.9$ (green).}
\end{figure}

\begin{figure}[h]
\vspace{-0.8cm}
\centering
\subfigure{
\includegraphics[width=0.4\textwidth]{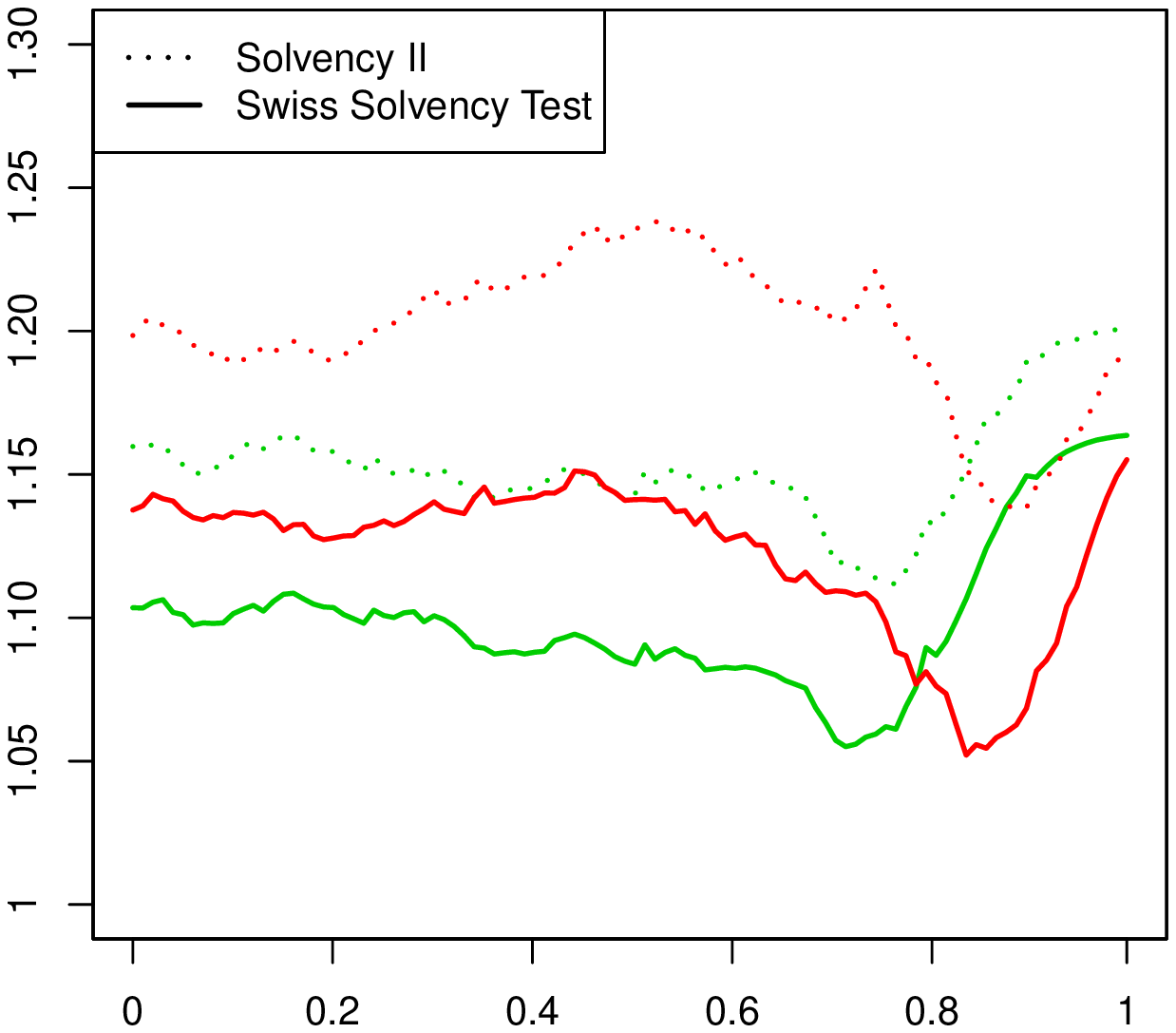}
}
\hspace{1cm}
\subfigure{
\includegraphics[width=0.4\textwidth]{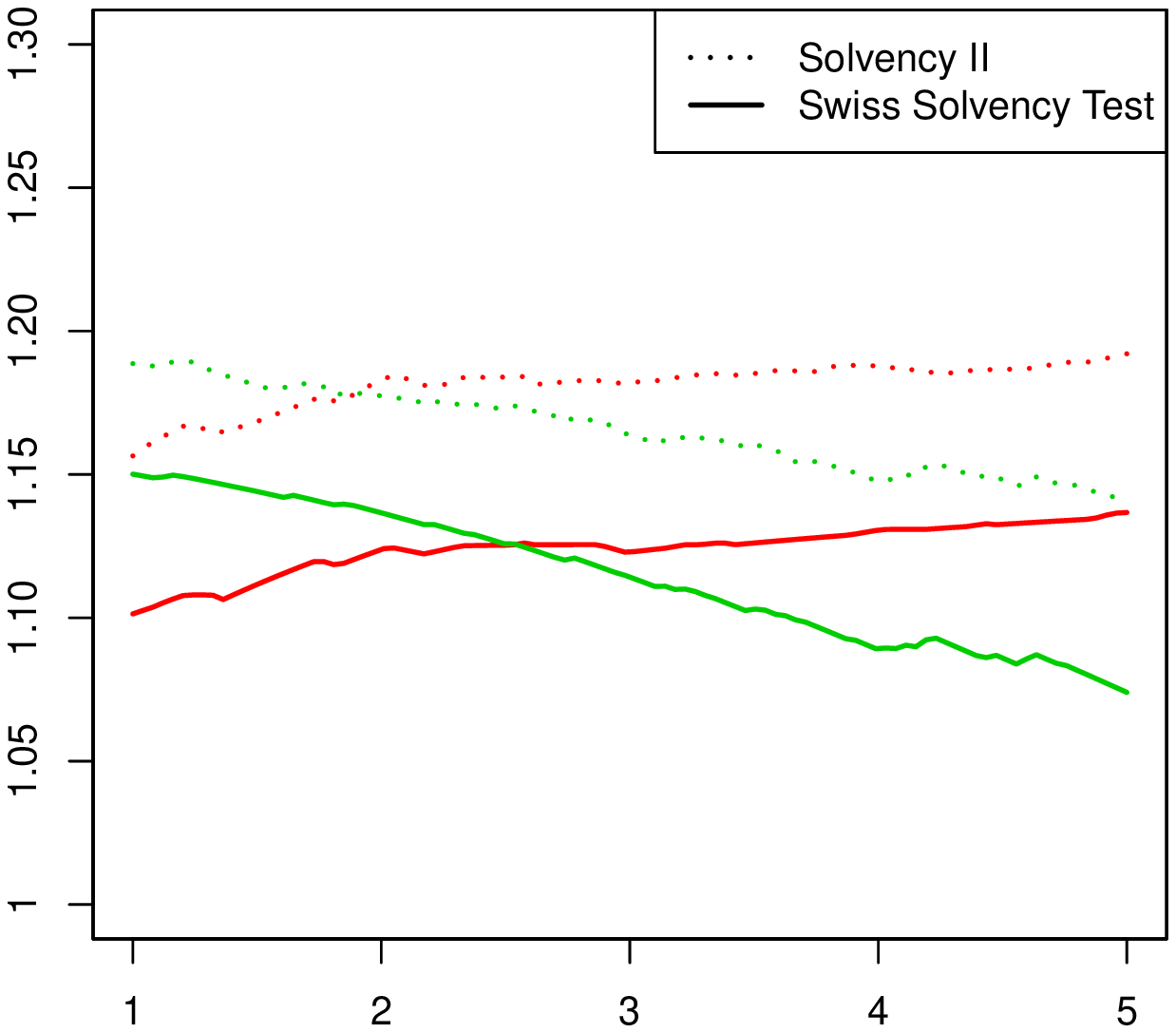}
}
\vspace{-0.5cm}
\caption{The recovery adjustment $\RecAdj(\beta,r)$ for $\beta=\beta_{min}=0.1\%$ and $r=r_{max}=90\%$ as a function of $\rho$ (left) for $\tau=1$ (green) and $\tau=5$ (red) and as a function of $\tau$ (right) for $\rho=0.1$ (red) and $\rho=0.9$ (green).}
\end{figure}


\vfill
\pagebreak

\bibliography{bibtex}
\bibliographystyle{jmr}

\end{document}